\pdfoutput=1
\def\stoc{1}
\def\inappendix{1}
\def\notes{1}   % 1 for notes and 0 to suppress

\documentclass[11pt]{article}

\usepackage{labelschanged}
\usepackage[utf8]{inputenc}
\usepackage{style}

\newcommand{\Arikan}{Ar\i kan}

\newcommand{\cC}{\mathcal{C}}

\newcommand{\Mat}[1]{M^{(#1)}}
\newcommand{\df}{\bar{\delta}}

\newcommand{\Dec}{\mathrm{Dec}}
\renewcommand{\H}{\bar{H}}
\newcommand{\usef}{\sqsubset_u}
\newcommand{\dom}{\preceq}
\newcommand{\HadMat}{\left(\begin{smallmatrix} 1 & 0 \\ \alpha & 1 \end{smallmatrix}\right)}
\newcommand{\wt}{\mathop\mathrm{wt}}

\usepackage{algorithm}
\usepackage[noend]{algpseudocode}
\algrenewcommand\algorithmicrequire{\textbf{Input:}}
\algrenewcommand\algorithmicensure{\textbf{Output:}}
\algrenewcommand\algorithmicwhile{\textbf{While}}
\algrenewcommand\algorithmicfor{\textbf{For}}
\algrenewcommand\algorithmicreturn{\textbf{Return}}
\algrenewcommand\algorithmicif{\textbf{If}}

\algnewcommand\algorithmicconst{\textbf{Constants:}}
\algnewcommand\Const{\item[\algorithmicconst]}

\renewcommand{\d}{\,\mathrm{d}}
\newcommand{\two}{}
\ifnum\stoc=0
\else
\fi

\title{\textbf{General Strong Polarization}\thanks{This paper combines results presented in preliminary form at STOC 2018~\cite{BGNRS18-conf} and RANDOM 2018~\cite{BGS18-conf}.}}
\author{%
Jaros\l aw B\l asiok\thanks{ Department of Computer Science, Columbia University, 500 West 120th Street,
New York, NY 10027, USA. Email: {\tt jb4451@columbia.edu}. This work was done when the author was a PhD student at Harvard University.}%, supported by ONR grant N00014-15-1-2388.}
\and Venkatesan Guruswami\thanks{Computer Science Department, Carnegie Mellon University, Pittsburgh, PA 15213, USA. Portions of this work were done during visits by the author to the School of Physical and Mathematical Sciences, Nanyang Technological University, Singapore, and the Center for Mathematical Sciences and Applications, Harvard University. {\tt venkatg@cs.cmu.edu}. Research supported in part by NSF grants CCF-1422045, CCF-1563742, and CCF-1814603, and a Simons Investigator award.}
\and Preetum Nakkiran\thanks{
Halicioğlu Data Science Institute,
University of California San Diego,
10100 Hopkins Dr, La Jolla, CA 92093, USA.
Email: {\tt preetum@ucsd.edu}.
Work supported in part by a Simons Investigator Award, NSF Awards CCF 1565641
and CCF 1715187, the NSF Graduate Research Fellowship Grant No. DGE1144152, a Google PhD Fellowship, and the NSF/Simons Collaboration on the Theoretical Foundations of Deep Learning.} $\qquad$ 
\and Atri Rudra\thanks{Computer Science and Engineering Department, University at Buffalo. {\tt atri@buffalo.edu}. Research supported in part by NSF grant CCF-1717134.}
\and Madhu Sudan\thanks{Harvard John A. Paulson School of Engineering and Applied Sciences, Harvard University, 33 Oxford Street,
Cambridge, MA 02138, USA. Email: {\tt madhu@cs.harvard.edu}. Work supported in part by a Simons Investigator Award and NSF Awards CCF 1565641 and CCF 1715187.}
}
\date{}

\begin{document}

\maketitle

\thispagestyle{empty}

\vspace{-0.5cm}
\begin{abstract}

\Arikan's exciting discovery of polar codes has provided an altogether new way to efficiently achieve Shannon capacity. Given a (constant-sized) invertible matrix $M$, a family of polar codes can be associated with this matrix and its ability to approach capacity follows from the {\em polarization} of an associated $[0,1]$-bounded martingale, namely its convergence in the limit to either $0$ or $1$ with probability $1$. \Arikan\ showed appropriate polarization of the martingale associated with the matrix $G_2 = \left( \begin{smallmatrix} 1 & 0 \\ 1 & 1 \end{smallmatrix} \right)$ to get capacity achieving codes. His analysis was later extended to all matrices $M$ that satisfy an obvious necessary condition for polarization.

\smallskip
While \Arikan's theorem does not guarantee that the codes achieve capacity at small blocklengths (specifically in length which is a polynomial in $1/\epsilon$ where $\epsilon$ is the difference between the capacity of a channel and the rate of the code), it turns out that a ``strong'' analysis of the polarization of the underlying martingale would lead to such constructions. Indeed for the martingale associated with $G_2$ such a strong polarization was shown in two independent works ([Guruswami and Xia, IEEE IT '15] and [Hassani et al., IEEE IT '14]), thereby resolving a major theoretical challenge associated with the efficient attainment of Shannon capacity.

\smallskip
In this work we extend the result above to cover martingales associated with all matrices that satisfy the necessary condition for (weak) polarization.
In addition to being vastly more general, our proofs of strong polarization are (in our view) also much simpler and modular. Key to our proof is a notion of {\em local polarization} that only depends on the evolution of the martingale in a single time step. We show that local polarization always implies strong polarization. We then apply relatively simple reasoning about conditional entropies to prove local polarization in very general settings.
Specifically, our result shows strong polarization over all prime fields and leads to efficient capacity-achieving source codes for compressing arbitrary i.i.d. sources, and capacity-achieving channel codes for arbitrary symmetric memoryless channels. We show how to use our analyses to achieve exponentially small error probabilities at lengths inverse polynomial in the gap to capacity. Indeed we show that we can essentially match any error probability while maintaining lengths that are only inverse polynomial in the gap to capacity.

\end{abstract}

\newpage
{\footnotesize 
\tableofcontents}
\thispagestyle{empty}

\thispagestyle{empty}

\newpage

\def\inappendix{1}

\parskip=0.5ex

\section{Introduction}

Polar codes, proposed in \Arikan's remarkable
work~\cite{arikan-polar}, gave a fresh information-theoretic approach
to construct linear codes that achieve the Shannon capacity of
symmetric channels, together with efficient encoding and decoding
algorithms. About a decade after their discovery, there is now a vast
and extensive body of work on polar coding spanning hundreds of
papers.  The underlying concept of polarizing transforms has emerged
as a versatile tool to successfully attack a diverse collection of
information-theoretic problems beyond the original channel and source
coding applications, including wiretap channels~\cite{MV11}, the
Slepian-Wolf, Wyner-Ziv, and Gelfand-Pinsker problems~\cite{korada},
broadcast channels~\cite{GAG13}, multiple access
channels~\cite{STY13,AT12}, and interference networks~\cite{WS14}. We
recommend the survey by \c{S}a\c{s}o\u{g}lu~\cite{sasoglu-book} for a
nice treatment of the early work on polar codes. On the practical
side, polar codes show impressive coding gains when a list
decoding variant of the decoder is applied~\cite{TV15}, and have been adopted for the enhanced mobile broadband control channels for the 5G NR (New Radio) interface.

\Arikan's original analysis was asymptotic and  established that capacity can be achieved in the limit of large block lengths but did not quantify the speed of convergence to capacity. Effective finite-length convergence bounds were provided several years later in \cite{GX15,HAU14,GV15} establishing that the polar coding approach leads to a family of codes of rate $C-\epsilon$ for transmission over a channel of (Shannon) capacity $C$, where the block length of the code and the decoding time grow only polynomially in $1/\epsilon$.
In contrast, for all previous constructions of codes, the decoding
algorithms required time exponential in $1/\epsilon$.  Getting a
polynomial running time in $1/\epsilon$ was one of the central
theoretical challenges in the field of algorithmic coding theory, and
polar codes were the first to overcome this challenge. Follow-up works
have also investigated concrete bounds on the \emph{scaling exponent}
$\mu$, i.e., the finite exponent $\mu$ for which the block length of
the code can be bounded by $(1/\epsilon)^\mu$~\cite{MHU16,GoldinB14}, culminating in recent works which achieved $\mu \to 2$ which is the optimal value, first for the erasure channel~\cite{PfisterU19,FHMV21}, and later for all channels~\cite{GRY20,WD21} using variants of polar codes.

The analyses of polar codes turn into questions about {\em polarizations} of certain {\em martingales} (which we refer to as Ar\i kan martingales in this work). The vast class of polar codes alluded to in the previous paragraph all build on polarizing martingales, and the results of  \cite{GX15,HAU14,GV15} show that for one of the families of polar codes, the underlying martingale polarizes ``extremely fast''---a notion we refer to as {\em strong polarization} and will define shortly.

The primary goal of this work is to understand the process of
polarization of martingales, and in particular to understand when a
martingale polarizes strongly. In attempting to study this question,
we come up with a local notion of polarization and show that this
local notion is sufficient to imply strong polarization. Applying this
improved understanding to the martingales arising in the study of
polar codes we show that a simple necessary condition for weak
polarization of such martingales is actually sufficient for strong
polarization. This allows us to extend the previous results on strong
polarization, which only applied to a specific class of codes, to a
broad class of codes and show essentially that all polarizing codes
lead to polynomial convergence to capacity. We further show that this can be achieved while maintaining the same exponentially falling error probability achieved in the original asymptotic analyses that did not give any quantitative bounds on the convergence to capacity. Below we formally describe
the notion of polarization of martingales and our results concerning them, along with their implications for quantitatively strong convergence to capacity of polar codes when applied to the associated Ar\i kan martingales.
Figure~\ref{fig:overview} gives a detailed roadmap of this paper with different columns indicating different categories of results and each column describing a hierarchy of results.

%\documentclass[11pt]{article}

%\usepackage[dvipsnames]{xcolor}
%\usepackage{tikz,tikz-qtree,tikz-qtree-compat,tikz-3dplot}
%\usetikzlibrary{arrows.meta}
%\usetikzlibrary{patterns, snakes}

%\usepackage{amsmath,amssymb}
%\usepackage{amsfonts}

%\begin{document}

\begin{figure}
\begin{center}
\begin{tikzpicture}[scale=1.25]

%%% Macros

\def\arrowwidth{6} %defines the width of the arrows
\def\arrowlen{11} % Define the length of the arrowhead

%%%%%% Codes

%Just the boxes
\draw [fill=Cyan, draw=gray, opacity = 0.2] (-1,0) rectangle (2,11);
\draw [fill=Cyan, draw=gray, opacity = 0.2] (-0.7,1) rectangle (1.7,9.5);
\draw [fill=blue, draw=gray, opacity = 0.2] (-0.5,2) rectangle (1.5,7.5);
\draw [fill=blue, draw=gray, opacity = 0.2] (-0.3,2.5) rectangle (1.3,5.5);

%The text
\node [right, text width = 3cm] at (-1, 10.5) {\small{Capacity achieving  codes}};
\node [right, text width = 3cm] at (-0.7, 9) {\small{$\exists n$, $n^{-c}$ decoding error}};
\node [right, text width = 2.5cm] at (-0.5, 6.9) {\small{$n=\mathrm{poly}\left(1/\epsilon\right)$, $n^{-c}$ decoding error}};
\node [right, text width = 2cm] at (-0.3, 4.5) {\small{$n=\mathrm{poly}\left(1/\epsilon\right)$, $\exp\left(n^{-\beta}\right)$ decoding error}};
%%%%%%%%%%%%%%%%%%%%%%%%%%%

%%%%%% Polarization

%Just the boxes
\draw [fill=green, draw=gray, opacity = 0.2] (3,0) rectangle (6,11);
\draw [fill=green, draw=gray, opacity = 0.2] (3.3,1) rectangle (5.7,9.5);
\draw [fill=Green, draw=gray, opacity = 0.2] (3.5,2) rectangle (5.5,7.5);
\draw [fill=Green, draw=gray, opacity = 0.2] (3.7,2.5) rectangle (5.3,5.5);

% The arrows
\hypersetup{linkcolor=white};
%Thm 1.18 arrow
\draw [-{Latex[length=\arrowlen mm]}, line width=\arrowwidth mm, color=Green] (3.7,4) -- (1.1, 4);
\node[right] at (2,4) {\textcolor{white}{\footnotesize{Thm. \ref{thm:thm1}}}};
%Thm 1.17 arrow
\draw [-{Latex[length=\arrowlen mm]}, line width=\arrowwidth mm, color=Green] (3.5,6.5) -- (1.3, 6.5);
\node[right] at (2,6.5) {\textcolor{white}{\footnotesize{Thm. \ref{thm:combo}}}};
%Thm 1.10 arrow
\draw [-{Latex[length=\arrowlen mm]}, line width=\arrowwidth mm, color=Green] (3.3,8.5) -- (1.5, 8.5);
\node[right] at (1.9,8.5) {\textcolor{white}{\footnotesize{Thm. \ref{thm:exp-code}}}};
\hypersetup{linkcolor=red};

%The text
\node [right, text width = 3cm] at (3, 10.5) {\small{Weak polarization} \textcolor{gray}{\footnotesize{(Def. \ref{def:weakpolarization})}}};
\node [right, text width = 3cm] at (3.3, 9) {\small{Regular polarization} \textcolor{gray}{\footnotesize{(Def. \ref{def:regpolarization})}} };
\node [right, text width = 2.5cm] at (3.5, 6.9) {\small{Strong polarization} \textcolor{gray}{\footnotesize{(Def. \ref{def:strongpolarization})}}  };
\node [right, text width = 2cm] at (3.7, 4.5) {\small{Exp. strong polarization} \textcolor{gray}{\footnotesize{(Def. \ref{def:exp-strong-polarization})}}  };
%%%%%%%%%%%%%%%%%%%%%%%%%%%

%%%%%%% Local polarization 

%The boxes
\draw [fill=orange, draw=gray, opacity = 0.2] (7,2) rectangle (9,7.5);
\draw [fill=orange, draw=gray, opacity = 0.2] (7.2,2.5) rectangle (8.8,5.5);

% The arrows
\hypersetup{linkcolor=white};
%Thm 1.9 arrow
\draw [-{Latex[length=\arrowlen mm]}, line width=\arrowwidth mm, color=orange] (7.2,4) -- (5.1, 4);
\node[right] at (5.7,4) {\textcolor{white}{\footnotesize{Thm. \ref{thm:local-to-global-exp}}}};
%Thm 1.7 arrow
\draw [-{Latex[length=\arrowlen mm]}, line width=\arrowwidth mm, color=orange] (7,6.5) -- (5.3, 6.5);
\node[right] at (5.7,6.5) {\textcolor{white}{\footnotesize{Thm. \ref{thm:local-global}}}};
\hypersetup{linkcolor=red};

%The text
\node [right, text width = 2.5cm] at (7, 6.9) {\small{Local polarization} \textcolor{gray}{\footnotesize{(Def. \ref{defn:polar-local})}} };
\node [right, text width = 2cm] at (7.2, 4.5) {\small{Exp. local polarization} \textcolor{gray}{\footnotesize{(Def. \ref{defn:polar-local-exp})}} };

%%%%%%%%%%%%%%%%%%%%%%%%%%%

%%%%%%% Matrix polarization 
\draw [fill=red, draw=gray, opacity = 0.2] (10,2) rectangle (12,7.5);
\draw [fill=red, draw=gray, opacity = 0.2] (10.2,2.5) rectangle (11.8,5.5);

% The arrows
\hypersetup{linkcolor=white};
%Thm 4.4 arrow
\draw [-{Latex[length=\arrowlen mm]}, line width=\arrowwidth mm, color=red] (10.2,4) -- (8.5, 4);
\node[right] at (8.75,4) {\textcolor{white}{\footnotesize{Thm. \ref{lem:polarizing-matrix-implies-martingale}}}};
%Thm 4.4 arrow
\draw [-{Latex[length=\arrowlen mm]}, line width=\arrowwidth mm, color=red] (10,6.5) -- (8.5, 6.5);
\node[right] at (8.75,6.5) {\textcolor{white}{\footnotesize{Thm. \ref{lem:polarizing-matrix-implies-martingale}}}};
\hypersetup{linkcolor=red};
%The text
\node [right, text width = 2.5cm] at (10, 6.9) {\small{Matrix polarization} \textcolor{gray}{\footnotesize{(Def. \ref{def:matrix-polarization})}} };
\node [right, text width = 2cm] at (10.2, 4.5) {\small{Exp. matrix polarization} \textcolor{gray}{\footnotesize{(Def. \ref{def:matrix-polarization})}}  };

%%%%%%%%%%%%%%%%%%%%%%%%%%%

%Mizing matrix

\draw [fill=Plum, draw=gray, opacity=0.4] (8.5, 9) rectangle (10.5, 10);
\draw [fill=Plum, draw=gray, opacity=0.4] (8.5, 0) rectangle (10.5, 1);

% The arrows
\hypersetup{linkcolor=white};
%Thm 1.15 arrow
\draw [-{Latex[length=\arrowlen mm]}, line width=\arrowwidth mm, color=Magenta] (8.5, 9.5) -- (8, 9.5) -- (8,7.5);
\node at (8,9) {\textcolor{white}{\rotatebox{-90}{\footnotesize{Thm. \ref{thm:triangle-local}}}}};
%Lemma 5.5 arrow
\draw [-{Latex[length=\arrowlen mm]}, line width=\arrowwidth mm, color=Plum] (10.5, 9.5) -- (11, 9.5) -- (11,7.5);
\node at (11,9) {\textcolor{white}{\rotatebox{-90}{\footnotesize{Lem. \ref{lem:k-by-k-polarize}}}} };
%Thm 1.16 arrow
\draw [-{Latex[length=\arrowlen mm]}, line width=\arrowwidth mm, color=Magenta] (8.5, 0.5) -- (8, 0.5) -- (8,2.5);
\node at (8,1) {\textcolor{white}{\rotatebox{90}{\footnotesize{Thm. \ref{thm:triangle-local-exp}}}} };
%Lemma 7.1 arrow
\draw [-{Latex[length=\arrowlen mm]}, line width=\arrowwidth mm, color=Plum] (10.5, 0.5) -- (11, 0.5) -- (11,2.5);
\node at (11,1) {\textcolor{white}{\rotatebox{90}{\footnotesize{Lem. \ref{lem:every-matrix-works}}}} };
\hypersetup{linkcolor=red};

%The text
\node [right, text width = 2.5cm] at (8.5, 9.5) {\small{Mixing matrix $M$} \textcolor{gray}{\footnotesize{(Def. \ref{def:mixing-matrix})}}   };
\node [right, text width = 2.5cm] at (8.5, 0.5) {\small{$M^{\otimes 2}$, mixing matrix $M$}};

%%%%%%%%%%%%%%

%Exp mising matrix

% The arrows

%The text

%%%%%%%%%%%%%%%%%%

\end{tikzpicture}
\end{center}
\caption{Overview of our results (excluding those in \cref{sec:extra-results}). The blue boxes (on the extreme left) represent the various coding results ($n$ is the code block length, $c$ and $\beta<1$ are absolute constants). The green boxes (middle left) are the various notations of polarizations that we study in the paper. The orange boxes (middle right) are the two notions of local polarization and the red boxes (extreme right) are the two notions of matrix polarizations we use. Purple boxes (top and bottom on right) show the notions of mixing matrices that we use. All the arrows denote the various results we prove (except for \cref{thm:exp-code}, which is implicit in \Arikan~\cite{arikan-polar}) in this paper.}
\label{fig:overview}
\end{figure}

%\end{document}

\subsection{Polarization of $[0,1]$-martingales}

Our interest is mainly in the (rate of) polarization of a specific family of martingales that we call the \Arikan\ martingales. We will define these objects later, but first describe the notion of polarization for general $[0,1]$-bounded martingales. The middle left (green) column in \cref{fig:overview} shows the various notions of polarization that we define in this section.

Recall that a sequence of random variables $X_0,\ldots,X_t,\ldots$ is said to be a {\em martingale} if for every $t$ and $a_0,\ldots,a_t$ it is the case that $\E[X_{t+1}| X_0 = a_0,\ldots,X_t = a_t] = a_t$. We say that that a martingale is {\em $[0,1]$-bounded} (or simply a $[0,1]$-martingale) if $X_t \in [0,1]$ for all $t \geq 0$.

\cstate{Weak Polarization}{definition}{defweakpolarization}{
    \label{def:weakpolarization}
	A $[0,1]$-martingale sequence $X_0,X_1,\ldots,X_t,\ldots$ is defined to be  {\em weakly polarizing} if $\lim_{t\to\infty}\{X_t\}$ exists with probability $1$, and this limit is either $0$ or $1$.
}
Note that  the limit of the martingale sequence $X_0,X_1,\ldots,X_t,\ldots$ is a Bernoulli random variable with expectation $X_0$.\footnote{The claim on expectation follows since by definition, $\E\brackets{X_{t+1}}=\E\brackets{X_t}$.}

Thus a polarizing martingale does not converge to a single value with probability $1$, but rather converges to one of its extreme values. For the applications to constructions of polar codes, we need more explicit bounds on the rates of convergence leading to the notions of (regular) polarization and strong polarization defined below in \cref{def:regpolarization,def:strongpolarization} respectively.

\cstate{$(\tau_\ell, \tau_h, \varepsilon)$-Polarization}{definition}{deftauepspolarization}{
	For functions $\tau_\ell, \tau_h,\varepsilon:\Z^+ \to \R^{\geq 0}$, a $[0, 1]$-martingale sequence $X_0, X_1, \ldots X_t, \ldots$ is defined to be {\em $(\tau_\ell, \tau_h, \varepsilon)$-polarizing} if for all $t$ we have
	\begin{equation*}
		\P(X_t \in (\tau_\ell(t), 1-\tau_h(t))) < \varepsilon(t).
	\end{equation*}
}

\cstate{Regular Polarization}{definition}{defregpolarization}{
	\label{def:regpolarization}
	A $[0,1]$-martingale sequence $X_0,X_1,\ldots,X_t,\ldots$ is defined to be {\em regular polarizing} if for all constant $\gamma > 0$,  there exist $\varepsilon(t) = o(1)$, such that the martingale $\{X_t\}_{t\ge 0}$ is $(\gamma^t, \gamma^t, \varepsilon(t))$-polarizing.
}

We refer to the above as being ``sub-exponentially'' close to the limit (since it holds for every $\gamma > 0$). While weak polarization by itself is an interesting phenomenon, regular polarization (of \Arikan\ martingales) leads to capacity-achieving codes (though without explicit bounds on the length of the code as a function of the gap to capacity) and thus regular polarization is well-explored in the literature and tight necessary and sufficient conditions are known for regular polarization of \Arikan\ martingales~\cite{arikan-telatar,KSU10}.

To get codes of block length polynomially small in the gap to capacity, an even stronger notion of polarization is needed, where we require that the sub-exponential closeness to the limit happens with \emph{all but exponentially small probability}. We define this formally next.

\cstate{Strong Polarization}{definition}{defstrongpolarization}{
	\label{def:strongpolarization}
	A $[0,1]$-martingale sequence $X_0,X_1,\ldots,X_t,\ldots$ is defined to be {\em strongly polarizing} if for all $\gamma > 0$ there exist $0<\eta < 1$ and $\beta < \infty$ such that the martingale $\{X_t\}_{t\ge 0}$ is $(\gamma^t, \gamma^t, \beta \cdot \eta^t)$-polarizing.
}

Finally to get codes where the decoding error probability is exponentially small in the block length, the codes need to polarize even more strongly. We abstract this notion as follows:

\begin{definition}[Exponentially Strong Polarization]
	\label{def:exp-strong-polarization}
	We say that $X_t$ has $\Lambda$-exponentially strong polarization if for every $0 < \gamma < 1$ there exist constants $0< \eta < 1$ and $\beta < \infty$ such that the martingale $\{X_t\}_{t \geq 0}$ is $(2^{-2^{\Lambda t}}, \gamma^t, \beta \eta^t)$-polarizing. 
\end{definition}

Note that this definition is asymmetric with respect to the two boundaries and expects tighter polarization when $X_t \to 0$ than when $X_t \to 1$. The reasons for this un-aesthetic choice are the following: (1) For the strong decoding results, the tighter polarization when $X_t \to 0$ suffices. (2) Several of the martingales we consider do not achieve sufficiently tight polarization when $X_t \to 1$ (the $\Lambda$ they achieve as $X_t \to 1$ is much smaller than what is needed in the decoding results). (3) The analysis of the best polarizations when $X_t \to 0$ is completely different than the analysis when  $X_t \to 1$. Due to these reasons we work with this asymmetric definition of exponentially strong polarization.

In contrast to the rich literature on regular polarization, results on strong polarization and exponentially strong polarization are quite rare, reflecting a general lack of understanding of this phenomenon. Indeed, while (roughly) an \Arikan\ martingale can be associated with every invertible matrix over any finite field $\F_q$, the only concrete matrix for which exponentially strong polarization was known prior to this work was for  $G_2 = \left( \begin{smallmatrix}
	1 & 0  \\
	1 & 1
\end{smallmatrix} \right)$ \cite{GX15,HAU14,GV15}.
\footnote{An exception is the work by Pfister and Urbanke~\cite{PfisterU19} who showed that for the \emph{$q$-ary erasure channel} for large enough $q$, the martingale associated with a $q \times q$ Reed-Solomon based matrix proposed in \cite{mori-tanaka} polarizes strongly, and the resulting polar codes achieve scaling exponent tending to $2$.}

Part of the reason behind the lack of understanding of strong polarization is that polarization is a ``limiting phenomenon'' in that one tries to understand $\lim_{t\to\infty} X_t$, whereas most stochastic processes, and the \Arikan\ martingales in particular, are defined by local evolution, i.e., one that relates $X_{t+1}$ to $X_t$. The main contribution of this work is to give a local definitions of polarization (\cref{defn:polar-local,defn:polar-local-exp}) and then showing that these definitions imply strong and exponentially strong polarization (\cref{thm:local-global,thm:local-to-global-exp}). Later we show that \Arikan\ martingales polarize locally whenever they satisfy a simple condition that is necessary even for weak polarization. And while the \Arikan\ martingale itself is not locally exponentially polarizing, we show that the ``two-step'' \Arikan\ martingale is exponentially locally polarizing under the same simple condition. (The ``two step'' version of a martingale $X_0,X_1,X_2,\ldots,$ is just the martingale $X_0,X_2,X_4,\ldots$.) As a consequence we get exponentially strong polarization for all \Arikan\ martingales for which previously only regular polarization was known.

\subsection{Results I: Local to strong global polarization of martingales}

Before giving the definition of local polarization, we motivate our definition using some simple examples. Consider the  martingale $Z_0,Z_1,\ldots$ where  $Z_0 = 1/2$, and $Z_{t+1} = Z_t +Y_{t+1} 2^{-(t+2)}$ where $Y_1,\ldots,Y_t,\ldots$ are chosen uniformly and independently from $\{-1,+1\}$. Clearly this sequence is not polarizing (the limit of $Z_t$ is uniform in $[0,1]$). One reason why this happens is that as time progresses, the martingale slows down and stops varying much. We would like to prevent this, but this is also inevitable if a martingale is polarizing and bounded. In particular, a polarizing martingale would be slowed at the boundaries (i.e., when $X_t$ is close to $0$ or close to $1$) and cannot vary much. The first condition in our definition of local polarization insists that this be the only reason a martingale slows down (we refer to this as {\em variance in the middle}).

Next we consider what happens when a martingale is close to the boundary. For this part consider a martingale $Z_0 = 1/2$ and $Z_{t+1} = Z_t + \frac12 Y_{t+1} \min\{Z_t,1-Z_t\}$ where again $Y_1,\ldots,Y_t,\ldots$ are chosen uniformly and independently from $\{-1,+1\}$. This martingale does polarize and even shows regular polarization, but it can also be easily seen that the probability that $Z_t < \frac12 \cdot 2^{-t}$ is zero (whereas we would like probability of being less than say $10^{-t}$ to go to $1$). So this martingale definitely does not show strong polarization. This is so since even in the best case the martingale is approaching the boundary at a fixed exponential rate, and not a sub-exponential one. To overcome this obstacle we require that when the martingale is close to the boundary, with a fixed constant probability it should get much closer in a single step (a notion we refer to as {\em suction at the ends}).

The middle right (orange) column in \cref{fig:overview} shows the notions of local polarization we define in this section (the arrows from the orange column to the middle left (green) columns show the main theorems in this section).

The definition below makes the above requirements precise.

\cstate{Local Polarization}{definition}{defnpolarlocal}{
	\label{defn:polar-local}
	A $[0,1]$-martingale sequence $X_0,\ldots,X_j,\ldots,$ is {\em locally polarizing} if the following conditions hold:
	\begin{enumerate}
		\item {\bf (Variance in the middle):} For every $\tau > 0$, there is a $\theta = \theta(\tau) > 0$ such that for all $j$, we have: If $X_j \in (\tau, 1-\tau)$ then
		$\E [(X_{j+1} - X_j)^2 | X_j] \geq \theta$.
		\item {\bf (Suction at the ends):} There exists an $\alpha > 0$, such that for all  $c < \infty$, there exists a $\tau =\tau(c) > 0, $
		such that:
		\begin{enumerate}
			\item
			If $X_j \leq \tau$ then $\Pr[X_{j+1} \leq X_j/c | X_j]\geq \alpha$.
			\item
			Similarly, if $1 - X_j \leq \tau$ then $\Pr[(1 - X_{j+1} \leq (1 - X_j)/c | X_j] \geq \alpha$.
		\end{enumerate}
		We refer to condition (a) above as {\em Suction at the low end} and condition (b) as {\em Suction at the high end}.
	\end{enumerate}
	When we wish to be more explicit, we refer to the sequence as $(\alpha,\tau(\cdot),\theta(\cdot))$-locally polarizing.
}

As such, it is not clear that this definition is of any use. E.g. it (1) neither obviously implies strong polarization, nor (2) is it obviously satisfiable by any interesting martingale.  In this paper, we address both these issues. First, we establish general theorems connecting local polarization to strong polarization, as described in Theorems~\ref{thm:local-global} and \ref{thm:local-to-global-exp} below. Then, we leverage this to prove quantitatively strong capacity-approaching properties of polar codes, via the strong polarization of \Arikan\ martingales associated with polar codes (Section~\ref{ssec:intro-arikan}). By our local-to-strong conversion, this in turn follows from the local polarization of \Arikan\ martingales which we establish in Theorems~\ref{thm:triangle-local} and \ref{thm:triangle-local-exp}.

\cstate{Local vs. Strong Polarization}{theorem}{thmlocalglobal}{
	\label{thm:local-global}
	If a $[0,1]$-martingale sequence $X_0,\ldots,X_t,\ldots,$ is locally polarizing, then it is also strongly polarizing.
}

If the suction at the ends shows by the martingale is even stronger, then we can get even stronger polarization. The following
definition captures the stronger suction property.

\begin{definition}[Exponential Local Polarization]\label{defn:polar-local-exp}
	We say that $X_t$ has $(\eta, b)$-exponential local polarization 
	if it satisfies local polarization (\cref{defn:polar-local}) and the following additional property
	\begin{enumerate}
		\item {\bf(Strong suction at the low end):} There exists $\tau > 0$ such that if $X_j \leq \tau$ then
		$\Pr[X_{j+1} \leq X_j^{b} | X_j] \geq \eta$.
	\end{enumerate}
\end{definition}

Note that the interesting range for the parameter $b$ is $b > 1$ and that is the range most of our results will focus on.

In the same way that local polarization implies strong global polarization of a martingale, this new stronger local condition implies a stronger global polarization behavior.

\begin{theorem}[Local to Global Exponential Polarization]
	\label{thm:local-to-global-exp}
	Let $\Lambda,b, \eta >0$ be such that $\Lambda < \eta \log_2 b$.  
	Then if a $[0,1]$-bounded martingale $X_0,X_1,X_2,\ldots$ satisfies $(\eta, b)$-exponential local polarization then it also satisfies $\Lambda$-exponentially strong polarization.\footnote{Note that to get $\eta \log_2 b > \Lambda > 0$ we need $\log_2 b > 0$ and so $b > 1$.}
\end{theorem}

\cref{thm:local-global,thm:local-to-global-exp} are proved in \cref{app:loc_to_glob_full}. 
In the rest of this section we turn to showing that 
the notions of local polarization are not vacuous. Indeed, in later sections we show that the \Arikan\ martingales polarize locally (under simple necessary conditions). First we give some background on Polar codes.

\subsection{The \Arikan\ martingale and capacity-achieving polar codes}
\label{ssec:intro-arikan}

The setting of polar codes considers an arbitrary {\em symmetric memoryless
	channel} and yields codes that aim to achieve the {\em capacity} of this
channel. These notions are reviewed in \cref{sec:channel}.
Given any $q$-ary memoryless channel $\C_{Y|Z}$ and invertible matrix $M \in \F_q^{k \times k}$, the theory of polar codes implicitly defines a martingale, which we call the \Arikan\ martingale associated with $(M,\C_{Y|Z})$ and studies its polarization. (An additional contribution of this work is that we give an explicit compact definition of this martingale, see \cref{def:arikan-martingale}. Since we do not need this definition for the purposes of this section, we defer it to \cref{sec:martingale}.) The consequences of regular polarization are described by the following remarkable theorem. (Below we use $M\otimes N$ to denote the tensor product of the matrix $M$ and $N$. Further, we use $M^{\otimes t}$ to denote the tensor of a matrix $M$ with itself $t$ times.)

\cstate{Asymptotic convergence to capacity; Implied by \Arikan~\cite{arikan-polar}}{theorem}{thmarikan}{
	\label{thm:exp-code}
	Let $\C$ be a $q$-ary symmetric memoryless channel and let $M \in \F_q^{k \times k}$ be an
	invertible matrix.
	If the \Arikan\ martingale associated with $(M,\C)$ polarizes regularly, then \iffalse\textcolor{red}{for every sufficiently large $t$, the rows of $(M^{-1})^{\otimes t}$ contain a basis of a capacity achieving code for the channel $\C$. I.e.,}\fi given $\epsilon > 0$ and $c < \infty$ there is a $t_0$ such that for every $t \geq t_0$ there is a code $C \subseteq \F_q^n$ for $n=k^t$ of dimension at least $(\mathrm{Capacity}(\C)-\epsilon)\cdot n$ such that $C$ is an affine code generated by the restriction of $(M^{-1})^{\otimes t}$ to a subset of its rows and an affine shift. Moreover there is a polynomial time decoding algorithm for these codes that has failure probability bounded by $n^{-c}$.\footnote{We remark that the encoding and decoding are not completely uniform as described above, since the subset of rows and the affine shift that are needed to specify the code are only guaranteed to exist. In the case of additive channels, where the shift can be assumed to be zero, the work of Tal and Vardy~\cite{tal-vardy} (or \cite[Sec. V]{GX15}) removes this non-uniformity by giving a polynomial time algorithm to find the subset.
	}
}

In order to obtain codes with faster convergence to capacity,
we will need stronger forms of polarization, and a
more quantitative version of this theorem, with effective upper bounds on $t_0$ as a function of the gap $\epsilon$ to capacity.
The following version relates parameters of polarization with the quality of the
associated code.

\cstate{Quantitative convergence to capacity~\cite{arikan-polar,GX15,HAU14}}{theorem}{thmarikanpolar-2}{
	\label{thm:main-quant}
	Let $\C$ be a $q$-ary symmetric memoryless channel and let $M \in \F_q^{k \times k}$ be an
	invertible matrix.
	If the \Arikan\ martingale associated with $(M,\C)$ satisfies $(\tau_\ell, \tau_h,
	\varepsilon)$-polarization, then for every $t$, there is an affine code $C$,
	that is generated by the rows of $(M^{-1})^{\otimes t}$ and an affine shift,
	such that the rate of $C$ is at least 
	\[ \mathrm{Capacity}(\C) - \varepsilon(t) -
	\two \tau_h(t) \ , \]
	and $C$ can be encoded and decoded\,\footnote{The running times count the number of floating point operations where real numbers are maintained with $\Oh(\log n)$ bits of precision.} in time $\Oh(n \log n)$ where $n
	= k^t$ and failure probability of the decoder is at most $\Oh(n \cdot \log q \cdot \tau_\ell(t) ) $.}

\begin{remark}
	\label{rem:using-channel-coding-result}
	So in particular if $\tau_h(t),\varepsilon(t) = \Oh(\rho^t)$ then we get $\epsilon$ close to capacity at block lengths roughly
	$(1/\epsilon)^{\log k/\log(1/\rho)}$ which is a polynomial in $\epsilon$ provided $\rho < 1$. 
	Of course for the code to be useful, we also need $\tau_\ell(t) \ll k^{-t}$. Both conditions are guaranteed by strong polarization.
	$\Lambda$-exponentially strong polarization guarantees decoding failure probability at most $\Oh(n \cdot\log q \cdot \exp(-\Omega(n^{\Lambda/\log_2 k})))$. 
\end{remark}

This theorem is implicit in the works above, but for completeness we include a proof in \cref{sec:fast-decoder} and \cref{sec:correspondence}.

For any binary input symmetric channel, 
\Arikan\ and Telatar~\cite{arikan-telatar} proved that the martingale associated with the matrix $G_2 = \left( \begin{smallmatrix}
	1 & 0  \\
	1 & 1
\end{smallmatrix} \right)$, polarizes regularly (\Arikan's original paper~\cite{arikan-polar} proved a weaker form of regular polarization with $\tau(t) < 2^{-5t/4}$ which also sufficed for decoding error going to $0$). Subsequent work generalized this to other matrices with the work of Korada, \c{S}a\c{s}o\u{g}lu, and Urbanke~\cite{KSU10} giving a precise characterization of matrices $M$ for which the \Arikan\ martingale polarizes (again over binary input channels). We will refer to such matrices as \emph{mixing}, formally defined below for all finite fields.
\cstate{Mixing Matrix}{definition}{defmix}{
	\label{def:mixing-matrix}
	A matrix $M \in \F_q^{k \times k}$ is said to be {\em mixing}, if it is invertible and none of the permutations of the rows of $M$ yields an upper triangular matrix, i.e., for every permutation $\pi:[k]\to [k]$ there exists $i,j \in [k]$ with $j < \pi(i)$ such that $M_{i,j} \ne 0$.\footnote{We use $1$-indexing in this paper.}
}
It is not too hard to show that the \Arikan\ martingale associated with non-mixing matrices do not polarize (even weakly). In contrast, \cite{KSU10} shows that every mixing matrix over $\F_2$ polarizes regularly. Mori and Tanaka~\cite{mori-tanaka} show that the same result holds for all prime fields, and give a slightly more complicated criterion that characterizes (regular) polarization for general fields. (These works show that the decoding failure probability of the resulting polar codes is at most $2^{-n^\beta}$ for some positive $\beta$ determined by the structure of the mixing matrix --- this follows from an even stronger decay in the first of the two parameters in the definition of polarization. However, they do \emph{not} show strong polarization, which is what we achieve.)

As alluded to earlier, strong polarization is defined such that it yields codes
with polynomial gap to capacity, via \cref{thm:main-quant}.

\cstate{\cite{arikan-polar,GX15,HAU14}}{theorem}{thmarikanpolar}{
	\label{thm:poly-code}
	Let $\C$ be a $q$-ary symmetric memoryless channel and let $M \in \F_q^{k \times k}$ be an
	invertible matrix. Suppose that the \Arikan\ martingale associated with $(M,\C)$ polarizes strongly.
	
	Then, for every $c$ there exists $t_0(x) = O_c(\log x)$\footnote{The notation $O_c(\cdot)$ hides a constant factor that only depends on $c$.} such that for every $\epsilon > 0$ and every $t \geq t_0(1/\epsilon)$ there is an affine code $C$, that is generated by the rows of $(M^{-1})^{(\otimes t)}$ and an affine shift, with the property that the rate of $C$ is at least $\mathrm{Capacity}(\C)-\epsilon$, and $C$ can be encoded and decoded in time $\Oh(n \log n)$ where $n=k^t$ and failure probability of the decoder is at most $n^{-c}$.
	
	If we assume that the \Arikan\ martingale associated with $(M,\C)$ has exponentially strong polarization, then the failure probability of the decoder is at most $\exp(-n^\beta)$ for some $\beta > 0$.\footnote{Throughout this paper we use the notation $\exp(x)$ to denote a function of the form $c^x$ for some constant $c > 1$. The exact value of $c$ may be different in each usage, but will always be bounded away from $1$.}
}

The proof of this theorem, as a direct corollary from~\cref{thm:main-quant} is included in \cref{sec:fast-decoder} for completeness. 

As alluded to earlier, the only \Arikan\ martingales that were known to polarize strongly were those where the underlying matrix was
$G_2 = \left( \begin{smallmatrix}
	1 & 0  \\
	1 & 1
\end{smallmatrix} \right)$. Specifically Guruswami and Xia~\cite{GX15} and Hassani et al.~\cite{HAU14} show strong polarization of the \Arikan\ martingale associated with this matrix over any binary input symmetric channel, and Guruswami and Velingker \cite{GV15} extended to the case of $q$-ary input channels for prime $q$. By using the concept of local polarization we are able to extend these results to \emph{all} mixing matrices.

\subsection{Results II: Local polarization of \Arikan\ martingales}

The results in this subsection appear as  the pink arrows (from top and bottom box on the right to the middle right (orange) boxes) in \cref{fig:overview}.

In our second main result, we show that every mixing matrix gives rise to an \Arikan\ martingale that is locally polarizing:
\cstate{Local polarization of \Arikan\ martingales}{theorem}{thmtrianglelocal}{
	\label{thm:triangle-local}
	For every prime $q$, for every mixing matrix $M \in \F_q^{k \times k}$, and for every symmetric memoryless channel
	$\C_{Y|Z}$ over $\F_q$,
	the associated \Arikan\ martingale is locally polarizing.
}

\cref{thm:triangle-local} is proved in \cref{subsec:thm-1.14}.

We also show that the ``two-step martingale,'' or equivalently the martingale associated with $M^{\otimes 2}$ for mixing matrices $M$ are exponentially locally polarizing. 

\cstate{Exponential local polarization of \Arikan\ martingales}{theorem}{thmtrianglelocalexp}{
	\label{thm:triangle-local-exp}
	For every prime $q$, $\varepsilon > 0$, every mixing matrix $M \in \F_q^{k \times k}$, and for every symmetric memoryless channel
	$\C_{Y|Z}$ over $\F_q$,
	the \Arikan\ martingale sequence associated with $M^{\otimes 2}$ and $\C_{Y|Z}$ is $(\frac{1}{k^2}, 2-\varepsilon)$-exponentially locally polarizing.
}

\cref{thm:triangle-local-exp} is proved in \cref{sec:exp-local-arikan}. 

\subsection{Implications for polar codes with polynomial convergence to capacity}

Results in this section are the two bottom green arrows (from the middle left (green) boxes to the left most (blue) boxes) in \cref{fig:overview}.

As a consequence of \cref{thm:poly-code,thm:local-global,thm:triangle-local}, we have the following theorem.

\cstate{Polynomially fast convergence to capacity \& inverse polynomial error probability}{theorem}{thmcombo}{
	\label{thm:combo}
	~\\ For every prime $q$, every mixing matrix $M\in\F_q^{k\times k}$, every symmetric memoryless channel $\C$ over $\F_q$, and every $c < \infty$, there is a polynomial $p$ such that for every $\epsilon > 0$, and every  $n = k^t > p(1/\epsilon)$, there is an affine code $C$, that is generated by the rows of $(M^{-1})^{(\otimes t)}$ and an affine shift, with the property that the rate of $C$ is at least $\mathrm{Capacity}(\C)-\epsilon$, and $C$ can be encoded and decoded in time $\Oh(n \log n)$ and failure probability of the decoder is at most $n^{-c}$.
}

Again, as a consequence of \cref{thm:main-quant,thm:local-to-global-exp,thm:triangle-local-exp}, we have the following theorem which achieves decoding failure probability that is $\exp(-n^\beta)$ for some $\beta > 0$. We refer to such a function as \emph{root-exponentially small}, and when $\beta \to 1$, we call it \emph{near-exponentially small}.

\begin{theorem}[Polynomial convergence to capacity \& root-exponentially small error probability]
	\label{thm:thm1}
	~\\ For every prime $q$, every mixing matrix $M\in\F_q^{k\times k}$, every symmetric memoryless channel $\C$ over $\F_q$, there is a polynomial $p$ and $\beta > 0$ such that for every $\epsilon > 0$ and every $n = k^t \geq p(1/\epsilon)$, 
	there is an affine code $C$, that is generated by the rows of $(M^{-1})^{(\otimes t)}$ and an affine shift, with the property that the rate of $C$ is at least $\mathrm{Capacity}(\C)-\epsilon$, and $C$ can be encoded and decoded in time $\Oh(n \log n)$ and failure probability at most 
	$\exp(-n^\beta)$.
\end{theorem}

\subsection{Additional results optimizing decoding error probability}
\label{sec:extra-results}

The above theorems shows that all polar codes associated with every mixing matrix achieves the Shannon capacity of a symmetric memoryless channel efficiently, thus, vastly expanding on the class of polar codes known to satisfy this condition.
By choosing the mixing matrix carefully, we can even achieve decoding error probability close to $2^{-\Omega(n)}$, specifically we can get near-exponentially small decoding error 
probability, i.e., falling as $\exp(-n^\beta)$ for any desired $\beta < 1$.

\begin{theorem}[Near-exponentially small error probability and polynomial convergence to capacity]
	\label{thm:thm2}
	For every prime $q$, every symmetric memoryless channel $\C$ over $\F_q$, and every $\beta < 1$, there exists $k$, a  mixing matrix $M\in\F_q^{k\times k}$, and a polynomial $p$ such that for every $\epsilon > 0$ and every $n = k^t \geq p(1/\epsilon)$, 
	there is an affine code $C$, that is generated by the rows of $(M^{-1})^{(\otimes t)}$ and an affine shift, with the property that the rate of $C$ is at least $\mathrm{Capacity}(\C)-\epsilon$, and $C$ can be encoded and decoded in time $\Oh(n \log n)$ and failure probability at most 
	$\exp(-n^\beta)$.
\end{theorem}

\cref{thm:thm2} is proved in \cref{sec:strong-exp-failure}. 

Finally, for a broad class of channels, we show that we achieve nearly the best possible error exponent for any given mixing matrix $M$, while achieving polynomial gap to capacity, using the proofs of this paper.

\begin{theorem}[Polynomial convergence to capacity at no price in decoding error probability]
	\label{thm:thm3}
	~\\ Suppose $M\in\F_q^{k \times k}$ and $\beta > 0$ satisfy the condition that for every $q$-ary symmetric channel\footnote{A $q$-ary symmetric channel is one where the symbol is unaltered with probability $1-\theta$, and flipped to a uniform value with probability $\theta$, for a channel parameter $\theta \in [0,1]$.} 
	$\C$ and for every $\epsilon > 0$, for sufficiently large $n=k^s$, there is an affine code $C$ of length $n$ generated by the rows of $(M^{-1})^{(\otimes s)}$ of rate at least $\mathrm{Capacity}(\C)-\epsilon$ such that $C$ can be decoded with failure probability at most $\exp(-n^\beta)$. 
	
	Then, for every $\beta' < \beta$ and every symmetric channel $\C'$ with inputs from $\F_q$, there is a polynomial $p$ such that for every $\epsilon > 0$ and every $n = k^t \geq p(1/\epsilon)$
	there is an affine code $C$, that is generated by the rows of $(M^{-1})^{(\otimes t)}$ and an affine shift, with the property that the rate of $C$ is at least $\mathrm{Capacity}(\C')-\epsilon$, and $C$ can be encoded and decoded in time $\Oh(n \log n)$ and failure probability at most $\exp(-n^{\beta'})$.
\end{theorem}

\cref{thm:thm3} is proved in \cref{sec:univ-local-polarization}. It is worth emphasizing two desirable aspects about \cref{thm:thm3}:
\begin{enumerate}
    \item We only need to assume that polar codes based on $M$ achieve capacity for the $q$-ary symmetric channel, but get a conclusion for every symmetric channel (with $\F_q$ inputs).
    \item Further, we assume nothing about the speed of convergence to capacity for the $q$-ary symmetric channel, and conclude polynomial convergence to capacity (positive scaling exponent) for arbitrary symmetric channels. We do assume root-exponential decoding error probability for the $q$-ary symmetric channel but this has been established for all mixing matrices in the limit of $n\to\infty$~\cite{KSU10,mori-tanaka}. Moreover in this limit \cite{KSU10} gives a characterization of the best possible exponent $\beta$ for any given matrix $M$. \cref{thm:thm3} asserts that essentially the same characterization applies with polynomial convergence to capacity.
\end{enumerate}

\subsection{Comparison with previous analyses of (strong) polarization}

While most of the ingredients going into our eventual analysis of strong polarization are familiar in the literature on polar codes, our proofs end up being much simpler and modular. We describe some of the key steps in our proofs and contrast them with those in previous works.

\medskip\noindent \textbf{Definition of Local Polarization.}
While we are not aware of a definition similar to local polarization being explicit in the literature before, such notions have been considered implicitly before.
For instance, for the variation in the middle (where we require that $\E[(X_{t+1}-X_t)^2] \geq \theta$ if $X_t \in (\tau,1-\tau)$)
some of the previous analyses (e.g., in \cite{GX15,GV15}) required  $\theta$ be quadratic in $\tau$.
In contrast, our requirement on the variation is very weak and qualitative, allowing any function $\theta(\tau) > 0$. Similarly, our requirement in the {\em suction at the ends} case is relative mild and qualitative. In previous analyses the requirements were of the form ``if $X_t \leq \tau$ then $X_{t+1} \leq X_t^2$ with positive probability.'' This high demand on the suction case prevented the analyses from relying only on the local behavior of the martingale $X_0,\ldots,X_t,\ldots$ and instead had to look at other parameters associated with it which essentially depend on the entire sequence. (For the reader familiar with previous analyses, this is where the Bhattacharyya parameters enter the picture.)
Our approach, in contrast, only requires arbitrarily large constant factor drop, and thereby works entirely with the local properties of $X_t$.

\medskip\noindent\textbf{Local Polarization implies Strong Polarization.}
Our proof that local polarization implies strong polarization is short (about 3 pages) and comes in two parts. The first part uses a simple variance argument to shows that $X_t$ is exponentially close (in $t$) to the limit except with probability exponentially small in $t$. The second part then amplifies $X_t$'s proximity to $\{0,1\}$ to sub-exponentially small values using the suction at the end guarantee of each local step, coupled with Doob's martingale inequality and standard concentration inequalities.
Such a two-part breakdown of the analysis is not new; however, our technical implementation is more abstract, more general and more compact all at the same time.

\medskip\noindent\textbf{Local Polarization of \Arikan\ martingales.}
We will elaborate further on the approach for this after defining the \Arikan\ martingales, but we can say a little bit already now: First we essentially reduce the analysis of the polarization of \Arikan\ martingale associated with an arbitrary mixing matrix $M$
to the analysis when $M = G_2$. This reduction loses in the parameters
$(\alpha,\tau(\cdot),\theta(\cdot))$ specifying the level of local polarization, but since our strong polarization theorem works for any function, such loss in performance does not hurt the eventual result.
Finally, local polarization for the case where the matrix is $G_2$ is of course standard, but even here our proofs (which we include for completeness)
are simpler since they follow from known entropic inequalities on sums of two independent random variables. We stress that even quantitatively weak forms of these inequalities meet our requirements of local polarization, and we do not need strong forms of such inequalities (like Mrs. Gerber's lemma for the binary case~\cite{sasoglu-book,GX15} and an ad hoc one for the prime case~\cite{GV15}).

\medskip\noindent\textbf{General vs. Prime Fields.}
One weaknesses in our analysis that, in contrast to the result of Mori and Tanaka~\cite{mori-tanaka} who characterize the set of matrices that lead to regular polarization over genertal fields, we only get a characterization (for strong polarization) over prime fields. We feel that this limitation is not inherent to our approach. The only (but crucial) place where the prime field plays a role is in the ``variance in the middle'' lemma (\cref{lem:2x2-middle}) for \Arikan's basic $2 \times 2$ kernel $G_2$, which in fact does not polarize regularly over general fields due to the existence of subfields. There might be a way around this by reduction to a different $2 \times 2$ kernel that actually polarizes regularly.

\medskip \noindent \textbf{Concrete polynomial upper bounds on block length.}
A second weakness in our analysis is that, while we develop a general framework to prove strong polarization and polynomial convergence to capacity, the constants are not
optimized and will lead to poor upper bounds on the exponent $\mu$ of
the polynomial in the block length as a function of the gap to capacity. This quantity is called the scaling exponent, and our main goal in this work is to prove that for every mixing matrix $M$ has a finite scaling exponent $\mu =\mu(M)$. 

For the case of
$M=G_2$ and binary alphabet (the original \Arikan\ setting), an upper bound
of $\mu \le 6$ was shown in \cite{HAU14}, and improved to 5.702
in \cite{GoldinB14}, and to $4.714$ in \cite{MHU16}. For the case of the binary erasure channel (BEC), \cite{MHU16} showed an upper bound of $\mu \le 3.639$, which is close to the heuristic value of $\approx 3.627$ reported in \cite{KMTU}. This latter value is also argued as a \emph{lower bound} on $\mu$ for the binary-erasure channel in \cite{HAU14} (for the proof technique of bounding decoding error probability by the sum of Bhattacharyya parameters of the channels seen by the successive cancellation decoder). For kernels besides $G_2$, we were unaware of any concrete (or even finite) upper bounds on $\mu$ besides our work (except for large random kernels discussed next).

\medskip\noindent\textbf{Subsequent work.}
Quantitative versions of Shannon's noisy coding theorem theorem show that one can achieve a scaling exponent of $2$ for any discrete memoryless channel, and converse theorems show that this is optimal~\cite{wolfowitz,strassen}.
For erasure channels over large alphabets, it was shown in \cite{PfisterU19} that random $\ell \times \ell$ kernels for larger $\ell$ achieve a scaling exponent approaching $2$. Such a result was then shown for the binary erasure channel (BEC) in~\cite{FHMV21}.

While these results hinted at the potential of polar codes to achieve near-optimal scaling exponents, they only applied to erasure channels. Analyzing polar codes for more general channels, including the basic binary symmetric channel (BSC), is significantly more complex.\footnote{For erasure channels, all intermediate channels seen by the decoder of the recursive polar code construction are also erasure channels, with varying erasure probabilities. Even for the BSC on the other hand, the intermediate channels become incredibly complex with huge alphabet sizes. So one must effectively argue about and find a construction that is able to handle a plethora of channels that don't admit analytically simple descriptions.}
Variants of polar codes were shown to achieve a scaling exponent approaching $2$ for all binary-input symmetric channels in \cite{GRY20}, together with polynomial time constructions and quasi-linear encoding/decoding complexity. A similar result was shown for all discrete memoryless channels over any finite alphabet in \cite{WD21}, albeit the efficient construction of such codes remains to be worked out (but once constructed the codes admit efficient encoding/decoding). These results also use large random kernels. For concrete kernels, this work remains the only general approach to show strong polarization and finite scaling exponent.

\subsection{Organization of the rest of this paper}
We first introduce some of
the notation and probabilistic preliminaries used to define and analyze the
\Arikan\ martingale in \cref{app:prelims}.
We then prove \cref{thm:local-global} showing that  local polarization implies strong
polarization in \cref{app:loc_to_glob_full}.
This is followed by the formal
definition of the \Arikan\ martingale in \cref{sec:martingale}.
\cref{sec:local-k-by-k} then asserts conditions on the entropy of the sum of two independent variables and uses these to prove \cref{thm:triangle-local} asserting the local polarization of the \Arikan\ martingale.
\cref{sec:entropic-proofs} proves these entropic conditions. 
\cref{sec:exp-local-arikan} proves the exponential local polarization of the two-step \Arikan\ martingale (\cref{thm:triangle-local-exp}). 
\cref{sec:lift} we prove \cref{thm:thm2,thm:thm3} which strengthen the error analysis for codes to nearly optimal. 
Finally in \cref{sec:codes-from-polarization} we show for completeness how the \Arikan\ martingale (and its convergence) can be used to construct capacity achieving codes.

\section{Preliminaries and Notation}
\label{app:prelims}

In this section we introduce the notation needed to define the \Arikan\ martingale (which will be introduced in Section~\ref{sec:martingale}).
We also include information-theoretic and probabilistic inequalities that will
be necessary for the subsequent analysis.

\subsection{Notation}

The \Arikan\ martingale is based on a recursive construction of a vector valued random variable. To cleanly describe this construction it is useful to specify our notational conventions for vectors, tensors and how to view the tensor products of matrices. These notations will be used extensively in the following sections.

\subsubsection{General Notation}

For a prime power $q$, we use $\F_q$ to denote the finite field with $q$ elements and use $\F_q^*$ to denote the non-zero elements in $\F_q$.

We will use $\Oh(\cdot)$ for ``Big-Oh" notation.

\subsubsection{Probability Notation}

Throughout this work, all random variables involved will be discrete.
For a probability distribution $D$ and random variable $X$, we write $X \sim D$ to mean that
$X$ is distributed according to $D$, and independent of all other variables.
Similarly, for a set $S$, we write $X \sim S$ to mean that $X$ is
independent and uniform over $S$.
For a set $S$, let $\Delta(S)$ denote the set of probability distributions over $S$.

We occasionally abuse notation by treating distributions as random variables.
That is, for $\bvec{D} \in \Delta(\F_q^k)$ and a matrix $M \in \F_q^{k \x k}$,
we write $\bvec{D}M$ to denote the distribution of the random variable
$\{\bvec{X}M\}_{\bvec{X} \sim \bvec{D}}$.
For a distribution $D$ and an event $E$, we write $D | E$ to denote the
conditional distribution of $D$ conditioned on $E$.

\subsubsection{Tensor Notation}
\label{sec:notation}
Here we introduce useful notation for dealing with scalars, vectors, tensors, and tensor-products.
All scalars will be non-boldfaced, for example: $X \in \F_q$.
All our vectors will be row vectors (except when explicitly noted) and will be boldfaced.
Any tensors of order $\geq 1$ (including vectors) will be boldfaced, for example: $\bvec{Y} \in \F_q^{k}$.
One exception to this is the matrix $M$ used in the polarization transforms,
which we do not boldface.

Subscripts are used to index tensors, with indices starting from $1$.
For example, for $\bvec{Y}$ as above, $\bvec{Y}_i \in \F_q$.
Matrices and higher-order tensors are indexed with multiple subscripts:
For $\bvec{Z} \in (\F_q^k)^{\otimes 3}$,
we may write $\bvec{Z}_{1, 2, 1} \in \F_q$.
We often index tensors by tuples (\emph{multiindices}), which will be boldfaced:
For $\bvec{i} = (1, 2, 1) \in [k]^3$, we write $\bvec{Z}_{\bvec{i}} = \bvec{Z}_{1, 2, 1}$.
Let $\prec$ be the lexicographic order on these indexing tuples.

When an index into a tensor is the concatenation of multiple tuples, we emphasize this by using brackets in the subscript. For example:
for tensor $\bvec{Z}$ as above, and
$\bvec{i} = (1, 2)$ and $j = 1$, we may write
$\bvec{Z}_{[\bvec i, j]} = \bvec{Z}_{1,2,1}$.

For a given tensor $\bvec{Z}$, we can consider fixing some subset of its indices, yielding a \emph{slice} of $\bvec{Z}$ (a tensor of lower order).
We denote this with brackets, using $\cdot$ to denote unspecified indices.
For example for tensor
$\bvec{Z} \in (\F_q^k)^{\otimes 3}$
as above, we have
$\bvec{Z}_{[1, 2, \cdot]} \in \F_q^k$
and
$\bvec{Z}_{[\cdot, 1]} \in (\F_q^k)^{\otimes 2}$.

We somewhat abuse the indexing notation, using
$\bvec{Z}_{\prec \bvec{i}}$ to mean the set of variables
$\{\bvec{Z}_{\bvec{j}} : \bvec j \prec \bvec{i}\}$.
Similarly,
$\bvec{Z}_{[\bvec i, < j]}
:= \{\bvec{Z}_{[\bvec i, k]} : k < j\}$.

We occasionally unwrap tensors into vectors, using the correspondence between $(\F_q^k)^{\otimes t}$ and $\F_q^{k^t}$.
Here, we unwrap according to the lexicographic order $\prec$ on tuples.

Finally, for matrices specifically, $M_{i, j}$ specifies the entry in the $i$-th row and $j$-th column of matrix $M$.
Throughout, all vectors will be row-vectors by default.

\subsubsection{Tensor Product Recursion}
The construction of polar codes and analysis of the \Arikan\ martingale
rely crucially on the recursive structure of the tensor product.
Here we review the definition of the tensor product, and state its recursive
structure.

For a linear transform $M: \F_q^k \to \F_q^k$,
let
$M^{\otimes t} :
(\F_q^k)^{\otimes t}
\to
(\F_q^k)^{\otimes t}$
denote the $t$-fold tensor power of $M$.
Explicitly (fixing basis for all the spaces involved),
this operator acts on tensors $\bvec{X} \in
(\F_q^k)^{\otimes t}$
as:
$$
[M^{\otimes t}(\bvec{X})]_{\bvec j}
=
\sum_{\bvec i \in [k]^t}X_{\bvec i}
M_{i_1, j_1}
M_{i_2, j_2}
\cdots
M_{i_t, j_t}.
$$

The tensor product has the following recursive structure:
$M^{\otimes t} = (M^{\otimes t-1}) \otimes M$,
which corresponds explicitly to:
\begin{equation}
	[M^{\otimes t}(\bvec{X})]_{[\bvec{a}, j_t]}
	=
	\sum_{i_t\in [k]} M_{i_t, j_t}
	[M^{\otimes t-1}(\bvec{X}_{[\cdot, i_t]})]_{\bvec{a}}.
\end{equation}

In the above, if we define tensor
$$\bvec{Y}^{(i_t)}:= M^{\otimes t-1}(\bvec{X}_{[\cdot, i_t]})$$
then this becomes
\begin{equation}
	\label{eqn:tensor1}
	[M^{\otimes t}(\bvec{X})]_{[\bvec{a}, \cdot]}
	= M(
	(
	\bvec{Y}_{\bvec{a}}^{(1)},
	\bvec{Y}_{\bvec{a}}^{(2)},
	\dots,
	\bvec{Y}_{\bvec{a}}^{(k)}
	)
	)
\end{equation}
where the vector
$(
\bvec{Y}_{\bvec{a}}^{(1)},
\bvec{Y}_{\bvec{a}}^{(2)},
\dots,
\bvec{Y}_{\bvec{a}}^{(k)}
) \in \F_q^k$.

Finally, we use that $(M^{\otimes t})^{-1} = (M^{-1})^{\otimes t}$.

\subsection{Information Theory Preliminaries}

For the sake of completeness we include the information-theoretic concepts and tools we use in this paper.

For a discrete random variable $X$,
let $H(X)$ denote its binary entropy:
$$H(X) := \sum_{a \in Support(X)} p_X(a) \log\parens{\frac{1}{p_X(a)}}$$
where $p_X(a) := \Pr\parens{X = a}$ is the probability mass function of $X$.
Throughout, $\log(\cdot)$ by default denotes $\log_2(\cdot)$.

For $p \in [0, 1]$, we overload this notation, letting $H(p)$ denote
the entropy $H(X)$ for $X \sim Bernoulli(p)$.

For arbitrary random variables $X, Y$, let $H(X | Y)$ denote
the conditional entropy:
$$H(X | Y) = \E_{Y}[H(X | Y = y)].$$

For a $q$-ary random variable $X \in \F_q$, let
$\bH(X) \in [0, 1]$ denote its (normalized) $q$-ary entropy:
\begin{equation}
	\label{eq:normalized-entropy}
	\bH(X) := \frac{H(X)}{\log(q)} \ . 
\end{equation}

Finally, the \emph{mutual information} between jointly distributed random
variables $X, Y$ is:
$$I(X; Y) := H(X) - H(X | Y) = H(Y) - H(Y | X)$$

We will use the following standard properties of entropy (see, for instance, \cite{CoverThomas}):

\begin{enumerate}
	\item {\bf (Adding independent variables increases entropy):}
	For any random variables $X, Y, Z$ such that $X, Y$ are conditionally independent given $Z$, we have
	\begin{equation}
		\label{eqn:adding-indep}
		H(X + Y | Z) \geq H(X | Z)
	\end{equation}
	
	\item {\bf (Transforming Conditioning):}
	For any random variables $X, Y$, any function $f$, and any bijection $\sigma$,
	we have
	\begin{equation}
		\label{eqn:transform-conditioning}
		H(X | Y) = H(X + f(Y) | Y) = H(X + f(Y) | \sigma(Y))
	\end{equation}
	
	\item {\bf (Chain rule):}
	For arbitrary random variables $X, Y$:
	$H(X, Y) = H(X) + H(Y | X)$.
	\item {\bf (Conditioning does not increase entropy):}
	For $X, Y, Z$ arbitrary random variables, $H(X | Y, Z) \leq H(X | Y)$.
	\item {\bf (Monotonicity):} For $p \in [0, 1/2)$, the binary entropy $H(p)$ is non-decreasing with $p$. And for $p \in (1/2, 1]$, the binary entropy $H(p)$
	is non-increasing with $p$.
	\item {\bf (Deterministic postprocessing does not increase entropy):} For arbitrary random variables $X, Y$ and function $f$ we have $H(X | Y) \geq H(f(X) |Y )$.
	\item {\bf (Conditioning on independent variables):}
	For random variables $X, Y, Z$ where $Z$ is independent from $(X, Y)$, we
	have $H(X | Y) = H(X | Y, Z)$.
\end{enumerate}

\subsubsection{Channels}
\label{sec:channel}
Given a finite field $\F_q$, and output alphabet $\mathcal{Y}$, a \emph{$q$-ary channel} $\C_{Y|Z}$
is a probabilistic function from $\F_q$ to $\mathcal{Y}$.
Equivalently, it is given by $q$ probability distributions $\{\C_{Y|\alpha}\}_{\alpha \in \F_q}$ supported on $\mathcal{Y}$.
We use notation $\mathcal{C}(Z)$ to denote the channel operating on inputs
$Z$.
A \emph{memoryless channel} maps $\F_q^n$ to $\mathcal{Y}^n$ by acting independently (and identically) on each coordinate.
A \emph{symmetric channel} is a memoryless channel where for every $\alpha,\beta \in \F_q$ there is a bijection $\sigma:\mathcal{Y}\to \mathcal{Y}$ such that for every $y \in \mathcal{Y}$ it is the case that $\C_{Y=y|\alpha} = \C_{Y = \sigma(y)|\beta}$, and moreover for any pair $y_1, y_2 \in \mathcal{Y}$, we have $\sum_{x \in \F_q} C_{Y=y_1 | x} = \sum_{x \in \F_q} C_{Y=y_2 | x}$ (see, for example,~\cite[Section 7.2]{CoverThomas}).
As shown by Shannon every memoryless channel has a finite capacity, denoted
$\mathrm{Capacity}(\C_{Y|Z})$.
For symmetric channels, this is the mutual
information $I(Y; Z)$ between the input $Z$ and output $Y$ where $Z$ is drawn uniformly from $\F_q$ and $Y$ is drawn from $\C_{Y|Z}$ given $Z$.

\subsection{Basic Probabilistic Inequalities}
\label{app:basic-prob}

In this section, we collect a few useful probabilistic and information-theoretic inequalities, all of which are standard. The proofs are included for convenience.

We first show that a random variable with small-enough entropy will usually take its most-likely value  and thus maximum likelihood recovery is successful with high probability. In fact we show that even if the likelihoods are known only very approximately  maximum likelihood decoding will still be quite successful.

\begin{lemma}
	\label{lem:ML-basic}
	Let $X$ be an arbitrary discrete random variable with range $\mathcal{X}$.
	Then there exist $\hat x\in\mathcal{X}$ such that
	$$\Pr\parens{X \neq \hat x} \leq H(X).$$
	In particular, one can take $\hat x = \argmax_\alpha \{\Pr\parens{X=\alpha}\}$.
	
	Furthermore, given $\tilde{p}_\alpha$'s satisfying $|\tilde{p}_\alpha - \Pr\parens{X = \alpha}| \leq 1/4$ for every $\alpha \in \mathcal{X}$, if we let $\tilde{x} = \argmax_\alpha \{\tilde{p}_\alpha \}$ then we have $\Pr\parens{X \neq \tilde{x}} \leq 3H(X)$.
\end{lemma}
\begin{proof}
	Let $\alpha := H(X)$ and let $p_i := \Pr_X \parens{ X = i}$.
	Let $\hat x = \argmax_i \{p_i\}$ be the
	value maximizing this probability.
	Let $p_{\hat x} = 1 - \gamma$. We wish to show that
	$\gamma \leq \alpha$.
	If $\gamma \leq 1/2$ we have
	\begin{align*}
		\alpha
		&=  H(X) = \sum_i p_i \log \frac1{p_i} \\
		& \geq  \sum_{i \ne \hat x} p_i \log \frac1{p_i} \tag{Since all summands are non-negative}\\
		& \geq  \sum_{i \ne \hat x} p_i \log \frac1{\sum_{j\ne \hat x} p_j} \tag{Since $p_i \leq \sum_{j\ne \hat x}
			p_j$.}\\
		& =  \left( \sum_{i \ne \hat x} p_i \right) \cdot \log \left( \frac1{\sum_{j\ne \hat x} p_j} \right)\\
		& =  \gamma \cdot \log 1/\gamma \\
		& \geq  \gamma \tag{Since $\gamma \leq 1/2$ and so $\log 1/\gamma \geq 1$}
	\end{align*}
	as desired. Now if $\gamma > 1/2$ we have a much simpler case since now we have
	\begin{align*}
		\alpha
		&=  H(X) =  \sum_i p_i \log \frac1{p_i} \\
		& \geq  \sum_i p_i \log \frac1{p_{\hat x}} \tag{Since $p_i \leq p_x$} \\
		& =  \log \frac1{p_{\hat x}} \tag{Since $\sum_i p_i =1$} \\
		& =  \log \frac1{1-\gamma}\\ 
		& \geq  1. \tag{Since $\gamma \geq 1/2$}
	\end{align*}
	But $\gamma$ is always at most $1$ so in this case also we have $\alpha \geq 1 \geq \gamma$ as desired.
	
	For the furthermore part note that if $\gamma < 1/4$ then, by the condition $|\tilde{p}_\alpha - p_\alpha|\leq 1/4$, we have $\tilde{p}_{\hat{x}} > 1/2$ while $\tilde{p}_{x'} < 1/2$ for every $x' \ne \hat{x}$. Thus in this case we have $\tilde{x} = \hat{x}$ and so by the first part above we have $\Pr\parens{X \neq \tilde{x}} = \Pr\parens{X \neq \hat{x}}\leq H(X)$.
		Now if $\gamma > 1/4$ as in the second part above, we have $H(X) \geq \log\frac1{1-\gamma} \geq .415 \geq 1/3$ and so we get 
		$\Pr\parens{X \neq \hat{x}}\leq 1 \leq 3 H(X)$.
	\end{proof}

For the decoder, we will need a conditional version of
\cref{lem:ML-basic}, 
saying that if a variable $X$
has low conditional entropy conditioned on $Y$,
then $X$ can be predicted well given the instantiation of variable $Y$.
\begin{lemma}
	\label{lem:ML-conditional}
	Let $X, Y$ be arbitrary discrete random variables with range $\mathcal{X}, \mathcal{Y}$ respectively.
	Then there exists a function $\hat X: \mathcal{Y}\to \mathcal{X}$ 
	such that
	$$\Pr_{X, Y}\parens{X \neq \hat X(Y)} \leq H(X|Y)$$
	In particular, the following estimator satisfies this:
	$$\hat X(y) := \argmax_x\set{\Pr\parens{X = x | Y = y}}.$$
	Furthermore, given $\tilde{p}_{x,y}$'s satisfying $|\tilde{p}_{x,y} - \Pr\parens{X = x | Y = y}| \leq 1/4$ for every $x \in \mathcal{X}, y \in \mathcal{Y}$, if we let $\tilde{X}(y) = \argmax_x \{\tilde{p}_{x,y} \}$ then we have $\Pr\parens{X \neq \tilde{X}(y)} \leq 3H(X|Y)$.
\end{lemma}
\begin{proof}
	For every setting of $Y = y$, we can bound the error probability
	of this estimator using
	\cref{lem:ML-basic}
	applied to the conditional distribution $X | Y = y$:
	\begin{align*}
		\Pr_{X, Y}\parens{X \neq \hat X(Y)}
		&= \E_{Y}\brackets{\Pr_{X | Y}\parens{\hat X(Y) \neq X}}\\
		&\leq \E_{Y}\brackets{H(X | Y = y)} \tag{\cref{lem:ML-basic}}\\
		&= H(X | Y) \ . \qedhere
	\end{align*}
	The furthermore part follows similarly by using the furthermore part of \cref{lem:ML-basic}.
\end{proof}

We also use the well-known Fano's inequality which works as a weak converse to the above lemma, asserting that if a random variable $X$ is predictable given $Y$ then its conditional entropy is small.

\begin{lemma}[Fano's inequality]
	\label{lem:prediction-gives-entropy}
	For a pair of random variables $(X,Y) \in \mathcal{X}\times \mathcal{Y}$, if there exists a function $\hat X: \mathcal{Y}\to \mathcal{X}$ such that $\P(\hat X(Y) \not= X) \leq \delta$ with $\delta < \frac{1}{2}$, then $H(X | Y) \leq 2 \delta (\log \delta^{-1} + \log |\mathcal{X}|)$.
\end{lemma}

We will need an inverse to the usual Chebychev inequality. Recall that Chebychev shows that variables with small variance are concentrated close to their expectation:
\[\Pr\parens{|Z-\E[Z]|\ge \lambda}\le \frac{\Var(Z)}{\lambda^2}.\] 
The Paley-Zygmund inequality below can be used to invert it (somewhat) --- for a random variable $W$ with comparable fourth and second central moment, by applying the lemma below to $Z = (W - \E[W])^2$ we can deduce that $W$ has positive probability of deviating noticeably from the mean.
\begin{lemma}[Paley-Zygmund]
	\label{lem:paley-zygmund}
	If $Z \geq 0$ is a random variable with finite variance, then
	\begin{equation*}
		\P(Z > \lambda \E[Z]) \geq (1 - \lambda)^2 \frac{\E[Z]^2}{\E[Z^2]}.
	\end{equation*}
\end{lemma}

Next, we define the notion of a sequence of random variables being adapted to another sequence of variables, which will be useful in our later proofs.

\begin{definition}
	\label{def:adapted-sequence}
	We say that a sequence $Y_1, Y_2 \ldots$ of random variables is \emph{adapted} to the sequence $X_1, X_2 \ldots$ if and only if for every $t$,  $Y_t$ is completely determined given $X_1, \ldots X_t$. We will use $\E[Z | X_{[1:t]}]$ as a shorthand $\E[Z | X_1, \ldots X_t]$, and $\P\parens{ E | X_{[1:t]}}$ as a shorthand for $\E[\1_E | X_1, \ldots X_t]$. If the  underlying sequence $X$ is clear from context, we will skip it and write just $\E[Z | \cF_t]$.
\end{definition}

\begin{lemma}
	\label{lem:sum-of-subexp}
	Consider a sequence of non-negative random variables $Y_1, Y_2, \ldots, Y_t, \ldots$ adapted to the sequence $X_1, X_2,\dots$. If for every $t$ we have $\Pr\parens{Y_{t+1} > \lambda\,|\, X_{[1:t]}} \le \exp(-\lambda)$, then for every $T>0$:
	\begin{equation*}
		\P\parens{\sum_{i \leq T} Y_i > CT} \leq \exp(-\Omega(T))
	\end{equation*}
	for some universal constant $C$.
\end{lemma}
\begin{proof}
	First, observe that
	\begin{align}
		\E[\exp(Y_{t+1}/2) | \cF_t] & = \int_{0}^\infty \P(\exp(Y_{t+1}/2) > \lambda | \cF_t) \d \lambda \notag\\
		& \leq 1 + \int_1^{\infty} \exp(-2 \log \lambda) \d \lambda  \notag\\
		& = 1 + \int_1^{\infty} \lambda^{-2} \d \lambda \notag\\
		& \leq \exp(C_0) \label{eq:mgf_bound}
	\end{align}
	for some constant $C_0$. On the other hand, we have decomposition (where we apply \cref{eq:mgf_bound} in the first inequality):
	\begin{align*}
		\E\brackets{\exp\parens{\sum_{i \leq T} \frac{Y_i}{2}}} & = \E\brackets{ \E\brackets{\exp \parens{\sum_{i \leq T} \frac{Y_i}{2}} | \cF_{T-1}}} \\
		& = \E \brackets{ \exp\parens{\sum_{i\leq T-1} Y_i/2} \E\brackets{\exp\parens{Y_T/2} | \cF_{T-1}}} \\
		& \leq \E \brackets{ \exp\parens{\sum_{i \leq T-1} Y_i/2} }\cdot \exp(C_0) \\
		& \leq \cdots \\
		& \leq \exp(C_0 T).
	\end{align*}
	
	\noindent In the above, the second equality follows from the fact that the sequence $Y_1,Y_2,\dots$ is adapted to $X_1,X_2,\dots$. 
	We can now apply Markov inequality to obtain the desired tail bound:
	\begin{align*}
		\P\parens{ \sum_{i \leq T} Y_i > 4 C_0 T} & = \P\parens{ \exp\parens{\frac{1}{2} \sum_{i \leq T} Y_i } > \exp(2 C_0 T)} \\
		& \leq \E\brackets{ \exp\parens{\frac{1}{2} \sum_{i \leq T} Y_i}} \cdot \exp(-2 C_0 T)\\
		& \leq \exp\parens{- C_0 T} \ . \qedhere
	\end{align*}
\end{proof}

The following bound for a moment generating function of a bounded random variable is standard and is commonly used in the proof of Bernstein inequality.
\begin{lemma}
\label{lem:bounded-subgamma}
For any random variable $X$ such that $|X| < 1$ with probability $1$, and every $0 < \lambda < \frac{1}{4}$, we have
\begin{equation*}
    \log \E[\exp(\lambda X)] \leq \lambda \E[X] + C \lambda^2 \E[X^2],
\end{equation*}
where $C$ is some universal constant.
\end{lemma}
\begin{proof}
    Since $|X| < 1$, we have $\E |X|^k \leq \E X^2$, and therefore
    \begin{align*}
        \E \exp(\lambda X) & = \sum_k \frac{\lambda^k}{k!} \E[X^k] \\
        & \leq 1 + \lambda \E[X] + (\lambda^2 + \Oh(\lambda^3)) \E [X^2].
    \end{align*}
    Moreover for some constant $C$, and every $|x| < \frac{1}{2}$, we have $\log(1+x) \leq x + C x^2$, therefore 
    \begin{align*}
     \log \E[\exp(\lambda X)] & \leq \lambda \E[X] + C \lambda^2(\E[X^2] + \E[X]^2) + \Oh(\lambda^3) \E[X^2] \\
     & \leq \lambda \E[X] + C' \lambda^2 \E[X^2]. \qedhere
    \end{align*}
\end{proof}

\begin{lemma}
	\label{lem:sum-of-events}
	Consider a sequence of random variables $Y_1, Y_2, \ldots$ with $Y_i \in \{0, 1\}$, adapted to the sequence $X_t$. If $\P(Y_{t+1} = 1 | X_{[1:t]}) > \mu_{t+1}$ for some deterministic value $\mu_t$,  then for $\mu := \sum_{t\leq T} \mu_t$ and any $\epsilon > 0$ we have
	\begin{equation*}
		\P\parens{\sum_{t\leq T} Y_t < (1 - \epsilon)\mu} \leq \exp\parens{-\Omega(\epsilon^2 \mu)}
	\end{equation*}
\end{lemma}
\begin{proof}
	Consider a random variable $M_{t+1} := \E[Y_{t+1} | X_{[1:t]}]$ (depending on $X_{[1:t]}$), we know that $M_t > \mu_t$ with probability 1, and let us take $Z_t := (1 - \epsilon) M_t - Y_t$. 

	Standard calculation involving Markov inequality yields following bound for any $\lambda > 0$
	\begin{align}
		\P\parens{\sum_{t \leq T} Y_t < \sum_{t \leq T} (1 - \epsilon) \mu_t} & \leq \P\parens{\sum_{t \leq T} Y_t < \sum_{t \leq T} (1 - \epsilon) M_t} \nonumber  \\
		& = \P\parens{\sum_{t \leq T} \lambda Z_t > 0} \nonumber \\
		& = \P\parens{\exp\parens{\sum_{t\leq T} \lambda Z_t} > 1} \nonumber \\
		& \leq \E\brackets{\exp\parens{\sum_{t \leq T} \lambda Z_t}}. \label{eq:first}
	\end{align}
	
	To bound this latter quantity, we introduce conditioning on $X_{[1:T-1]}$
	\begin{align}
		\E\brackets{\exp\parens{\sum_{t \leq T} \lambda Z_t}} & = \E\brackets{\E\brackets{\exp\parens{\sum_{t \leq T} \lambda Z_t} | X_{[1:T-1]}}} \nonumber \\
		& = \E\brackets{\exp\parens{\sum_{t \leq T-1} \lambda Z_t} \E[\exp(\lambda Z_T) | X_{[1:T-1]}]}, \label{eq:soe-mgf}
	\end{align}
	where the second equality follows from the fact that $Z_t$ is adapted to $X_t$.
	
	By \cref{lem:bounded-subgamma} for any $0 < \lambda < \frac{1}{4}$, we have 
	\begin{equation*}
	     \E \brackets{ \exp(\lambda Z_T) | X_{[1:T-1]}} \leq \exp(- \lambda \varepsilon M_T  + C_1 \lambda^2 M_T) 
	\end{equation*}
	for some constant $C_1$. 
	Now if we chose $\lambda = \frac{1}{2 C_1} \varepsilon$, we get
	\begin{align}
	    \E\brackets{\exp\parens{\lambda Z_T} | X_{[1:T-1]}} & \leq \E\brackets{\exp\parens{-C \epsilon^2 M_T}} \nonumber \\
	    &  \leq \exp(- C \epsilon^2 \mu_T))
	\end{align}
	where $C=\frac{1}{8 C_1}$, since $\mu_T \leq M_T$ deterministically.

	Together with \eqref{eq:soe-mgf} this yields
	\begin{align}
	    \E\brackets{\exp\parens{\sum_{t \leq T} \lambda Z_t}}  & \leq \E\brackets{\exp\parens{\sum_{t \leq T-1} \lambda Z_t}} \exp(- C \epsilon^2 \mu_T) \nonumber \\
	    & \leq \cdots \nonumber \\
	    & \leq \E\brackets{\exp\parens{\sum_{t \leq T} - C \epsilon^2 \mu_t}} = \exp(-\Omega(\epsilon^2 \mu))  \label{eq:second}.
	\end{align}
Finally, combining \eqref{eq:first} and \eqref{eq:second} we have $\P\parens{\sum_{t \leq T} Y_t < (1 - \epsilon) \mu} \leq \exp\parens{-\Omega(\epsilon^2 \mu)}$ as desired.
\end{proof}

Finally, we will use the well-known Doob's martingale inequality:
\begin{lemma}[Doob's martingale inequality {\cite[Theorem 5.4.2]{Durrett11probability:theory}}] \label{lem:doobs-ineq}
	If a sequence $X_0, X_1, \ldots $ is a martingale, then for every $T$ we have
	\begin{equation*}
		\P\parens{\sup_{t \leq T} X_t > \lambda} \leq \frac{\E [|X_T|]}{\lambda}
	\end{equation*}
\end{lemma}
\begin{corollary}
	\label{cor:doobs}
	If $X_0,  X_1, \ldots$ is a  \emph{nonnegative}  martingale, then for every $T$ we have
	\begin{equation*}
		\P\parens{\sup_{t \leq T} X_t > \lambda} \leq \frac{\E[X_0]}{\lambda}
	\end{equation*}
\end{corollary}

\section{Local to global polarization}
\label{app:loc_to_glob_full}

In this section we prove \cref{thm:local-global,thm:local-to-global-exp}, which assert that every (exponentially) locally polarizing $[0,1]$-martingale is also (exponentially) strongly polarizing. The proofs in this section depend on some basic probabilistic concepts and inequalities mentioned in \cref{app:basic-prob}.

The proof of both statements are implemented in two main steps. In the first step, common to both,  we show that any locally polarizing martingale, is mildly polarizing, namely that it is $\parens{\parens{1 - \frac{\nu}{2}}^t, \parens{1 - \frac{\nu}{2}}^t, \parens{1 - \frac{\nu}{4}}^t}$-polarizing for \emph{some} constant $\nu$ depending only on the parameters $\alpha, \tau, \theta$ of local polarization. This means that, except with exponentially small probability, $\min\{X_{t/2}, 1- X_{t/2}\}$ is exponentially small in $t$, which we can use to ensure that $X_{s}$ for all $\frac{t}{2} \leq s \leq {t}$ stays in the range where the conditions of \emph{(strong) suction at the ends} apply (again, except with exponentially small failure probability). In the second step, we show that if the martingale stays in the \emph{suction at the ends} regime, it will polarize strongly --- i.e. if we have a $[0,1]$-martingale, such that in each step it has probability at least $\alpha$ to decrease by a factor of $c$, we can deduce that at the end we have $\P(X_T  > c^{-\alpha T/4}) \leq \exp(-\Omega(\alpha T))$.\footnote{This is enough since we pick $c$ to be large enough (given $\gamma$) so that $c^{-\alpha T/4}\le \gamma^T$ and we pick $\beta$ and $\eta$ such that $\beta\eta^T\ge \exp(-\Omega(\alpha T))$.} A completely similar argument shows that when the martingale shows strong suction at the low end we have  $\P(X_T  > \exp(-\Delta^T)) \leq \exp(-\Omega(\alpha T))$, for some $\Delta > 1$, thus yielding exponentially strong polarization.

\subsection{Mild Polarization}\label{ssec:local-to-mild}

We start by showing that in the first $t/2$ steps we do get exponentially small polarization, with all but exponentially small failure probability. This is proved using a simple potential function $\min\{\sqrt{X_t},\sqrt{1-X_t}\}$ which we show shrinks by a constant factor, $1- \nu$ for some $\nu > 0$, in expectation at each step. Previous analyses in \cite{GX15,GV15} tracked
$\sqrt{X_t(1-X_t)}$ (or some tailormade algebraic functions~\cite{HAU14,MHU16}) as potential functions, and relied on quantitatively strong forms of variance in the middle to demonstrate that the potential diminishes by a constant factor in each step.
While such analyses can lead to sharper bounds on the parameter $\nu$, which in turn translate to better \emph{scaling exponents} in the polynomial convergence to capacity, e.g. see \cite[Thm. 18]{HAU14} or \cite[Thm. 1]{MHU16}, these analyses are more complex, and less general.

\begin{restatable}{lemma}{lempotentialfunction}
	\label{lem:potential-function}
	If a $[0,1]$-martingale sequence $X_0, \ldots X_t, \ldots,$ is $(\alpha,\tau(\cdot),\theta(\cdot))$-locally polarizing, then there exist $\nu > 0$, depending only on $\alpha, \tau, \theta$, such that 
	\[\E [\min(\sqrt{X_t}, \sqrt{1-X_t}) ] \leq (1 - \nu)^{t}.\]
\end{restatable}

\begin{proof}
	Set $\tau_0 = \tau(4), \theta_0 = \theta(\tau_0)$.
	We will show that $\E [\min(\sqrt{X_{t+1}}, \sqrt{1-X_{t + 1}}) | X_{t}] \leq (1 - \nu) \min(\sqrt{X_{t}}, \sqrt{1 - X_{t}})$,
	for some $\nu>0$ depending on $\tau_0,\theta_0$ and $\alpha$. The statement of the lemma will follow by induction. The base case of $t=0$ follows since $X_0\in [0,1]$.

	Let us condition on $X_t$, and first consider the case $X_t \in (\tau_0, 1 - \tau_0)$. We know that
	\[\E [\min(\sqrt{X_{t+1}}, \sqrt{1 - X_{t+1}})] \leq \min(\E[\sqrt{X_{t + 1}}], \E[\sqrt{1 - X_{t + 1}}]),\]
	we will show that $\E[\sqrt{X_{t+1}}] \leq (1 - \nu) \sqrt{X_t}$. The proof of $\E[\sqrt{1 - X_{t+1}}] \leq (1 - \nu)\sqrt{1-X_t}$ is symmetric. 
	
	Indeed, let us take $R := \sqrt{\frac{X_{t+1}}{X_t}}$.
	Because $(X_t)_t$ is a martingale, we have $\E [R^2] = 1$, and by Jensen's inequality, we have that $\E[R] \leq \sqrt{\E[R^2]} \leq 1$, where all the expectations above are conditioned on $X_t$. Take $\delta$ such that $\E[R] = 1 -\delta$. We will show a lower bound on $\delta$ in terms of $\theta_0, \tau_0$ and $\alpha_0$.
	
	We note that
	\begin{equation}
		\label{eq:var:T_ub}
		\Var(R)=\E[R^2]-\parens{\E[R]}^2 = 1-(1-\delta)^2= 2\delta-\delta^2 
		\le 2\delta.
	\end{equation}
	
	The high-level idea of the proof is that we can show that local polarization criteria implies that $T$ is relatively far from $1$ with noticeable probability, but if $\E[R]$ were close to one, by Chebyshev inequality we would be able to deduce that $R$ is far from its mean with much smaller probability. This implies that mean of $R$ has to be bounded away from $1$.
	
	More concretely, observe first that by Chebyshev inequality, we have $\P(|R - \E [R]| > \lambda) < \frac{\Var(R)}{\lambda^2}  
	\leq \frac{2\delta}{\lambda^2}$, where the inequality follows from~\eqref{eq:var:T_ub}.
	Hence, for $C_0=4$, we have: 
	\begin{equation}
		\P\left(|R - 1| \ge \delta + C_0 \sqrt{\delta} \theta_0^{-1} \tau_0^{-2}\right) \leq \frac{1}{8} \theta_0^2 \tau_0^4.
		\label{eq:t-close-to-one}
	\end{equation}
	
	On the other hand, because of the \emph{Variation in the middle condition} of local polarization, we have
	\begin{equation*}
		\Var(R^2) = \frac{\E[X_{t+1}^2]}{X_t^2}-\frac{\E[X_{t+1}]^2}{X_t^2}=\frac{\E[X_{t+1}^2] - X_t^2}{X_t^2} \geq \frac{\theta_0}{X_t^2} \geq \theta_0,
	\end{equation*}
	where the second equality follows from the fact that $\E[X_{t+1}]=X_t$ and the last inequality follows since $X_t \leq 1$.
	Moreover $R<\frac{1}{\sqrt{\tau_0}}$, because $\sqrt{X_{t+1}} < 1$ and $\sqrt{X_t} > \sqrt{\tau_0}$.
	
	Let us now consider $Z = (R^2 - 1)^2$. We have $\E[Z] = \Var(R^2) \geq \theta_0$, and moreover $\E [Z^2] < \tau_0^{-4}$
	(because $R$ is bounded and $\tau_0\le 1$), hence by \cref{lem:paley-zygmund} (for $C_1=1/2$)
	\begin{equation*}
		\P\left((1 - R^2)^2 > C_1 \theta_0\right) \geq \frac{1}{4} \theta_0^2 \tau_0^4.
	\end{equation*}
	And also $1 - R^2 = - (1-R)^2 + 2 (1-R) < 2 (1-R)$, hence if $(1 - R^2)^2 > C_1 \theta_0$ then $|1-R| > \frac{\sqrt{C_1}}{2} \sqrt{\theta_0}$, which implies (for the choice of $C_2=\sqrt{C_1}/2$): 
	\begin{equation}
		\P\parens{|R-1| > C_2 \sqrt{\theta_0}} \geq \frac{1}{4} \theta_0^2 \tau_0^4.
		\label{eq:t-far-from-one}
	\end{equation}
	By comparing \cref{eq:t-close-to-one,eq:t-far-from-one}, we deduce  that $C_2 \sqrt{\theta_0}<\delta + C_0 \sqrt{\delta} \theta_0^{-1} \tau_0^{-2}$, which in turn implies that $\delta \geq C_4 \theta_0^3\tau_0^4$,
	(for $C_4=C_2^2/(4C_0^2)$-- note that with our choice of parameters, we have $ C_0 \sqrt{\delta} \theta_0^{-1} \tau_0^{-2}\ge \delta$)
	and by the definition of $\delta$ we have $\E[\sqrt{X_{t+1}} | X_t] \le(1 - \delta) \sqrt{X_t}]$.
	The same argument applies to show that $\E[\sqrt{1 - X_{t+1}}|X_t] \leq (1 - C_4 \theta_0^3 \tau_0^4) \sqrt{1 - X_t}$.
	
	Consider now the case when $X_t < \tau_0$. For $T, \delta$ as above (and again after conditioning on $X_t$), we have $\Var(R) < 2\delta$ (note that the argument for this inequality from the previous case also holds here), and hence by Chebyshev inequality (for the choice of $C_5=2$): 
	\begin{equation}
		\P\left(|R - 1| \ge \delta +  C_5 \sqrt{\frac{\delta}{\alpha}}\right) \leq \frac{\alpha}{2}.
		\label{eq:t-close-to-one-2}
	\end{equation}
	On the other hand, because of the \emph{suction at the end} condition of local polarization, we know that with probability $\alpha$, we have $R \le \frac{1}{2}$,
	which means $|R - 1| \ge \frac{1}{2}$ and by comparing this with~\cref{eq:t-close-to-one-2},
	we deduce that $\delta +  C_5 \sqrt{\frac{\delta}{\alpha}}\ge \frac 12$, which in turn implies that $\delta \geq C_6 \alpha$ (for $C_6=\frac{1}{16 C_5^2}$-- note that by our parameter choices we have $C_5 \sqrt{\frac{\delta}{\alpha}}\ge \delta$). 
	Therefore, in the case $X_t < \tau_0$, we have $\E [\sqrt{X_{t+1}} | X_t] \leq (1- C_6 \alpha) \sqrt{X_{t}} = (1 - C_6 \alpha) \min(\sqrt{X_{t}}, \sqrt{1 - X_{t}})$. The case $X_t > 1 - \tau_0$ is symmetric and is omitted.
	
	This implies the statement of the lemma with $\nu = \min(C_6\alpha, C_4 \theta_0^3\tau_0^2)$.
\end{proof}

\begin{restatable}{corollary}{cormodestpolarization}
	\label{cor:modest-polarization}
	If a $[0,1]$-martingale sequence $X_0, \ldots X_t, \ldots,$ is $(\alpha,\tau(\cdot),\theta(\cdot))$-locally polarizing, then there exist $\nu > 0$, depending only on $\alpha, \tau, \theta$, such that 
	\[\P\left(\min(X_{t/2}, 1-X_{t/2}) > \lambda (1 - \frac{\nu}{2})^{t} \right) \leq (1 - \frac{\nu}{4})^{t} \frac{1}{\sqrt{\lambda}}.\]
\end{restatable}
\begin{proof}
	By applying Markov Inequality to the bound from \cref{lem:potential-function} (with $t/2$ instead of $t$), we get 
	\begin{align*}
		\P\left(\min\left(X_{t/2}, 1- X_{t/2}\right) > \lambda (1 - \frac{\nu}{2})^{t}\right) &= \P\left(\min\left(\sqrt{X_{t/2}}, \sqrt{1 - X_{t/2}}\right) > \sqrt{\lambda} (1 - \frac{\nu}{2})^{t/2}\right)\\
		& \leq (1 - \nu)^{t/2}(1-\frac{\nu}{2})^{-t/2}\frac{1}{\sqrt{\lambda}}\\
		& \leq ( 1 -\frac{\nu}{4})^{t}\frac{1}{\sqrt{\lambda}}.\qedhere
	\end{align*}
\end{proof}

\subsection{Strong Polarization}\label{ssec:local-to-strong-global}

Next we show that if a [0,1]-martingale indeed stays in the \emph{suction at the ends} range for all steps $s\geq \frac{t}{2}$, i.e. in each step it has constant probability $\alpha$ of dropping by some large constant factor $C$, then at the end we may expect it to be $\left(C^{-\alpha t / 8}, C^{-\alpha t / 8}, \exp(-\Omega(\alpha t))\right)$-polarizing.

\begin{restatable}{lemma}{lemstrongpolarization}
	There exists $c < \infty$, such that for all $K, \alpha$ with $K \alpha \geq c$ the following holds.
	Let $X_t$ be a martingale satisfying $\P\left(X_{t+1} < e^{-K} X_t | X_t\right) \geq \alpha$, where $X_0 \in (0, 1)$. Then $\P(X_T > \exp(- \alpha K T/4)) \leq \exp(-\Omega(\alpha T))$.
	\label{lem:strong-polarization}
\end{restatable}

\begin{proof}
	Consider $Y_{t+1} := \log \frac{X_{t+1}}{X_t}$, and note that sequence $Y_{t}$ is adapted to sequence $X_t$ in the sense of \cref{def:adapted-sequence}. We have the following bounds on the upper tails of $Y_{t+1}$, conditioned on $X_{[1: t]}$, given by Markov inequality (and recalling that $\E[X_{t+1}|X_t]=X_t$):
	\begin{equation*}
		\P(Y_{t+1} > \lambda \,|\, \mathcal{F}_t) = \P\left(\frac{X_{t+1}}{X_t} > \exp(\lambda) \,\middle|\, X_{[1: t]}\right) = \P\left(X_{t+1} > \exp(\lambda) X_{t} \,|\, X_{[1: t]}\right) \leq \exp(-\lambda).
	\end{equation*}
	
	Let us decompose $Y_{t+1} =: (Y_{t+1})_+ + (Y_{t+1})_-$, where $(Y_{t+1})_+ := \max(Y_{t+1}, 0)$. By \cref{lem:sum-of-subexp} and the fact that $(Y_{t+1})_+\ge Y_{t+1}$,
	\begin{equation*}
		\P\left(\sum_{t \leq T} (Y_{t+1})_+ > C T\right) \le \exp(-\Omega(T)).
	\end{equation*}
	
	On the other hand, let $E_{t+1}$ be the indicator of $Y_{t+1} \leq - K$. It is again adapted to the sequence $X_t$, and we know that $\P(E_{t+1} | X_{[1: t]}) \geq \alpha$, hence by \cref{lem:sum-of-events} with probability at most $\exp(-\Omega(\alpha T))$ at most $\alpha T/2$ of those events holds. Note that $(Y_t)_{-}\le 0$, which implies that if at least $\alpha T/2$ of the events $E_t$ hold then we have $\sum_{t \leq T} (Y_t)_{-} \le - \alpha K T/2$. Thus, we have $\P(\sum_{t \leq T} (Y_t)_{-} > - \alpha K T/2) \le \exp(-\Omega(\alpha T))$. Therefore, as long as $\alpha K/4 > C$ (which is true if we set $c=4C$), we can conclude
	\begin{equation*}
		\P\left(\sum_{t\leq T} Y_t > - \alpha KT/4\right) \leq \exp(-\Omega(T)) + \exp(-\Omega(\alpha T )) \le \exp(-\Omega(\alpha T)).
	\end{equation*}
	The proof is complete by noting that $\sum_{t\le T} Y_t=\log(X_T/X_0)$ and recalling that $X_0\le 1$.

\end{proof}

We are now ready to show that local polarization leads to strong polarization:
\begin{proof}[Proof of \cref{thm:local-global}]
	For given $\gamma$, we take $K$ to be large enough so that $\exp(- \alpha K/8) \le \gamma$, and moreover $\alpha K$ to be large enough to satisfy assumptions of \cref{lem:strong-polarization}. Let us also take $\tau_0 = \tau(e^K)$.
	We consider $\nu$ as in \cref{cor:modest-polarization}. We have
	\begin{align*}
		\P\left(\min(X_{t/2}, 1 - X_{t/2}) > \left(1 - \frac{\nu}{2}\right)^t \tau_0\right) &\leq (1 - \frac{\nu}{4})^{-t} \frac{1}{\sqrt{\tau_0}}.
	\end{align*}
	
	Now Doob's martingale inequality~(\cref{cor:doobs}) implies that, conditioned on $X_{t/2} < (1 - \frac{\nu}{4})^t \tau_0$, we have $\P(\sup_{i \in (t/2, t)} X_i > \tau_0) \leq (1 - \frac{\nu}{4})^t$.
	
	Finally, after conditioning on $X_i\le \tau_0$, $\forall\, {t/2 \leq i \leq t}$, process $X_i$ for $i \in (t/2, t)$ satisfies conditions
	of \cref{lem:strong-polarization}, because $X_i$ always stays below $\tau_0$ and as such \emph{suction at the end} condition of local polarization corresponds exactly to the assumption in this lemma. Therefore we can conclude that except with probability $\exp(-\Omega(\alpha t))+(1 - \frac{\nu}{4})^{-t} \frac{1}{\sqrt{\tau_0}}$ (which is $\exp\parens{-\Omega_{\alpha,\nu}(t)}$), we have $X_t < \exp(-\alpha K t / 8) = \gamma^t$.
	The other case ($1 - X_{t/2} < (1 - \frac{\nu}{2})^t \tau_0$) is symmetric, and in this case we get $1 - X_t < \exp(-\alpha K t / 8)$ except with probability $\exp\parens{- \Omega_{\alpha,\nu}( t)}$.
\end{proof}

\subsection{Exponentially strong polarization \label{sec:exp-local-to-global}}

In this section, we prove the analog of \cref{thm:local-global}-- \cref{thm:local-to-global-exp}. We first prove a helper lemma.
\begin{lemma}
	\label{lem:single-piece-suction}
	There exist $C < \infty$ such that for all $0<\eta<1, b\ge 1, 0<\varepsilon<1$
	following holds. Let $X_t$ be a martingale satisfying $\Pr(X_{t+1} < X_t^{b} | X_{t}) \geq \eta$, where $X_0 \in (0, 1)$. Then
	\begin{equation*}
		\P(\log X_T > (\log X_0 + CT) b^{(1-\varepsilon)\eta T}) < \exp(-\Omega( \varepsilon^2 \eta T))
	\end{equation*}
\end{lemma}
\begin{proof}
	As in the proof of \cref{lem:strong-polarization}, let us consider random variables $Y_{t+1} := \log (X_{t+1}/X_t)$. This sequence of random variables is adapted to the sequence $X_t$ in the sense of \cref{def:adapted-sequence}. Let us decompose $Y_t = Y_t^+ + Y_t^-$, where $Y_t^+ = \max(Y_t,0)$.
	Note that by Markov inequality 
	\begin{equation*}
		\P\parens{Y_{t+1} > \lambda | X_{[1:t]}} = \P\parens{X_{t+1} > X_t \exp(\lambda) | X_{[1:t]}} \leq \exp(-\lambda)\frac{\E[X_{t+1} | X_{[1:t]}]}{X_t} = \exp(-\lambda).
	\end{equation*}
	
	By \cref{lem:sum-of-subexp} we deduce that for some $C$, we have
	\begin{equation*}
		\P\parens{\sum_{i \leq T} Y_i^{+} > C T} \leq \exp\parens{-\Omega(T)}.
	\end{equation*}
	On the other hand, if we take $Z_t$ to be the indicator variable for an event $X_t < X_{t-1}^{b}$. Note that the sequence $z_t$ is adapted to the sequence $X_t$. By \cref{lem:sum-of-events} we have
	\begin{equation*}
		\P\parens{\sum_{i \leq T} Z_i \leq (1-\varepsilon) \eta T} \leq \exp\parens{-\Omega(T\varepsilon^2\eta)}.
	\end{equation*}
	If neither of these unlikely events hold, that is we simultaneously have $\sum_{i \leq T} Y_i^{+} \le CT$ and $\sum_{i \leq T} Z_i > (1-\varepsilon) \eta T$, we can deduce that $\log X_T \leq (\log X_0 + CT) b^{(1 - \varepsilon) \eta T}$ --- i.e. the largest possible value of $X_T$ is obtained if all the initial $Y_i$ were positive and added up to $CT$ (at which point value of the martingale would satisfy $\log X_{T'} \leq \log X_0 + CT$), followed by $(1-\varepsilon) \eta T$ steps indicated by variables $Z_i$ --- for each of those steps, $\log X_{t+1} \leq b \log X_t$.
\end{proof}

We are now ready to prove the analog of \cref{lem:strong-polarization} for exponentially strong polarization:
\begin{lemma}
	\label{lem:second-phase}
	For all $0<\eta<1, b\ge 1, 0<\varepsilon<1$ 
	the following holds. Let $X_t$ be a martingale with values in $(0, 1)$ satisfying $\Pr(X_{t+1} < X_t^{b} | X_t) \geq \eta$, where $X_0 < \exp(- \gamma T)$ for some $\gamma>0$, then
	\begin{equation*}
		\P(\log X_T \ge -b^{(1 - \varepsilon) \eta T}) < \exp(-\Omega_{\varepsilon,\eta,\gamma}(T)))
	\end{equation*}
\end{lemma}
\begin{proof}
	Consider sequence $t_0, t_1, \ldots t_m \in [T]$, where $t_0 = 0, t_m = T$, and $\frac{\gamma T}{C} \leq |t_i - t_{i-1}| \leq \frac{\gamma T }{2C}$, and therefore $m = \Oh(C \gamma^{-1})$, where $C$ is a constant appearing in the statement of \cref{lem:single-piece-suction}. For each index $s \in [m]$ we consider a martingale $X^{(s)}_i := X_{t_s + i}$, and we will apply \cref{lem:single-piece-suction} 
	to this martingale $X^{(s)}$, with $T = t_{s+1} - t_s$. We can union bound total failure probability by $m \exp(-\Omega(\gamma \varepsilon^2 \eta T))$, which is upper bounded by the claim bound of $\exp(-\Omega_{\varepsilon,\eta,\gamma}(T)))$. 
	
	In case we succeed, we can deduce that for each $i$ we have
	\begin{equation}
		\label{eq:lucky-event}
		\log X_{t_i} < (\log X_{t_{i-1}} + C(t_i - t_{i-1})) b^{(1-\varepsilon)\eta (t_{i} - t_{i-1})}. 
	\end{equation}
	We will show that by our choice of parameters, we can bound $C(t_i - t_{i-1}) \leq -\frac{1}{2} \log X_{t_i}$. Let us first discuss how this is enough to complete the proof. Indeed, in such a case we have
	\begin{equation}
		\log X_{t_i} < \frac{1}{2} (\log X_{t_{i-1}}) b^{(1 - \varepsilon) \eta (t_{i} - t_{i-1})}, 
		\label{eq:inductive-step}
	\end{equation}
	and by induction
	\begin{equation*}
		\log X_{t_m} < \frac{1}{2^m} (\log X_{0}) b^{(1-\varepsilon) \eta t_m}.
	\end{equation*}
	For fixed $\eta, m$ and $T$ large enough (depending on $\eta, m, \varepsilon$), this yields $\log X_T < - b^{(1 - 2\varepsilon) \eta T}$, and the result follows up by changing $\varepsilon$ by a factor of $2$.
	
	All we need to do is to show is that for every $i$ we have 
	\begin{equation}
		C(t_{i+1} - t_i) \leq -\frac{1}{2} \log X_{t_i},
		\label{eq:inductive-condition}
	\end{equation}
	assuming that \cref{eq:lucky-event} holds for every $i$.
	We will show this inductively, together with $\log X_{t_i} \leq - \gamma T$. Note that we assumed this inequality to be true for $X_{t_0} = X_0$. By our choice of parameters we have $C(t_{i+1} - t_i) \leq \frac{\gamma T}{2}$, therefore for $t_{i+1}$ the inequality~(\ref{eq:inductive-condition}) is satisfied. 
	
	We will now show that $\log X_{t_{i+1}} \leq \log X_{t_i} \leq - \gamma T$ to finish the proof by induction. We can apply \cref{eq:inductive-step} to $X_{t_i}$, to deduce that $\log X_{t_{i+1}} \leq \frac{1}{2} (\log X_{t_i}) b^{\frac{1}{2} \frac{\gamma}{C} T}$ (which is true since $b\ge 1, \eta\le 1, \epsilon\ge 0$). This for large values of $T$ (given parameters $b, \gamma$ and $C$) yields $\log X_{t_{i+1}} < \log X_{t_i}$ --- indeed this inequality will be true as soon as $b^{\frac{\gamma}{2C} T} > 2$, because both $\log X_{t_{i+1}}$ and $\log X_{t_i}$ are negative, which completes the proof.
\end{proof}

We are now ready to prove local polarization to global polarization theorem for exponential polarization.

\begin{proof}[Proof of \cref{thm:local-to-global-exp}]
	Consider exponentially locally polarizing martingale, and let us fix some $\varepsilon > 0$. By \cref{cor:modest-polarization} with $t = 2 \varepsilon T$ and $\lambda = 1$ we deduce that for some $\nu>0$ we have
	\begin{equation*}
		\P\parens{\max\parens{X_{\varepsilon T}, 1 - X_{\varepsilon T}} \geq (1 - \frac{\nu}{2})^{2\varepsilon T}} < \exp\parens{-\Omega_{\varepsilon,\nu}(T)}.
	\end{equation*}
	
	We condition on $\max(X_{\varepsilon T}, 1 - X_{\varepsilon T}) < (1 - \frac{\nu}{2})^{2\varepsilon T}$. 
	Now let $K$ be a large enough constant depending on $\alpha$ and $\gamma$, the target rate of polarization in the high end.
	Now let 
	$\tau>0$  be such that $\tau\le \min\parens{\tau\parens{e^K},\tau_0}$, where $\tau_0$ is given by the definition of suction at the low end and $\tau(\cdot)$ is from the suction at the high end. Note that this implies that (1)  if $X_t < \tau$, we have 
	\begin{equation}
		\label{eq:exp-low-suction-end-def-in-proof}
		\P\parens{X_{t+1} < X_t^b|X_t} \geq \eta,
	\end{equation} 
	which holds since $\tau \leq \tau_0$ and (2) if $1 - X_t < \tau$, we have  \begin{equation}
		\label{eq:high-suction-end-def-in-proof}
		\P((1 - X_{t+1}) < \exp(-K) (1 - X_t)|X_t) \geq \alpha,
	\end{equation} 
	which follows from the condition on suction at the high end.
	By Doob's martingale inequality (specifically \cref{cor:doobs}), we deduce that $\P(\max_{t \in [\varepsilon T, T]} \max(X_{t}, 1 - X_t) > \tau) \leq \tau^{-1} (1 - \frac{\nu}{2})^{-2\varepsilon T} \leq \exp(-\Omega_{\tau,\nu,\varepsilon}(T))$. Let us now condition in turn on this event not happening.
	
	We will consider first the case when $X_{\varepsilon T} < (1 - \frac{\nu}{2})^{2\varepsilon T}$, and let us put $\gamma_0 := - 2 \varepsilon \log (1 - \frac{\nu}{2})$ (note that $\gamma_0>0$), so that $X_{\varepsilon T} < \exp(- \gamma_0 T)$.
	We can now apply \cref{lem:second-phase} to the martingale sequence starting with $X_{\varepsilon T}$. (Note that the assumptions of \cref{lem:second-phase}  are satisfied as long as $X_t$ for $t\in [\varepsilon T,T]$ stays bounded by $\tau$ due to \cref{eq:exp-low-suction-end-def-in-proof}.)
	Hence we deduce that in this case, except with probability $\exp(-\Omega_{\gamma_0, \varepsilon,\eta}(T))\le \exp(-\Omega_{\nu, \varepsilon,\eta}(T))$, we have
	\begin{equation*}
		\log X_{T} < -b^{(1-\varepsilon)^2 \eta T},
	\end{equation*}
	and therefore $X_T < 2^{-b^{(1-\varepsilon)^2 \eta T}}$. Note that this implies that $\Lambda=\log_2\parens{b^{(1-\varepsilon)^2 \eta}}=(1-\varepsilon)^2 \eta\log_2{b}$ (hence for any $\Lambda< \eta\log_2{b}$, we pick $\varepsilon$ appropriately). We also pick $\beta$ and $\eta$ such that $\beta\eta^T\ge \exp\parens{-\Omega_{\varepsilon,\eta,\mu,K}(T)}$.
	
	On the other hand, if $1 - X_{t} < \tau$ for all $\varepsilon T \leq t \leq T$, the suction at the high end condition of local polarization applies (i.e. \cref{eq:high-suction-end-def-in-proof} holds), and we can apply \cref{lem:strong-polarization} (we pick $K$ large enough so that $K\alpha >c$) to martingale $\tilde{X}_t \triangleq 1 - X_{\varepsilon T + t}$ to deduce that except with probability $\exp(-\Omega_{\alpha}(T))$, we have $1 - X_T < \exp(-\alpha K (1-\varepsilon T) / 4) < \gamma^T$ for suitable choice of $K$ depending on $\gamma$ and $\alpha$. Finally, we pick $\beta$ and $\eta$ such that $\beta\eta^T\ge \exp\parens{-\Omega_{\varepsilon,\eta,\mu,K,\alpha}(T)}$.
\end{proof}

\section{\Arikan~Martingale and its local polarization}
\label{sec:martingale}

In this section we formally describe the \Arikan~martingale associated with an invertible matrix $M \in \F_q^{k \times k}$
and a channel $\C_{Y|Z}$.

Before we proceed with the formal definition, to provide overview of the goals of this construction, we shall briefly point out its main features for the special case of \Arikan~martingale $\{X_t\}_{t=0}^{\infty}$ associated with an additive channel $\mathcal{C}$ --- where channel output $Y = Z + U$, with $U$ being some random variable in $\F_q$ not depending on $Z$.
\begin{enumerate}
\item
For given $t$, marginal distribution $X_t$ is distributed identically as $\bH( (\bvec{U} M^{\otimes t})_i | (\bvec{U} M^{\otimes t})_{<i})$ for uniformly random index $i$, where $\bvec{U}$ is a vector of $k^t$ i.i.d. random variables distributed as the error $U$.
\item
 Sequence $X_t$ is a martingale --- in particular we provide coupling of the distributions above over different $t$ in a non-trivial way.
\item
 Definition of the martingale $X_t$ is ``local'' in some sense, which makes it manageable to analyze how $X_t$ and $X_{t+1}$ are related and eventually show local polarization.
\end{enumerate}

In \cref{sec:correspondence} we elaborate on the connection of the \Arikan~martingale with polar codes --- specifically the main link is a more general version of the first property for all symmetric channels, and is proved as~\cref{lem:mart-code}.

Briefly, the \Arikan~martingale measures at time $t$, the distribution of conditional entropy of a random variable $\bvec{A'}_{\bvec{i}}$, conditioned on the values of a vector of variables $\bvec{B'}$ and on the values of $\bvec{A'}_{\bvec{j}}$ for $\bvec{j}$ smaller (according to $\prec$) than $\bvec{i}$ for a random choice of the index $\bvec{i}$. Here $\bvec{A'}$ is a vector of $k^t$ random variables taking values in $\F_q$ while
$\bvec{B'} \in \mathcal{Y}^{k^t}$.
The exact construction of the joint distribution of these $2k^t$ variables is the essence of the \Arikan\ construction of codes, and we describe it shortly. The hope with this construction is  that eventually (for large values of $t$) the conditional entropies are either very close to $0$, or very close to $\log q$ for most choices of $\bvec{i}$.

When $t=1$, the process starts with $k$ independent and identical pairs of variables
$\{(A_i,B_i)\}_{i\in [k]}$ where $A_i \sim \F_q$ and $B_i \sim \C_{Y|Z=A_i}$. (So each pair corresponds to an independent input/output pair from transmission of a uniformly random input over the channel $\C_{Y|Z}$.) Let $\bvec{A} = (A_1,\ldots,A_k)$ and $\bvec{B'} = (B_1,\ldots,B_k)$, and note that the conditional entropies $H(A_i | \bvec{A}_{\prec i},\bvec{B'})$ are all equal, and this entropy, divided by $\log_2 q$, will be our value of $X_0$. On the other hand, if we now let
$\bvec{A'} = \bvec{A}\cdot M$ then the conditional entropies $H(\bvec{A'}_i | \bvec{A'}_{\prec i},\bvec{B'})$ are no longer equal (for most, and in particular for all mixing, matrices $M$). On the other hand, conservation of conditional entropy on application of an invertible transformation tells us that $\E_{i \sim [k]}[ H(\bvec{A'}_i | \bvec{A'}_{\prec i},\bvec{B'})/\log_2 q ] = X_0$. Thus letting $X_1 =
H(\bvec{A'}_i | \bvec{A'}_{\prec i},\bvec{B'})/\log_2 q$ (for random $i$) gives us the martingale at time $t=1$.

While this one step of multiplication by $M$ {\em differentiates} among the $k$ (previously identical) random variables, it doesn't yet polarize. The hope is by iterating this process one can get polarization\footnote{In the context of Polar coding, {\em differentiation} and {\em polarization} are good events, and hence our `hope.'}. But to get there we need to describe how to iterate this process. This iteration is conceptually simple (though notationally still complex) and illustrated in
Figure~\ref{fig:polarization}.
Roughly the idea is that at the beginning of stage $t$, we have defined a joint distribution of $k^t$ dimensional vectors
$(\bvec{A},\bvec{B})$ along with a multi-index $\bvec{i}\in [k]^t$. We now sample  $k$ independent and identically distributed pairs of these random variables  $\{(\bvec{A}^{(\ell)},\bvec{B}^{(\ell)})\}_{\ell \in [k]}$ and view $(\bvec{A}^{(\ell)})_{\ell\in [k]}$ as a $k^t \times k$ matrix which we multiply by $M$ to get a new $k^t \times k$ matrix. Flattening this matrix into a $k^{t+1}$-dimensional vector gives us a sample from the distribution of $\bvec{A'} \in \F_q^{k^{t+1}}$. $\bvec{B'}$ is simply the concatenation of all the
vectors $(\bvec{B}^{(\ell)})_{\ell \in [k]}$. And finally the new index $\bvec{j}\in [k]^{t+1}$ is simply obtained by extending $\bvec{i} \in [k]^t$ with a $(t+1)$th coordinate distributed uniformly at random in $[k]$. $X_{t+1}$ is now defined to be
$\bH(\bvec{A'}_\bvec{j} | \bvec{A'}_{\prec \bvec{j}},\bvec{B'})$, where $\bH(\cdot)$ is the normalized $q$-ary entropy defined in \eqref{eq:normalized-entropy}. The formal description is below.

\cstate{\Arikan\ martingale}{definition}{defarikanmartingale}{
	\label{def:arikan-martingale}
	Given an invertible matrix $M \in \F_q^{k\times k}$ and a channel description $C_{Y|Z}$ for $Z \in \F_q, Y \in \mathcal{Y}$, the \Arikan-martingale $X_0, \ldots X_t, \ldots$ associated with it is defined as follows.
	For every $t \in \mathbb{N}$,
	let $D_t$ be the distribution on pairs $\F_q^{k^t} \times \mathcal{Y}^{k^t}$ described inductively below:
	
	A sample $(A,B)$ from $D_0$ supported on $\F_q \times \mathcal{Y}$ is obtained by sampling $A \sim \F_q$, and $B \sim C_{Y|Z=A}$.
	For $t \geq 0$, a sample $(\bvec{A}', \bvec{B}') \sim D_{t+1}$ supported on $\F_q^{k^{t+1}}\times \mathcal{Y}^{k^{t+1}}$ is obtained as follows:
	\begin{itemize}
		\item Draw $k$ independent samples $(\bvec{A}^{(1)}, \bvec{B}^{(1)}), \dots, (\bvec{A}^{(k)}, \bvec{B}^{(k)}) \sim D_{t}$.
		\item Let $\bvec{A}'$ be given by
		$\bvec{A}'_{[\bvec{i}, \cdot]}
		= (\bvec{A}^{(1)}_{\bvec{i} } ~ ,\dots, ~ \bvec{A}^{(k)}_{\bvec{i}})\cdot  M$
		for all $\bvec{i}\in [k]^{t}$
		and $\bvec{B}' = (\bvec{B}^{(1)}, \bvec{B}^{(2)}, \ldots \bvec{B}^{(k)})$.
	\end{itemize}
	
	Then, the sequence $X_t$ is defined as follows:
	Sample $i_l \in [k]$ uniformly and independently for $l=1,2,\dots,t$.
	Let $\bvec j = (i_1,\ldots,i_t)$
	and let
	$X_t := \bH(\bvec{A}_{\bvec j} | \bvec{A}_{\prec \bvec j}, \bvec{B})$,
	where the entropies are with respect to the distribution
	$(\bvec{A}, \bvec{B}) \sim D_t$.\footnote{We stress that the only randomness in the evolution of $X_t$ is in the choice of $i_1,\ldots,i_t,\ldots$.
		The process of sampling
		$\bvec{A}$ and $\bvec{B}$ is only used to define
		the distributions
		for which we consider
		the conditional entropies $H(\bvec{A}_{\bvec j} | \bvec{A}_{\prec \bvec j}, \bvec{B})$.}
}
\begin{figure}[h!]
	\begin{center}
		\newcommand{\mev}{\pgfmathsetmacro}
\begin{tikzpicture}
\iffalse
\draw [help lines, color=black!20!white] (-\d, -\d) grid (\d, \d);
\foreach \i in {-\d,...,\d}{
	\ifthenelse{\not \i = 0}{
		\node [scale=0.6, color=black!20!white] at (\i, -0.1) {\i};
		\node [scale=0.6, color=black!20!white] at (-0.1, \i) {\i};
	}
};
\node [scale = 0.6, color=black!20!white] at (-0.1, -0.1) {0};
\fi

\def\k{3};
\def\s{5};
\def\channelh{0.25}
\def\channels{0.15}
\def\blockspace{0.6}
\def\steph{(\channelh + \channels)}

\foreach \i in {1, ..., \k} {
	\mev{\bs}{4 - (\i - 1)*\s*\steph - (\i - 1)*\blockspace};
	\mev{\top}{\bs - \s*\steph + \channels}
	% The $M^{\otimes t}$ boxes
	\draw (-3, \bs) rectangle (-1.8, \top);
	\mev{\mid}{(\bs + \top)/2};
	\node  (M\i) at (-2.4, \mid) {$M^{\otimes t}$};
	\node [scale=0.7]  at (-1.5, \mid) {$A^{(\i)}_{*}$};

	\foreach \j in {1, ..., \s} {
		% Channel boxes
		\mev{\pos}{\bs - (\j - 1) * \steph}
		\draw (-4.8, \pos) rectangle (-4, \pos - \channelh);
		\draw [->] (-4, \pos - \channelh/2) -- (-3, \pos - \channelh/2);
		\draw [->] (-5, \pos - \channelh/2) -> (-4.8, \pos - \channelh/2);
		\node [scale = 0.5] at (-4.4, \pos - \channelh/2) {$C_{Y|Z}$};
		\coordinate (Mt{\i}x{\j}) at (-1.8, \pos - \channelh/2);
	};
};

% M boxes
	
\mev{\totalheight}{\k*\steph*\s + (\k - 1)*\blockspace - \channels};
\def\aratio{0.3}
% gap = box * ratio
% (s-1)*gap + s*box = totalheight 
\mev{\boxh}{\totalheight / ( (\s-1)*\aratio + \s) }
\mev{\gap}{\boxh*\aratio};
\def\ledge{2}
\def\lwid{1}

\foreach \j in {1, ..., \s}{
	\mev{\pos}{4 - (\boxh + \gap)*(\j - 1)};
	\pgfmathtruncatemacro{\smj}{\s - \j};
	\ifthenelse{\j < 3 \OR \smj < 1}{
		\draw (\ledge, \pos) rectangle (\ledge + \lwid, \pos - \boxh);
		\node at (\ledge + \lwid/2, \pos - \boxh/2) {M};
		\foreach \i in {1, ..., \k}{
			\mev{\inpos}{\pos - (\i - 1)*\boxh/\k - \channelh/2};
			\coordinate (M{\i}x{\j}) at (\ledge, \inpos);
			\ifthenelse{\j = 1}{
				\draw [line width=1.5, color=red, ->] (Mt{\i}x{\j}) .. controls +(0.6, 0) and +(-0.6, 0) .. (M{\i}x{\j}) node [pos=0.1, above,scale=0.5] {$A^{(\i)}_{\vec{1}}$};
				
				\node [scale=0.5] at (\ledge + \lwid + 0.45, \inpos - 0.02) {$A'_{[\vec{1}, \i]}$};
			}{
				\draw [dashed, ->] (Mt{\i}x{\j}) .. controls +(1, 0) and +(-1, 0) .. (M{\i}x{\j});
			};
				
			\draw [->] (\ledge + \lwid, \inpos) -> (\ledge + \lwid + 0.2, \inpos);
			
		};
	}{ \draw [fill] (\ledge + \lwid/2, \pos - \boxh/2) circle (0.05);};
};
\end{tikzpicture}
	\end{center}
	\caption{Evolution of \Arikan\ martingale for $3 \times 3$ matrix $M$.}
	\label{fig:polarization}
\end{figure}

Figure~\ref{fig:polarization} illustrates the definition by highlighting the construction of the vector $\bvec{A'}$, and in particular highlights the recursive nature of the construction.

It is easy (and indeed no different than in the case $t=1$) to show that $\E[X_{t+1}|X_t] = X_t$ and so the \Arikan\ martingale is indeed a martingale.
This is shown below.

\cstate{}{proposition}{propactuallymartingale}{
	\label{prop:actually-martingale}
	For every matrix $M$ and channel $\C_{Y|Z}$, the \Arikan\ martingale is a martingale and in particular a $[0,1]$-martingale.
}

\begin{proof}
	The fact that $X_t \in [0,1]$ follows from the fact for $0 \leq H(\bvec{A}_{\bvec{i}} | \bvec{A}_{\prec \bvec{i}}, \bvec{B}) \leq H(\bvec{A}_{\bvec{i}}) \leq \log_2 q$ (the upper bound follows since $\bvec{A}_{\prec \bvec{i}}\in \F_q$)
	and so $0 \leq X_t = H(\bvec A_{\bvec{i}} | \bvec{A}_{\prec \bvec{i}},\bvec{B})/\log_2 q \leq 1$.
	
	We turn to showing that $\E[X_{t + 1}| X_t = a] = a$. To this end, consider a sequence of indices $\bvec{i} = (i_1, \ldots i_{t})$, such that $\bH(\bvec A_{\bvec{i}} ~|~ \bvec{A}_{\prec \bvec{i}}, \bvec{B}) = a$. We wish to show that $\E_{i_{t+1} \sim [k]} [\bH(\bvec A'_{[\bvec{i}, i_{t+1}]} ~|~ \bvec{A}'_{\prec [\bvec{i}, i_{t+1}]}, \bvec{B}')] = a$.
	
	Since the pairs $(\bvec{A}^{(s)},\bvec{B}^{(s)})$ are independent samples from $D_t$, note that for any $s$, we have $\bH(\bvec A_{\bvec{i}}^{(s)} ~|~ \bvec{A}_{\prec \bvec i}^{(s)}, \bvec{B}^{(s)}) = a$. Furthermore, because of the same independence, 
	we have
	$$\bH(\bvec A^{(s)}_{\bvec{i}} ~|~ \bvec{A}^{(s)}_{\prec \bvec{i}}, \bvec{B}^{(s)})
	= \bH(\bvec A^{(s)}_{\bvec{i}} ~|~ \cup_{j \in [k]} \bvec{A}^{(j)}_{\prec \bvec{i}}, \cup_{j \in [k]} \bvec{B}^{(j)})$$
	$$ \mbox{ and }\bH(\bvec A^{(1)}_{\bvec{i}}, \ldots, \bvec A^{(k)}_{\bvec{i}} | \cup_{j \in [k]} \bvec{A}^{(j)}_{\prec \bvec{i}}, \cup_{j \in [k]} \bvec{B}^{(j)}) = k\cdot a.$$
	By the invertibility of $M$ we have
	$$\bH(\bvec A'_{[\bvec{i}, 1]}, \ldots \bvec A'_{[\bvec{i}, k]} ~|~ \cup_{j \in [k]} \bvec{A}^{(j)}_{\prec \bvec{i}}, \cup_{j \in [k]} \bvec{B}^{(j)}) = \bH(\bvec A^{(1)}_{\bvec{i}}, \ldots, \bvec A^{(k)}_{\bvec{i}} | \cup_{j \in [k]} \bvec{A}^{(j)}_{\prec \bvec{i}}, \cup_{j \in [k]} \bvec{B}^{(j)}) = k\cdot a.$$ We can apply again invertibility of the matrix $M$ to deduce that conditioning on $\cup_{j \in [k]} \bvec{A}_{\prec \bvec{i}}^{(j)}$ is the same as conditioning on $\bvec{A'}_{\prec [\bvec{i}, 1]}$ --- i.e. for any multiindex $\bvec{i'} \prec \bvec{i}$ variables $\bvec{A}_{\bvec{i'}}^{(1)}, \ldots \bvec{A}_{\bvec{i'}}^{(k)}$ and
	$\bvec{A'}_{[\bvec i', 1]}, \ldots \bvec{A'}_{[\bvec i', k]}$ are related via invertible transform $M$.
	This yields 
	\[\bH(\bvec{A'}_{[\bvec{i}, 1]}, \ldots \bvec{A'}_{[\bvec{i},k]} ~|~ \bvec{A'}_{\prec [\bvec{i}, 1]}, B') = \bH(\bvec{A'}_{[\bvec{i}, 1]}, \ldots \bvec{A'}_{[\bvec{i}, k]} ~|~ \cup_{j \in [k]} \bvec{A}^{(j)}_{\prec \bvec{i}}, \cup_{j \in [k]} \bvec{B}^{(j)})  = k a.\]
	
	Finally by the Chain rule of entropy we have
	\begin{eqnarray*}
		\bH(\bvec{A'}_{[\bvec{i}, 1]}, \ldots \bvec{A'}_{[\bvec{i},k]} ~|~ \bvec{A'}_{\prec [\bvec{i}, 1]}, \bvec B')
		& = & \sum_{i_{t+1} = 1}^k \bH(\bvec{A'}_{[\bvec{i}, i_{t+1}]} ~|~ \bvec{A}'_{[\bvec{i}, <i_{t+1}]}, \bvec{A'}_{\prec [\bvec{i}, 1]}, \bvec B') \\
		& = &
		\sum_{i_{t+1} = 1}^k \bH(\bvec{A'}_{[\bvec{i}, i_{t+1}]} ~|~ \bvec{A}'_{\prec [\bvec{i}, i_{t+1}]}, \bvec B')
	\end{eqnarray*}
	Putting these together, we have $\E[X_{t + 1} | X_t = a] = \E_{i_{t+1}} [ \bH(\bvec{A'}_{[\bvec{i}, i_{t+1}]} ~|~ \bvec{A'}_{\prec [\bvec{i}, i_{t+1}]}, \bvec B') ] = \frac 1k\cdot ka=a$.
\end{proof}

Finally, we remark that based on the construction it is not too hard to see that if $M$ were an identity matrix, or more generally a non-mixing matrix, then $X_t$ would deterministically equal $X_0$. (There is no differentiation and thus no polarization.) The thrust of this paper is to show that in all other cases we have strong polarization.

\subsection{Matrix Polarization and the \Arikan\ martingale}

Note that the definition of the \Arikan\ martingale is itself complex, and in particular the distribution of $X_t$, the variable at the $t$th step, needs a description whose complexity grows with $t$. The essence of the polarization argument does not depend on this intricacy of the definition, most of which can be abstracted away. Indeed we do so formally by considering a simpler (single step) randomization process associated with a matrix $M$.
We define a matrix $M$ to be {\em polarizing} if this single step process satisfies properties similar to those of local polarization (see \cref{def:matrix-polarization}). Then, in \cref{lem:polarizing-matrix-implies-martingale} we show that if a matrix $M$ satisfies matrix polarization then for every channel $\C$ the \Arikan\ martingale associated with $M$ and $\C$ is locally polarizing. This is a notationally heavy but conceptually light proof, whose essence is to verify that certain variables are independent and so conditioning on such variables does not change entropies. This will allow us focus on a simpler single step process in future sections to prove (exponential) local polarization.

We start with the definition of matrix polarization.

\begin{definition}[Matrix (exponential) polarization]
	\label{def:matrix-polarization}
	We say that a matrix $M \in \F_q^{k\times k}$ satisfies matrix polarization, if and only if for every pair of random variables $(\bvec{U}, W)$, such that  $\bvec{U} = (\bvec{U}_1,\dots,\bvec{U}_k) \in \F_q^k$, $W = (W_1,\ldots,W_k)$ is supported on some finite set, and the pairs  $(\bvec{U}_i,W_i)$ are independently and identically distributed for $i \in [k]$, the vector $\bvec{V} := \bvec{U} \cdot M$ satisfies the following conditions:
	\begin{enumerate}
		\item
		{\bf (Variance in the middle):}
		There is some index $j \in [k]$
		for which the following holds:
		For every $\tau > 0$, there exists $\eps=\eps(\tau) > 0$ such that
		if $\bH(\bvec{U}_1 | W) \in (\tau, 1-\tau)$,
		then
		$$\bH((\bvec{V})_j | \bvec{V}_{< j}, W)  \geq \bH(\bvec{U}_1 | W) + \eps.$$
		
		\item
		{\bf (Suction at the lower end):}
		There is some index $j \in [k]$
		for which the following holds:
		
		For every $c < \infty$, there  exists $\tau > 0$ such that
		if
		$\bH(\bvec{U}_1 | W) < \tau$
		then
		$$\bH(\bvec{V}_j | \bvec{V}_{< j}, W)  \leq \frac{1}{c}~\bH(\bvec{U}_1 | W).$$

		\item
		{\bf (Suction at the high end):}
		Analogously to suction at the low end, there is some index $j \in [k]$
		for which the following holds:
		
		For every $c < \infty$, there  exists $\tau > 0$ such that
		if
		$\bH(\bvec{U}_1 | W) > 1-\tau$
		then
		$$1 - \bH(\bvec{V}_j | \bvec{V}_{< j}, W)
		\leq \frac{1}{c}(1-\bH(\bvec{U}_1 | W)).$$
	\end{enumerate}
	Additionally, we say that $M$ satisfies $(\eta,b)$-exponential matrix polarization if we have the following property:
	\begin{enumerate}
		\item[2'.] {\bf (Strong Suction at the lower end):}
		There  exists $\tau > 0$ such that
		if
		$\bH(\bvec{U}_1 | W) < \tau$ then for at least $\eta$ fraction of the indices $j \in [k]$ we have: 
		then
		$$\bH(\bvec{V}_j | \bvec{V}_{< j}, W)  \leq \bH(\bvec{U}_1 | W)^b .$$
	\end{enumerate}
\end{definition}

Thus the notion of matrix polarization is somewhat more general than polarization of the corresponding \Arikan\ martingale. 
\begin{enumerate}
	\item[(1)] In the latter, the conditioning in the entropy prescribes some specific relations between $\bvec{U}$ and $W$, where in the former $W$ is arbitrary (subject to the condition that the pairs $(\bvec{U}_j,W_j)$ are i.i.d.).
	\item[(2)] Furthermore the definitions also make slight changes to the conditions of Variance in the middle and suction only requiring the existence of $j \in [k]$ having a certain property as opposed requiring that a random choice of $j$ satisfy some condition. 
\end{enumerate} 

The differences in (1) above allows for cleaner proofs, since the specific structure of $W$ is not needed. The class of differences in (2) above changes some probabilities and/or variances by factors depending on $k$, but this difference is negligible. In the following section we formally confirm that matrix polarization is a sufficient condition for martingale polarization.

\subsection{Matrix polarization implies local polarization of \Arikan~martingale}
\label{ssec:proof-local}

In  this section we prove that matrix polarization implies local polarization of the corresponding \Arikan\ martingale.

\begin{theorem}
	\label{lem:polarizing-matrix-implies-martingale}
	For every  matrix $M\in \F_q^{k \times k}$ and every symmetric memoryless channel $\C_{Y|Z}$, if $M$ satisfies matrix polarization then the \Arikan~martingale associated with $M$ and $\C$ is satisfies local polarization. Furthermore if $M$ satisfies $(\eta,b)$-exponential matrix polarization, then the \Arikan-martingale satisfies $(\eta,b)$-exponential local polarization. 
\end{theorem}

We begin with a lemma that will be useful in the proof of \cref{lem:polarizing-matrix-implies-martingale}:
\begin{lemma}
	\label{lem:pushing-back-cond}
	Let $\bvec{A}^{(1)}, \ldots \bvec{A}^{(k)}$, and $\bvec{A}'$ be defined as in \cref{def:arikan-martingale}, and let $V, W$ be arbitrary random variables. Then for any multiindex $\bvec{i} \in [k]^{t}$ and any $i_{t+1} \in [k]$ we have
	\begin{equation*}
		\bH(V ~|~ \bvec{A}'_{\prec [\bvec{i}, i_{t+1}]}, W)
		= \bH(V ~|~ \bvec{A}'_{[\bvec{i}, < i_{t+1}]}, \bvec{A}^{(1)}_{\prec \bvec{i}}, \bvec{A}^{(2)}_{\prec \bvec{i}},\ldots \bvec{A}^{(k)}_{\prec \bvec{i}}, W) \ .
	\end{equation*}
\end{lemma}
\begin{proof}
	Observe first that by definition of the order $\prec$ we have that $\bvec{A}'_{\prec [\bvec{i}, i_{t+1}]} = (\bvec{A}'_{\prec [\bvec{i}, 1]}, \bvec{A}'_{[\bvec{i}, < i_{t+1}]})$, hence
	\begin{equation*}
		\bH(V ~|~ \bvec{A}'_{\prec [\bvec{i}, i_{t+1}]}, W) =
		\bH(V ~|~ \bvec{A}'_{[\bvec{i}, < i_{t+1}]}, \bvec{A}'_{\prec [\bvec{i}, 1]}, W) \ .
	\end{equation*}
	The definition of the sequence $\bvec{A}'$ in terms of $\bvec{A}$ (in \cref{def:arikan-martingale}) reads
	\begin{equation*}
		\bvec{A}'_{[\bvec{j}, \cdot]} = (\bvec{A}_{\bvec{j}}^{(1)}, \cdots, \bvec{A}_{\bvec{j}}^{(k)}) M.
	\end{equation*}
	Note that if random variables $B, B'$ are related by invertible function $B = f(B')$, then $\bH(A | B) = \bH(A | B')$. By definition of mixing matrix, $M$ is invertible, and hence variables $\bvec{A}'_{\prec [\bvec{i},1]}$ and variables $\bvec{A}_{\prec \bvec{i}}^{(1)}, \ldots \bvec{A}_{\prec \bvec{i}}^{(k)}$ are indeed related by invertible (linear) transformation, which yields
	\begin{equation*}
		\bH(V ~|~ \bvec{A}'_{[\bvec{i}, < i_{t+1}]}, \bvec{A}'_{\prec [\bvec{i}, 1]},W)
		= \bH(V ~|~ \bvec{A}'_{[\bvec{i}, < i_{t+1}]}, \bvec{A}^{(1)}_{\prec \bvec{i}}, \bvec{A}^{(2)}_{\prec \bvec{i}},\ldots \bvec{A}^{(k)}_{\prec \bvec{i}}, W)
		\ .
		\qedhere
	\end{equation*}
\end{proof}

We now turn to the proof of \cref{lem:polarizing-matrix-implies-martingale}.

\begin{proof}[Proof of \cref{lem:polarizing-matrix-implies-martingale}]
	Fix a matrix $M$, channel $\C_{Y|Z}$ and time $t$. 
	We start by recalling the definition of the variables $X_t$ and $X_{t+1}$ in the definition of the \Arikan\ martingale, and also recall what local polarization entails for these variables.
	
	Let $(\bvec{A}, \bvec{B}), (\bvec{A}^{(1)}, \bvec{B}^{(1)}), \ldots (\bvec{A}^{(k)}, \bvec{B}^{(k)}) \sim D_t$ denote independent random variables. Let $(\bvec{A}', \bvec{B}')$ constructed from $(\bvec{A}^{(1)}, \bvec{B}^{(1)}), \ldots (\bvec{A}^{(k)}, \bvec{B}^{(k)})$ as in \cref{def:arikan-martingale}, i.e., we have
	$\bvec{A}'_{[\bvec{i'}, \cdot]}
	= (\bvec{A}^{(1)}_{\bvec{i'} } ~ ,\dots, ~ \bvec{A}^{(k)}_{\bvec{i'}})\cdot  M$ for every $\bvec{i'} \in [k]^t$
	and $\bvec{B}' = (\bvec{B}^{(1)}, \bvec{B}^{(2)}, \ldots \bvec{B}^{(k)})$.
	Now 
	let $\bvec{i} \triangleq (i_1, \ldots i_t)$ be sampled uniformly from $[k]^t$ and let $i_{t+1} \in [k]$ be chosen independently and uniformly.
	
	Then by \cref{def:arikan-martingale} we have: 
	\[X_t = \bH(\bvec{A}_{\bvec{i}} ~|~ \bvec{A}_{\prec \bvec{i}}, \bvec{B})\]
	\[\mbox{and~} X_{t+1} = \bH(\bvec{A}'_{[\bvec{i}, i_{t+1}]} | \bvec{A}'_{\prec [\bvec{i}, i_{t+1}]}, \bvec{B}').\]
	
	That is, for $\bvec{U}=\parens{\bvec{A}_{\bvec{i}}^{(1)},\dots,\bvec{A}_{\bvec{i}}^{(k)}}$, we have $\bvec{A'}_{[\bvec{i}, \cdot]} = \bvec{U} \cdot M$, and $\bvec{B'} = (\bvec{B}^{(1)}, \ldots, \bvec{B}^{(k)})$.  
	
	We will use the property of matrix polarization of $M$ with $\bvec{U} = (\bvec{U}_1,\ldots,\bvec{U}_k)$ where $\bvec{U}_s =  \bvec{A}^{(s)}_\bvec{i}$ and $W = (W_1,\ldots,W_k)$ where $W_s = (\bvec{A}^{(s)}_{\prec \bvec{i}}, \bvec{B}^{(s)})$ to deduce local polarization of \Arikan~martingale. 
	Note that the pairs $(\bvec{U}_1,W_1),\ldots,(\bvec{U}_k,W_k)$ are identically distributed and independent as required. 
	We let $\bvec{V} = \bvec{U}\cdot M$.  By the definition of \Arikan~martingale we have $\bvec{A'}_{[\bvec{i}, \cdot]} = \bvec{V} = \bvec{U} \cdot M$.
	Thus the matrix polarization of $M$ implies bounds on the conditional entropy of $\bH(\bvec{V}_j|\bvec{V}_{<j},W)$, where $\bvec{V}_j = \bvec{A}'_{[\bvec{i},j]}$ where $X_{t+1}$ studies conditional entropy of $(\bvec{V}_{i_{t+1}} | \bvec{A}'_{\prec[\bvec{i},i_{t+1}]},\bvec{B}')$. In what follows we verify that despite the difference the latter can be bounded as required for the condition of (exponential) local polarization of the \Arikan\ martingale. We tackle each of the conditions in order
	but first we note that by \cref{lem:pushing-back-cond} we have
	\begin{align}
		\bH(\bvec{A}'_{[\bvec{i}, j]} ~|~ \bvec{A}'_{\prec [\bvec{i}, j]}, \bvec{B}') &  = 
		\bH(\bvec{A}'_{[\bvec{i}, j]}  ~|~ \bvec{A}'_{[\bvec{i}, <j]}, \bvec{A}^{(1)}_{\prec \bvec{i}}, \ldots, \bvec{A}^{(k)}_{\prec \bvec{i}}, \bvec{B}'), \nonumber \\
		& = \bH( (\bvec{U} \cdot M)_{j} | (\bvec{U} \cdot M)_{<j}, \bvec{A}^{(1)}_{\prec \bvec{i}}, \ldots, \bvec{A}^{(k)}_{\prec \bvec{i}}, \bvec{B}' ) \nonumber \\
		& = \bH( \bvec{V}_{j} | \bvec{V}_{<j}, W ).
		\label{eq:entrop-eq}
	\end{align}
	
	Let $h \triangleq X_t = \bH(\bvec{A}_{\bvec{i}} ~|~ \bvec{A}_{\prec \bvec{i}}, \bvec{B})$. Note that for every $s \in [k]$ we have $\bH(\bvec{A}^{(s)}_{\bvec{i}} ~|~ \bvec{A}^{(s)}_{\prec \bvec{i}}, \bvec{B}^{(s)}) = h$, because all the pairs $(\bvec{A}^{(s)}, \bvec{B}^{(s)})$ are distributed independently and identically to $(\bvec{A},\bvec{B})$. Moreover for every $j\in[k]$ we have $\bH(\bvec{U}_j \, | \, W) = \bH(\bvec{A}_{\bvec{i}}^{(s)} \, | \, \bvec{A}^{(s)}_{\prec \bvec{i}}, \bvec{B}^{(s)}) = h$, where the first equality follows from the fact that pairs $(\bvec{A}^{(s)}, \bvec{B}^{(s)})_{s \in [k]}$ are identically and independently distributed (so the index $j$ does not matter).

	We will start with the \emph{Variance in the middle} condition of martingale local polarization (\cref{defn:polar-local}).  As a reminder, what we need to show is that if $h \in (\tau, 1-\tau)$, then 
	\[ \Var_{i_{t+1} \sim [k]} (\bH(\bvec{A}'_{[\bvec{i}, i_{t+1}]} | \bvec{A}'_{\prec [\bvec{i}, i_{t+1}]}, \bvec{B}') - \bH(\bvec{A}_{\bvec{i}} | \bvec{A}_{\prec \bvec{i}}, \bvec{B})) > \theta(\tau) \ . \]
	Note that by the martingale property (\cref{prop:actually-martingale}) we have $$\E_{i_{t+1}\sim [k]} [\bH(\bvec{A}'_{[\bvec{i}, i_{t+1}]} | \bvec{A}'_{\prec [\bvec{i}, i_{t+1}]}, \bvec{B}') - \bH(\bvec{A}_{\bvec{i}} | \bvec{A}_{\prec \bvec{i}}, \bvec{B})]=0.$$
	and as such to obtain the lower bound on the variance it is enough to show that
	\begin{equation}
		\P_{i_{t+1} \sim [k]}\left(\bH(\bvec{A}'_{[\bvec{i}, i_{t+1}]} | \bvec{A}'_{\prec [\bvec{i}, i_{t+1}]}, \bvec{B}') \geq h + \varepsilon(\tau)\right) \geq \frac{1}{k}.
		\label{eq:variation-goal}
	\end{equation}
	This would allow us to deduce that the variance above is lower bounded by $\varepsilon(\tau)^2/k$. (Note that this lower bound is true for every $h$ and hence the actual variance needed in the statement of the \emph{Variance in the middle} condition is also true.)
	
	We now use the \emph{Variance in the middle} condition of the matrix polarization (\cref{def:matrix-polarization}) for $M$ with variables $(\bvec{U}, W)$ . This condition asserts that for some index $j$ we have entropy gain $\bH( (\bvec{U} \cdot M)_j | (\bvec{U} \cdot M)_{<j}, W) \geq h + \varepsilon(\tau)$. 
	Combining this with \cref{eq:entrop-eq} proves inequality \cref{eq:variation-goal}, and therefore shows \emph{variation in the middle}  for the \Arikan\ martingale.
	
	Next we turn to the proof of \emph{suction at the upper end}. Here, we will show that for every $c$ if $1 - h < \tau(c)$, then with probability at least $\frac{1}{k}$ over the choice of $i_{t+1}$,  we will have $1 - \bH(\bvec{A}'_{[\bvec{i}, i_{t+1}]} | \bvec{A}'_{\prec [\bvec{i}, i_{t+1}]}, \bvec{B}') \le \frac{1}{c}(1 - h)$.
	The corresponding \emph{suction at the upper end} condition of matrix polarization asserts the existence of an index $j$, such that $1 - \bH( \bvec{V}_j | \bvec{V}_{<j}, W) \le \frac{1}{c}(1-h)$. With probability at least $\frac{1}{k}$ we have $i_{t+1} = j$, and in this case we have
	\begin{align*}
		1 - \bH(\bvec{A}'_{[\bvec{i}, i_{t+1}]} ~|~ \bvec{A}'_{\prec [\bvec{i}, j]}, \bvec{B}') &  = 1 - \bH( (\bvec{U} \cdot M)_{j} | (\bvec{U} \cdot M)_{<j}, W ) 
		\leq \frac{1}{c} (1 - h),
	\end{align*}
	where the first equality above is by \cref{eq:entrop-eq}.
	
	The proof of \emph{suction at the lower end} is symmetric. We now turn to the proof of \emph{strong suction at the low end} (\cref{defn:polar-local-exp}).
	Let $M$ satisfy $(\eta, b)$-exponential matrix polarization. Recall that we wish to show that $\P_{i_{t+1} \sim [k]}(\bH(\bvec{A'}_{ [\bvec{i}, i_{t+1}]} ~|~ \bvec{A'}_{\prec [\bvec{i}, i_{t+1}]}, \bvec{B'}) < h^b) \geq \eta$.
	Once again by \cref{eq:entrop-eq} we have, for every $j \in [k]$, $\bH(\bvec{A'}_{ [\bvec{i}, j]} ~|~ \bvec{A'}_{\prec [\bvec{i}, j]}, \bvec{B'})= \bH(\bvec{V}_j ~|~ \bvec{V}_{<j},W)$. This is exactly the property given by the \emph{strong suction at the low end} property of matrix polarization (using $h = \bH(\bvec{U}_1|W)$). 
	
	This concludes the proof.
\end{proof}

Thus to prove \cref{thm:triangle-local,thm:triangle-local-exp} we now need to prove that for every mixing matrix $M$, $M$ satisfies matrix polarization and $M^2$ satisfies exponential polarization. 
We argue the former in \cref{sec:matrix-polar} and the latter in \cref{sec:exp-local-arikan}.

\section{Proof of Matrix Polarization}\label{sec:matrix-polar}

In this section we prove that every mixing matrix satisfies matrix polarization, 
modulo some entropic inequalities whose proofs are deferred to
\cref{sec:entropic-proofs}. 
Combined with \cref{lem:polarizing-matrix-implies-martingale} this immediately yields \cref{thm:triangle-local} which asserts the local polarization of the \Arikan\ martingale. 

Informally, this section can be viewed as reducing matrix polarization of a general (mixing) matrix to the matrix polarization of the matrix 
$G_2 = \left(\begin{smallmatrix} 1 & 0 \\ 1 & 1 \end{smallmatrix}\right)$. Formally what we do is state three entropic inequalities
(see \cref{sec:entropic-statements}) 
that arise naturally in the proof of the matrix polarization of $G_2$. These inequalities relate the conditional entropy of a sum of two random variables to the entropy of each of those random variables. Indeed, these inequalities can be used to show immediately that $G_2$ satisfies matrix polarization, and we do so in \cref{lem:g2-polarizes}.
But the bulk of the work, and novelty, in this section is in \cref{sec:local-k-by-k} where we show (via carefully executed ``Gaussian elimination'') that these entropic inequalities suffice to show the matrix polarization of \emph{every} mixing matrix.

\subsection{Entropic Lemmas in the $2 \times 2$ Case}
\label{sec:entropic-statements}

We state here the three entropic lemmas. The proofs of the first two are deferred to \cref{sec:entropic-proofs}. The third lemma is well-known and we providee a reference for its proof.

The first lemma arises from the analysis of the \emph{suction at the upper end} (for $X_t > 1 - \tau$) condition of \cref{def:matrix-polarization}.
\begin{lemma}
	\label{lem:suc-upr-cond}
	For every finite field $\F_q$ and every $\gamma > 0$, there exist $\tau$, such that if $(X_1, A_1)$ and $(X_2, A_2)$ are independent random variables with $X_i \in \F_q$ such that 
	$1 - \bH(X_2 ~|~ A_2) \leq \tau$, then
	\begin{equation*}
		1 - \bH(X_1 + X_2 ~|~ A_1, A_2) \leq \gamma (1 - \bH(X_1 ~|~ A_1)).
	\end{equation*}
\end{lemma}

The next lemma comes analogously from the \emph{suction at the low end} (for $X_t < \tau$) condition of \cref{def:matrix-polarization}.
\begin{lemma}
	\label{lem:suc-lwr-cond}
	For every finite field $\F_q$ and every $\gamma > 0$, there exist $\tau$ such that the following holds.
	Let $(X_1, A_1)$ and $(X_2, A_2)$  be any pair of independent random variables with $X_i \in \F_q$, and such that $A_1, A_2$ are identically distributed, and moreover for every $a$ we have $\bH(X_1 ~|~ A_1 = a) = \bH(X_2 ~|~ A_2 = a)$. Then if $\bH(X_1 ~|~ A_1) = \bH(X_2 ~|~ A_2) \leq \tau$, we have
	\begin{equation*}
		H(X_1 ~|~ X_1 + X_2, A_1, A_2) \leq \gamma \bH(X_1 ~|~ A_1).
	\end{equation*}
\end{lemma}

Finally we use the following lemma due to \cite[Lemma 4.2]{sasoglu-book}, which corresponds to the \emph{Variance in the middle} condition of \cref{def:matrix-polarization}.
This is the only place where we need the field size $q$ to be prime.

\begin{lemma}[{\cite[Lemma 4.2]{sasoglu-book}}]
	\label{lem:2x2-middle}
	For every $\tau > 0$ and prime finite field $\F_q$, there exist $\varepsilon > 0$ such that if $(X_1, A_1)$ and $(X_2, A_2)$ are independent pairs of random variables (but not necessarily identically distributed), with $X_i \in \F_q$ for some prime $q$. Then
	\begin{equation*}
		\bH(X_1 ~|~ A_1), \bH(X_2~|~ A_2) \in (\tau, 1-\tau)
	\end{equation*}
	implies
	\begin{equation*}
		\bH(X_1 + X_2|A_1, A_2) \geq \max\{\bH(X_1 ~|~ A_1), \bH(X_2 ~|~ A_2)\} + \varepsilon.
	\end{equation*}
\end{lemma}

\subsection{Matrix polarization of \Arikan's $2 \times 2$ kernel}
As an illustration of how the lemmas arise in the study of matrix polarization we prove that $G_2$ satisfies matrix polarization. We remark that we do not need this lemma for the rest of this paper --- we present it purely as an example.

\begin{lemma}\label{lem:g2-polarizes}
	Over every prime field $\F_q$, the matrix $G_2 = \left(\begin{smallmatrix} 1 & 0 \\ 1 & 1 \end{smallmatrix}\right)$ satisfies matrix polarization.
\end{lemma}
\begin{proof}
	Note that we have
	\[(V_1,V_2)=(U_1,U_2)\cdot \left(\begin{matrix} 1 & 0 \\ 1 & 1 \end{matrix}\right),\]
	i.e. $V_1=U_1+U_2$ and $V_2=U_2$.
	By \cref{lem:2x2-middle} we have that the choice $j=1$ satisfies the \emph{variance in the middle} condition of matrix polarization for $G_2$. By \cref{lem:suc-lwr-cond} we have the choice $j=2$ yields the \emph{suction at the low end} condition (with $c=\frac 1\gamma$) of matrix polarization for $G_2$. Finally by \cref{lem:suc-upr-cond} we have that the choice $j=1$ satisfies the \emph{Suction at the upper end} condition (with $c=\frac 1\gamma$) for $G_2$.
\end{proof}

\subsection{Polarization of $k \times k$ mixing matrices \label{sec:local-k-by-k}}

In this section we prove the following
\begin{lemma}
	\label{lem:k-by-k-polarize}
	For every prime field $\F_q$ and every positive $k$ every mixing matrix $M \in \F_q^{k \times k}$ satisfies matrix polarization.
	
\end{lemma}

We will apply Gaussian elimination on $M$ to reduce to the entropic inequalities of the $2 \x 2$ case
from \cref{sec:entropic-statements}.
The high-level strategy for showing polarization of $k \x k$ mixing matrix $M$ is as follows. Consider i.i.d. random variables $(\bvec{U}_1, {W}_1), \ldots (\bvec{U}_k, {W}_k)$, and linearly transformed variables $\bvec{V} = \bvec{U} \cdot M$, where $\bvec{U}=\parens{\bvec{U}_1,\dots,\bvec{U}_k}$.

In \cref{sec:reduction} we will show that:
\begin{enumerate}
	\item  There are some indices $j, \ell, s \in [k]$ and some $\alpha \in \F_q^*$
	for which
	$$\bH(\bvec{V}_j | \bvec{V}_{< j}, {W})
	\geq \bH(\bvec{U}_\ell + \alpha \bvec{U}_s | {W}).$$
	\item  There are some indices $j, \ell, s \in [k]$ and some $\alpha \in \F_q^*$
	for which
	$$\bH(\bvec{V}_j | \bvec{V}_{< j}, {W}) \leq \bH(\bvec{U}_\ell | \bvec{U}_\ell + \alpha \bvec{U}_s, {W}).$$
\end{enumerate}

Those two, together with entropic inequalities stated in
Section~\ref{sec:entropic-statements}, are enough to show polarization of a given matrix. 

Before we proceed with the formal proof of those two inequalities, we give an informal exposition of the main idea behind it. For the sake of this exposition, let us focus on the inequality $\bH(\bvec{V}_j | \bvec{V}_{<j}, W) \leq \bH(\bvec{U}_\ell | \bvec{U}_\ell + \alpha \bvec{U}_s, W)$, and let us skip conditioning on $W$. 

The main observation is that if $\bvec{B}_1, \ldots \bvec{B}_m, \bvec{B}_{m+1}$ are all linear combinations of variables $\bvec{V}_1, \ldots \bvec{V}_{j-1}$, we have $\bH(\bvec{V}_j | \bvec{V}_{<j}) = \bH(\bvec{V}_j + \bvec{B}_{m+1} | \bvec{V}_{<j}, \bvec{B}_1, \ldots \bvec{B}_m) \leq \bH(\bvec{V}_j + \bvec{B}_{m+1} | \bvec{B}_1, \ldots \bvec{B}_m)$. Here it is enough to instantiate this observation with $m=1$. Since variables $\bvec{V}_i$ are themselves linear combinations of variables $\bvec{U}$ (with coefficients given by matrix $M$), all we need to do is find an index $j$, and some indices $\ell, s$, such that
\begin{equation}
	\label{eq:U_l-and-U_s}
	\bvec{U}_\ell \in \bvec{V}_j + \mathrm{span}\{\bvec{V}_{<j}\}  \quad \text{and} \quad \bvec{U}_{\ell} + \alpha \bvec{U}_s \in \mathrm{span}\{\bvec{V}_{<j}\}  \ . 
\end{equation}
In particular, we use $\bvec{B}_2\in \mathrm{span}\{\bvec{V}_{<j}\}$ for the first inclusion to set $\bvec{U}_\ell = \bvec{V}_j+\bvec{B}_2$ and use $\bvec{B}_1\in \mathrm{span}\{\bvec{V}_{<j}\}$ for the  second inclusion to set $\bvec{U}_{\ell} + \alpha \bvec{U}_s=\bvec{B}_1$. Note that with the inequalities in the above paragraph would give us the desired inequality.

Turns out that if the matrix $M$ is mixing, this can be achieved by carefully applying Gaussian Elimination on this matrix, as we explain next. 

\subsection{Reduction to the $2\times 2$ case} 
\label{sec:reduction}

This section will be devoted to proving following two lemmas.

\begin{lemma}[Reduction for suction at the upper end and variance]
	\label{lem:red-upper}
	Let $(\bvec{U}, W)$ be a joint distribution where $\bvec{U}=\parens{\bvec{U}_1,\dots,\bvec{U}_k} \in \F_q^{k}$ (with $\bvec{U}_i$ for $i\in [k]$ being independent conditioned on $W$)
	and
	let $M$ be any mixing matrix.
	Then, there exist three indices $j, \ell, s \in [k]$, and $\alpha \in \F_q^*$, such that
	$$\bH((\bvec{U} M)_{j} ~|~ (\bvec{U}M)_{< j}, W) \geq \bH(\bvec{U}_\ell + \alpha \bvec{U}_s ~|~ W).$$
\end{lemma}

\begin{lemma}[Reduction for suction at the lower end]
	\label{lem:red-lower}
	Let $(\bvec{U}, W)$ be a joint distribution, where $\bvec{U}=\parens{\bvec{U}_1,\dots,\bvec{U}_k} \in \F_q^k$, and let $M$ be any mixing matrix.
	Then, there exist three indices $j, \ell,s \in [k]$, and $\alpha \in \F_q^*$, such that
	$$\bH((\bvec{U} M)_j ~|~ (\bvec{U}M)_{< j}, W) \leq \bH(\bvec{U}_\ell ~|~ \bvec{U}_\ell + \alpha \bvec{U}_s , W).$$
\end{lemma}

As discussed previously, in order to show \cref{lem:red-lower} and \cref{lem:red-upper}, we will apply Gaussian Elimination to prove the following three lemmas about mixing matrices.

We start with a simple equivalent characterization of a mixing matrix: 
\begin{lemma}
	\label{lem:mixing-equiv}
	Invertible matrix $M$ is mixing if and only if there exists $j$ such that the support of the first $j$ columns has size greater than $j$.
\end{lemma}
\begin{proof}
	We need to prove that if we let $S_j = \{i \in [k] | \text{ exists } j' \in [j] \mbox{ s.t. } M_{i,j'}\ne 0\}$ then there exists a $j$ s.t. $|S_j| > j$. To see this note that $|S_j|$ is invariant under permutation of the rows, and for upper triangular matrices $|S_j|\leq j$. So if $M$ is not mixing then for all $j$ we have $|S_j| \leq j$. Conversely, if for every $j$ we have $|S_j|\leq j$, then either we have $|S_j| < j$ for some $j$, and in which case $M$ is not invertible, or $|S_j| = j$ for every $j$, in which case we can find a permutation $\pi:[k]\to [k]$ such that for every $j$ $S_j = \{\pi(1),\ldots,\pi(j)\}$. Permuting the rows so that $\pi(j)$ is the $j$th row makes $M$ upper triangular and so once again we get $M$ is not mixing.
\end{proof}

We will now state the linear-algebraic properties of a mixing matrices that correspond directly to the entropic inequalities in \cref{lem:red-upper} and \cref{lem:red-lower}. Specifically, it will not be too difficult to deduce \cref{lem:red-upper} from \cref{lem:red-upper-mat} as we discussed before --- the crux of the argument is that $\bH((\bvec{U} M)_{j} ~|~ (\bvec{U}M)_{< j}, W) = \bH((\bvec{U} M)_j + \bvec{B} ~|~ (\bvec{U} M)_{<j}, W)$ where $\bvec{B}$ is some linear combination of variables $(\bvec{U} M)_{1}, \ldots (\bvec{U} M)_{j-1}$, and~\cref{lem:red-upper-mat} describes how to find suitable $\bvec{B}$. \cref{lem:red-lower-mat} plays the same role in the proof of \cref{lem:red-lower}.

\begin{lemma}
	\label{lem:red-upper-mat}
	Let $M$ be a mixing matrix, and let $\bvec{a}_1, \ldots \bvec{a}_k \in \F_q^k$ denote columns of $M$. Then there exists index $j$ and a vector $\bvec{v} \in \bvec{a}_j + \mathrm{span}\{\bvec{a}_1, \ldots \bvec{a}_{j-1}\}$, such that $|\supp(\bvec{v})| \geq 2$ and $\supp(\bvec{v}) \cap \supp(\bvec{a}_i) = \emptyset$ for $i < j$, where $\supp(\bvec{v})$ is a set of non-zero coordinates of $\bvec{v}$.
\end{lemma}
\begin{proof}
	Let $S_i = \cup_{t \leq i} \supp(\bvec{a}_t)$. By \cref{lem:mixing-equiv}, this means that there exist a $j$ such that $|S_j|>j$.
	
	Consider the smallest $j$ satisfying $|S_j| > j$.  By a straightforward inductive argument for any $k <j$, we have $\mathrm{span}\{\bvec{a}_1, \ldots \bvec{a}_{k} \} = \mathrm{span} \{ \bvec{e}_\ell : \ell \in S_{k} \}$, where $\bvec{e}_i$ are the standard basis vectors. Now, we can decompose $\bvec{a}_j = \bvec{v} + \bvec{w}$ where $\supp(\bvec{w}) \subseteq S_{j-1}$ and $\supp(\bvec{v}) \cap S_{j-1} = \emptyset$. Since $|S_{j-1}| = j-1$ and $|S_{j}| > j$, we have $|\supp(\bvec{v})| \geq  2$, and by construction $\bvec{w} \in \mathrm{span}\{\bvec{e}_\ell : \ell \in S_{j-1}\} = \mathrm{span}\{\bvec{a}_1, \ldots \bvec{a}_{j-1}\}$.
\end{proof}
\begin{lemma}
	\label{lem:red-lower-mat}
	Let $M$ be a mixing matrix, and let $\bvec{a}_1, \ldots \bvec{a}_k \in \F_q^k$ denote columns of $M$. Then, there exists three indices $j,\ell,s \in [k]$ and $\alpha_1, \alpha_2 \in \F_q^*$, such that $\alpha_1 \bvec{e}_\ell \in \bvec{a}_j + \mathrm{span}\{\bvec{a}_1, \ldots \bvec{a}_{j-1}\}$ and $\bvec{e}_\ell + \alpha_2 \bvec{e}_s \in \mathrm{span}\{\bvec{a}_1, \ldots \bvec{a}_{j-1}\}$, where $\bvec{e}_i \in \F_q^k$ are the standard basis vectors.
\end{lemma}

In our proof of this lemma we will use a process which is essentially the well-known Gaussian elimination applied to a matrix $M$. Specifically the following proposition captures the properties of intermediate matrices in a Gaussian elimination process useful for our argument.

\begin{proposition}
	\label{prop:column-span}
	For any $k \times k$ invertible matrix $M$ there is a permutation $\pi : [k] \to [k]$ and a sequence of matrices $M^{(1)}, \ldots M^{(k)}$ (we call matrix $\Mat{j}$ the $j$-th step matrix) with the following properties.
	
	For any $j$, if we use $\bvec{a}_1, \ldots \bvec{a}_k$ to denote columns of $M$, $\bvec{a}'_1, \ldots \bvec{a}'_k$ to denote columns of $M^{(j)}$, and $\bvec{e}_1, \ldots \bvec{e}_k$ to denote standard basis vectors, we have
	\begin{enumerate}
		\item For every $s \in [k]$ we have $\bvec{a}'_s \in  \bvec{a}_s + \mathrm{span}\{\bvec{a}_1, \ldots \bvec{a}_{j}\}$.
		\item For every $s \in [k]$ and every $\ell \leq j$ we have $\inprod{\bvec{a}'_s, \bvec{e}_{\pi(\ell)}} = \left\{\begin{array}{ll} 1 & \mathrm{ if }\, \ell = s \\ 0 & \mathrm{otherwise.} \end{array}\right.$
		\item $\mathrm{span}\{\bvec{a}_1, \ldots, \bvec{a}_{j}\} = \mathrm{span}\{\bvec{a}'_1, \ldots \bvec{a}'_{j}\}.$
	\end{enumerate}
\end{proposition}
For example, if $j=3$, then $\Mat{3}$ up to some row permutation $\pi$ must have the form:
\begin{equation}
	\label{eq:M-3-example}
	\Mat{3} =
	\begin{bmatrix}
		1 		& 	0	& 	0 & 0 & \hdots\\
		0 		& 	1 	& 	0 & 0 & \hdots\\
		0       &   0   &	1&  0 & \hdots\\
		\star   & 	\star 	&	\star & \star & \hdots\\
		\star 	& 	\star	&	\star & \star & \hdots
	\end{bmatrix},
\end{equation}
where each column of $M^{(3)}$ is a corresponding column of $M$ shifted by some linear combination of the first three columns of $M$.

\begin{proof}[Proof of~\cref{prop:column-span}]
	The proof proceeds by induction. For the base case we consider $M^{(0)}=M$ and it is easy to verify the properties for $M$.
	
	Let $j\ge 1$. For the inductive hypothesis, we assume a matrix $M^{(j-1)}$ satisfying properties above, and a one-to-one map $\pi : [j-1] \to [k]$. For the inductive step, we want to find $M^{(j)}$, and $\pi(j)$ as in the statement of this proposition. Note that at the end of the induction, when $j=k$ the one-to-one map $\pi$ is also onto and hence $\pi:[k]\to [k]$ is a permutation as needed.
	
	Let $\bvec{a}^{(j-1)}_1, \ldots \bvec{a}^{(j-1)}_{k}$ denote columns of $M^{(j-1)}$. Since $M$ is invertible, we have $\bvec{a}_j \not \in \mathrm{span}\{\bvec{a}_1, \ldots \bvec{a}_{j-1}\}$. Using properties 1 and 3 for $\bvec{M}^{(j-1)}$, we conclude that $\bvec{a}_j^{(j-1)} \not\in \mathrm{span}\{\bvec{a}^{(j-1)}_1, \ldots \bvec{a}^{(j-1)}_{j-1}\}$. In particular, this implies that $\bvec{a}_j^{(j-1)} \not= \bvec{0}$. Let $\pi(j)$ be such that $\inprod{\bvec{a}^{(j-1)}_j, \bvec{e}_{\pi(j)}} \not= 0$. Note that $\pi(j) \not= \pi(s)$ for any $s < j$ by property 2 of the matrix $M^{(j-1)}$, and therefore $\pi$ is a one-to-one mapping.
	
	For $s \not= j$ let us take $\bvec{a}'_s := \bvec{a}^{(j-1)}_s - \frac{\inprod{\bvec{a}^{(j-1)}_{s}, \bvec{e}_{\pi(j)}}}{\inprod{\bvec{a}^{(j-1)}_{j}, \bvec{e}_{\pi(j)}}} \bvec{a}^{(j-1)}_j$, and finally $\bvec{a}'_j := \frac{1}{\inprod{\bvec{a}^{(j-1)}_{j}, \bvec{e}_{\pi(j)}}} \bvec{a}^{(j-1)}_j$.
	
	Next, we verify properties 1--3 hold for matrix $M^{(j)}$ given by the columns $\bvec{a}'_1, \ldots \bvec{a}'_k$ defined above.
	
	Indeed, the first property holds since for any $s$, we have $\bvec{a}'_s = \bvec{a}^{(j-1)}_s + \gamma \bvec{a}^{(j-1)}_j$ where $\gamma$ is some scalar. By induction, we have $\bvec{a}^{(j-1)}_s \in \bvec{a}_s + \mathrm{span}\{\bvec{a}_1, \ldots \bvec{a}_{j-1}\}$, and $\bvec{a}^{(j-1)}_j \in \mathrm{span}\{\bvec{a}_1, \ldots, \bvec{a}_j\}$, therefore  $\bvec{a}'_s \in \bvec{a}_s + \mathrm{span}\{\bvec{a}_1, \ldots \bvec{a}_j\}$.
	
	To show the second property, for any $s \not= j$ we have
	\begin{equation*}
		\inprod{\bvec{a}'_s, \bvec{e}_{\pi(\ell)}} = \inprod{\bvec{a}^{(j-1)}_s, \bvec{e}_{\pi(\ell)}} - \frac{\inprod{\bvec{a}^{(j-1)}_{s}, \bvec{e}_{\pi(j)}}}{\inprod{\bvec{a}^{(j-1)}_{j}, \bvec{e}_{\pi(j)}}} \inprod{\bvec{a}^{(j-1)}_j, \bvec{e}_{\pi(\ell)}}.
	\end{equation*}
	If $\ell < j$, by induction we have $\inprod{\bvec{a}^{(j-1)}_j, \bvec{e}_{\pi(\ell)}} = 0$, hence the second term vanish, and we have $\inprod{\bvec{a}'_s, \bvec{e}_{\pi(\ell)}} = \inprod{\bvec{a}^{(j-1)}_s, \bvec{e}_{\pi(\ell)}}$, which again by induction is $1$ if $s = \ell$ and $0$ otherwise.
	For $\ell = j$ we have $\inprod{\bvec{a}'_s, \bvec{e}_{\pi(\ell)}} = \inprod{\bvec{a}^{(j-1)}_s, \bvec{e}_{\pi(\ell)}} - \inprod{\bvec{a}^{(j-1)}_{s}, \bvec{e}_{\pi(j)}} = 0$. Further,  when $s = j$, we have $\inprod{\bvec{a}'_j, \bvec{e}_{\pi(\ell)}} = \frac{\inprod{\bvec{a}_j^{(j-1)}, \bvec{e}_{\pi(\ell)}}}{\inprod{\bvec{a}_j^{(j-1)}, \bvec{e}_{\pi(j)}}}$ is $1$ exactly when $\ell = j$ and $0$ when $\ell < j$ (where the latter claim follows from property 2 for $\Mat{j-1}$.
	
	Finally, for the third property, the inclusion $\mathrm{span}\{ \bvec{a}'_1, \ldots \bvec{a}'_j\} \subset \mathrm{span}\{\bvec{a}_1, \ldots \bvec{a}_j\}$ follows from the property 1, and the $\dim \mathrm{span}\{\bvec{a}'_1, \ldots \bvec{a}'_j\} = j$ by property $2$, which implies $\mathrm{span}\{\bvec{a}'_1, \ldots \bvec{a}'_j\} = \mathrm{span}\{\bvec{a}_1, \ldots \bvec{a}_j\}$.
\end{proof}

\begin{proof}[Proof of Lemma~\ref{lem:red-lower-mat}]
	By~\cref{lem:mixing-equiv} a matrix $M$ is mixing if for some index $i$ the support of the first $i$ columns has size strictly greater than $i$. Let $j-1$ be the largest index with this property, and note that $j \leq k$ (since all the $k$ columns trivially have support size of $k$).
	
	Let $\Mat{j-1}$ be the $(j-1)$-th step matrix of $M$ defined as in \cref{prop:column-span}.

	By definition of $j$, the span of the first $j$ columns of $M$ must exactly equal
	$\text{span}\{\bvec{e}_{\pi(1)}, \dots, \bvec{e}_{\pi(j )}\}$, since the total support of all those columns has size exactly $j$, the columns are linearly independent, and each of $\pi(1), \ldots \pi(j)$ is in this support by property 2 and 3 of matrix $M^{(j)}$.

	Thus, all of the first $j$ columns of $\Mat{j-1}$ can only be supported on coordinates $\{\pi(1), \dots, \pi(j)\}$.
	Further, by the second property in \cref{prop:column-span} of the $(j-1)$-th step matrix, the $j$-th column of $\Mat{j-1}$ has zero on all coordinates $\pi(s)$ for $s < j$. Thus, it must be of form $\alpha_1 \bvec{e}_{\pi(j)}$ for some scalar $\alpha_1 \not= 0$ (since otherwise the $j$-th column of $\Mat{j-1}$ would be $\bvec{0}$, which would contradict the fact that $M$ is invertible). 

	Finally, because total support of the first $(j-1)$ columns is larger than $j-1$,
	there must exist some column $\ell < j$ of $\Mat{j-1}$
	that is supported on the coordinate $\pi(j)$.
	This, along with the second property of $M^{(j-1)}$ in \cref{prop:column-span} implies that  the $\ell$-th column
	of $\Mat{j-1}$ must be exactly $(\bvec{e}_{\pi(\ell)} + \beta_2 \bvec{e}_{\pi(j)})$ for some $\beta_2 \in \F_q^*$.
	
	We can now conclude the statement of the Lemma.  We have shown that the $j$-th column of $\Mat{j-1}$ is of form $\alpha_1 \bvec{e}_{\pi(j)}$, and by the first property of $\Mat{j-1}$ in \cref{prop:column-span} it is contained in $\bvec{a}_j + \mathrm{span}\{\bvec{a}_1, \ldots, \bvec{a}_{j-1}\}$. This proves the first part of the statement of the Lemma (by using $\ell\gets\pi(j)$). The argument for the second part is as follows.
	We have shown that on one hand the $\ell$-th column of $\Mat{j-1}$ is of form $\bvec{e}_{\pi(\ell)} + \beta_2 \bvec{e}_{\pi(j)}$, on the other hand it is contained in the span of first $j-1$ columns of the matrix $M$ by the first property of $\Mat{j-1}$ in \cref{prop:column-span}. In other words, we have $\beta_2^{-1}\cdot \bvec{e}_{\pi(\ell)}+ \bvec{e}_{\pi(j)}$ is in the span of first $j-1$ columns of the matrix $M$. This shows the second part of the statement of the Lemma (by picking $\alpha_2\gets\beta_2^{-1}$ and $s\gets \pi(\ell)$ and recalling that in the first part we have already set $\ell\gets \pi(j)$). 
	\end{proof}
	
With Lemmas~\ref{lem:red-upper-mat} and \ref{lem:red-lower-mat} in hand, we are well equipped to show Lemmas~\ref{lem:red-upper} and \ref{lem:red-lower} accordingly.

\begin{proof}[Proof of~\cref{lem:red-upper}]
	Let $\bvec{a}_1, \ldots \bvec{a}_k$ be columns of matrix $M$. By \cref{lem:red-upper-mat} there is an index $j$, and a vector $\bvec{v} = \bvec{a}_j + \bvec{w}$ where $\bvec{w} \in \mathrm{span}\{\bvec{a}_1, \ldots \bvec{a}_{j-1}\}$ such that $\supp(\bvec{v}) \cap \supp(\bvec{a}_i) = \emptyset$ for each $i < j$. 
	
	This implies 
	\begin{align*}
		\bH( (\bvec{U}M)_j\,|\, (\bvec{U}M)_{<j}, W) & = \bH(\inprod{\bvec{U}, \bvec{a}_j} \,|\, \inprod{\bvec{U}, \bvec{a}_1}, \ldots, \inprod{\bvec{U}, \bvec{a}_{j-1}}, W) \\
		& = \bH( \inprod{\bvec{U}, \bvec{a}_j} + \inprod{\bvec{U}, \bvec{w}} \, | \, \inprod{\bvec{U}, \bvec{a}_1}, \ldots \inprod{\bvec{U}, \bvec{a}_{j-1}}, W) 
		\tag{Since $\bvec{w} \in \mathrm{span}\{\bvec{a}_1, \ldots \bvec{a}_{j-1}\}$}\\
		& = \bH( \inprod{\bvec{U}, \bvec{v}} \, | \,  W).
	\end{align*}
	where the last equality follows from the fact that $\inprod{ \bvec{U}, \bvec{v}}$ is independent from $\inprod{\bvec{U}, \bvec{a}_1}, \ldots \inprod{\bvec{U}, \bvec{a}_{j-1}}$ conditioned on $W$ (since $\bvec{v}$ has disjoint support with all $\bvec{a}_i$ for $i < j$).
	
	Now, since $|\supp(\bvec{v})| > 2$, let us say that $\bvec{v} = \alpha_\ell \bvec{e}_\ell + \alpha_2 \bvec{e}_s + \bvec{r}$, where $\supp(\bvec{r}) \cap \{\ell, s\} = \emptyset$. We have 
	\begin{align*}
		\bH(\inprod{\bvec{U}, \bvec{v}} \, | \, W) & 
		= \bH(\alpha_1 \bvec{U}_\ell + \alpha_2 \bvec{U}_s + \inprod{\bvec{U}, \bvec{r}} \, | \,  W)\\
		& \geq \bH(\alpha_1 \bvec{U}_\ell + \alpha_2 \bvec{U}_s + \inprod{\bvec{U}, \bvec{r}} \, | \, \inprod{\bvec{U}, \bvec{r}}, W)\tag{Since conditioning does not increase entropy}\\
		& = \bH(\bvec{U}_\ell + \alpha_1^{-1} \alpha_2 \bvec{U}_s \, | \, \inprod{\bvec{U}, \bvec{r}}, W)\tag{Since $x\mapsto \alpha_1^{-1}\cdot x$ is a bijection}\\
		& = \bH(\bvec{U}_\ell + \alpha_1^{-1} \alpha_2 \bvec{U}_s \, | \, W), 
	\end{align*}
	where the last equality follows since $\supp(\bvec{r}) \cap \{\ell, s\} = \emptyset$. The proof is complete by setting $\alpha=\alpha_1^{-1}\alpha_2$.
\end{proof}

\begin{proof}[Proof of~\cref{lem:red-lower}]
	Let $\bvec{a}_1, \ldots \bvec{a}_k$ be columns of matrix $M$.
	By \cref{lem:red-lower-mat}, there are indices $j,\ell,s$ and $\alpha_1,\alpha_2 \in \F_q^*$, such that $\alpha_1\cdot\bvec{e}_\ell = \bvec{a}_j + \bvec{w}$ 
	where $\bvec{w} \in \mathrm{span}\set{\bvec{a}_1, \ldots \bvec{a}_{j-1}}$, and $\bvec{e}_\ell + \alpha_2 \bvec{e}_s \in \mathrm{span}\set{\bvec{a}_1, \ldots \bvec{a}_{j-1}}$.
	
	This implies
	\begin{align*}
		\bH((\bvec{U} M)_j \, | \, (\bvec{U} M)_{<j}, W)
		& = \bH(\inprod{\bvec{U}, \bvec{a}_j} \,|\, \inprod{\bvec{U}, \bvec{a}_1}, \ldots, \inprod{\bvec{U}, \bvec{a}_{j-1}}, W) \\
		& = \bH(\inprod{\bvec{U}, \alpha_1\bvec{e}_\ell} - \inprod{\bvec{U}, \bvec{w}} \,|\, \inprod{\bvec{U}, \bvec{a}_1}, \ldots, \inprod{\bvec{U}, \bvec{a}_{j-1}}, \inprod{\bvec{U}, \bvec{e}_\ell + \alpha_2 \bvec{e}_s}, W)
		\tag{Since $\bvec{e}_\ell + \alpha_2 \bvec{e}_s$ is in $\mathrm{span}\{\bvec{a}_1, \ldots \bvec{a}_{j-1}\}$}\\
		& = \bH(\inprod{\bvec{U}, \alpha_1\bvec{e}_\ell} \,|\, \inprod{\bvec{U}, \bvec{a}_1}, \ldots, \inprod{\bvec{U}, \bvec{a}_{j-1}}, \inprod{\bvec{U}, \bvec{e}_\ell + \alpha_2 \bvec{e}_s}, W) 
		\tag{Since $\bvec{w}$ is in $\mathrm{span}\{\bvec{a}_1, \ldots \bvec{a}_{j-1}\}$}\\
		& \leq \bH(\inprod{\bvec{U}, \alpha_1\bvec{e}_\ell} \,|\, \inprod{\bvec{U}, \bvec{e}_\ell + \alpha_2 \bvec{e}_s}, W) 
		\tag{Additional conditioning decreases entropy.}\\
		& = \bH(\alpha_1\bvec{U}_\ell \, | \, \bvec{U}_\ell + \alpha_2 \bvec{U}_s, W)\\
		& = \bH(\bvec{U}_\ell \, | \, \bvec{U}_\ell + \alpha_2 \bvec{U}_s, W),
	\end{align*}
	where the last equality follows since $\alpha_1\in\F_q^*$ and hence the map $x\mapsto \alpha_1\cdot x$ is a bijection.
\end{proof}

We are now ready to prove that every mixing matrix is a polarizing matrix.
\begin{proof}[Proof of \cref{lem:k-by-k-polarize}]
	The proof follows easily by combining \cref{lem:red-upper,lem:2x2-middle,lem:suc-upr-cond,lem:red-lower,lem:suc-lwr-cond} as we elaborate on below.
	
	Let $M \in \F_q^{k\times k}$ be a mixing matrix, and let $(\bvec{U}, W)$ be random variables such that  $\bvec{U} = (\bvec{U}_1,\dots,\bvec{U}_k) \in \F_q^k$, $W = (W_1,\ldots,W_k)$ is supported on some finite set, and the pairs  $(\bvec{U}_i,W_i)$ are independently and identically distributed for $i \in [k]$. Further let the vector $\bvec{V} := \bvec{U} \cdot M$. 
	
	For the {\em Variance in the middle} condition we need to show that there exists $j \in [k]$ such that for every $\tau >0$ there exists $\eps=\eps(\tau) > 0$ such that
	if $\bH(\bvec{U}_1 | W) \in (\tau, 1-\tau)$, then
	$\bH((\bvec{V})_j | \bvec{V}_{< j}, W)  \geq \bH(\bvec{U}_1 | W) + \eps$. By \cref{lem:red-upper} we have that there exist $j,\ell,s\in[k]$ and $\alpha \in \F_q^*$ such that \begin{equation}\label{eqn:polar-verify-one}
	\bH((\bvec{V})_j | \bvec{V}_{< j}, W)  \geq \bH(\bvec{U}_\ell + \alpha\bvec{U}_s | W) =\bH(\bvec{U}_\ell + \alpha\bvec{U}_s | W_\ell,W_s),
	\end{equation}
	where the equality above uses the fact that the $(\bvec{U}_i,W_i)$ pairs are independent.
	By \cref{lem:2x2-middle} applied with $X_1 = \bvec{U}_\ell$, $A_1 = W_\ell$, $X_2 = \alpha \bvec{U}_s$ and $A_2 = W_s$ we conclude that for every $\tau >0$ there exists $\eps>0$ such that 
	$$\bH(\bvec{U}_\ell + \alpha\bvec{U}_s | W_\ell,W_s) \geq \max\{\bH(\bvec{U}_\ell|W_\ell),\bH(\alpha\bvec{U}_s|W_s)\}+\eps =  \bH(\bvec{U}_1|W_1)+\eps = \bH(\bvec{U}_1|W)+\eps,$$
	where the first equality uses the fact that the map $\alpha\bvec{U}_s\mapsto\bvec{U}_s$ is invertible  the fact and that the $(\bvec{U}_i,W_i)$ pairs are identically distributed and the final equality uses the fact that these pairs are independent. 
	The Variance in the middle condition follows by combining the two steps above.
	
	For suction at the high end we need to show there is some index $j \in [k]$ such that for every $c < \infty$, there  exists $\tau > 0$ such that if
	$\bH(\bvec{U}_1 | W) > 1- \tau$ then $1 - \bH(\bvec{V}_j | \bvec{V}_{< j}, W)  \leq \frac{1}{c}(1-\bH(\bvec{U}_1 | W))$. Here again by \cref{lem:red-upper} we have that there exist $j,\ell,s\in[k]$ and $\alpha \in \F_q^*$ such that \cref{eqn:polar-verify-one} holds. Now applying \cref{lem:suc-upr-cond} to $X_1,A_1,X_2,A_2$ as in the previous paragraph and $\gamma = 1/c$ we get that there exists $\tau > 0$ such that the requirement for suction at the higher end holds.

	Finally, for suction at the lower end we need to show there is some index $j \in [k]$ such that for every $c < \infty$, there  exists $\tau > 0$ such that if
	$\bH(\bvec{U}_1 | W) < \tau$ then $\bH(\bvec{V}_j | \bvec{V}_{< j}, W)  \leq \frac{1}{c}~\bH(\bvec{U}_1 | W)$. We first apply \cref{lem:red-lower} to get that there exist
    $j,\ell,s\in[k]$ and $\alpha \in \F_q^*$ such that 
    $$
	\bH((\bvec{V})_j | \bvec{V}_{< j}, W)  \leq \bH(\bvec{U}_\ell | \bvec{U}_\ell + \alpha\bvec{U}_s , W) =\bH(\bvec{U}_\ell | \bvec{U}_\ell + \alpha\bvec{U}_s,  W_\ell,W_s).$$
	Now applying \cref{lem:suc-lwr-cond} with $X_1,A_1,X_2,A_2$ and $\gamma$ as in the previous paragraph we get that 
	$$\bH(\bvec{U}_\ell | \bvec{U}_\ell + \alpha\bvec{U}_s,  W_\ell,W_s) \leq \frac1c \bH(\bvec{U}_\ell | W_\ell) = \frac1c \bH(\bvec{U}_1 | W_1) = \frac1c \bH(\bvec{U}_1 | W).$$
\end{proof}

This concludes our analysis of the reductions.

\subsection{Proof of \cref{thm:triangle-local}}
\label{subsec:thm-1.14}
For completeness and easy reference, we restate \cref{thm:triangle-local} below and include its proof.

\thmtrianglelocal*

\begin{proof}[Proof of \cref{thm:triangle-local}]
	Let $M\in \F_q^{k \times k}$ be a mixing matrix. By \cref{lem:k-by-k-polarize} we have that $M$ satisfies matrix polarization. Now by   \cref{lem:polarizing-matrix-implies-martingale} we have that 
	for every symmetric memoryless channel
	$\C_{Y|Z}$ over $\F_q$, the \Arikan\ martingale associated with $M$ and $\C_{Y|Z}$ is locally polarizing.
\end{proof}

\section{Proofs of Entropic Lemmas}
\label{sec:entropic-proofs}
We now turn to the entropic lemmas stated and used  in \cref{sec:entropic-statements}.

\subsection{Suction at the upper end}

To establish \cref{lem:suc-upr-cond}, we will first show similar kind of statement for unconditional entropies. To this end, we first show that for random variables taking values in \emph{small} set, having entropy close to maximal is essentially the same as being close to uniform with respect to $L_2$ distance. The $L_2$ distance of a probability distribution to uniform is controlled by the sum of squares of non-trivial Fourier coefficients of the distribution, and all the non-trivial Fourier coefficients are significantly reduced after adding two independent variables close to the uniform distribution.

Finally a simple averaging argument is sufficient to lift this result to conditional entropies, establishing \cref{lem:suc-upr-cond}.

\begin{lemma}
	\label{lem:kl-l2-comparable}
	If $X \in \F_q$ is a random variable with a distribution $\cD_X$, then
	\[d_2(\cD_X, U)^2 \frac{1}{2 \log q} \leq 1 - \bH(X) \leq d_2(\cD_X, U)^2 \Oh(q^2),\]
	where $U$ is a uniform distribution over $\F_q$, and $d_p(\cD_1, \cD_2) := \left(\sum_{x \in \F_q} (\cD_1(x) - \cD_2(x))^p\right)^{1/p}$.
\end{lemma}

\begin{proof}
	Pinskers inequality \cite{pinsker64} yields $d_1(\cD_X, U) \leq \sqrt{2\log q} \cdot \sqrt{1 - \bH(X)}$, and by standard relations between $\ell_p$ norms, we have $d_2(\cD_X, U) \leq d_1(\cD_X, U)$, which after rearranging yields the bound $d_2(\cD_X, U)^2 \leq (2 \log q)(1 - \bH(X))$, which in turn proves the claimed lower bound.
	
	For the upper bound, given $i \in \F_q$ let us take $\delta_i$ such that $\cD_X(i) \stackrel{\text{def}}{=} \P\parens{X=i} = \frac{1 + \delta_i}{q}$. Note that this implies (along with the fact that $\sum_{i\in\F_q} \cD_X(i)=1)$:
	\begin{equation}
		\label{eq:sum-delta-i}
		\sum_{i \in \F_q} \delta_i = 0
	\end{equation} 
	and 
	\begin{equation}
		\label{eq:L-2-dist}
		d_2(\cD_X, U)^2 = \frac{1}{q^2} \sum_i \delta_i^2.
	\end{equation}
	Now
	\begin{equation*}
		1 - \bH(X) = 1+\frac{1}{\log q} \sum_{i \in \F_q}\frac{(1 + \delta_i)}{q} \log \parens{\frac{1 + \delta_i}q}=\frac{1}{\log q} \sum_{i \in \F_q} \frac{(1 + \delta_i)}{q} \log (1 + \delta_i),
	\end{equation*}
	where the second equality follows from \cref{eq:sum-delta-i}.

	By Taylor expansion we have $\log (1 + \delta_i) = \delta_i + \cE(\delta_i)$ with some error term $\cE(\delta_i)$ such that $|\cE(\delta_i)| \leq 2 \delta_i^2$ for $|\delta_i| < 1$. Therefore in the case when all $\delta_i < 1$, we have (for some constant $C$):
	\begin{align*}
		1 - \bH(X) & = \frac{1}{q \log q} \sum_{i \in F_q} (1 + \delta_i)(\delta_i + \cE(\delta_i)) \\
		& \leq \frac{1}{q \log q} \sum_{i \in F_q} [\delta_i + \delta_i^2 + \Oh(\delta_i^2)] \\
		& \leq \frac{1}{q \log q}\left[ \sum_{i \in F_q} \delta_i +  C \sum_{i \in F_q} \delta_i^2\right] \\
		& \leq C q \cdot d_2(\cD_X, U)^2,
	\end{align*}
	where the inequality follows from \cref{eq:sum-delta-i} and \cref{eq:L-2-dist}.
	If some $\delta_i \ge 1$, then the inequality is satisfied trivially: $d_2(\cD_X, U) \ge \frac{1}{q}$, hence $1 - \bH(X) \le q^2 \cdot d_2(\cD_X, U)^2$.
\end{proof}

\begin{lemma}
	\label{lem:entropy-unconditional}
	If $X, Y \in \F_q$ are independent random variables, then $1 - \bH(X + Y) \leq \poly(q) (1 - \bH(X))(1 - \bH(Y))$.
\end{lemma}
\begin{proof}
	By \cref{lem:kl-l2-comparable} it is enough to show that $d_2(\cD_{X+Y}, U)^2 \leq \poly(q) d_2(\cD_X, U)^2 d_2(\cD_Y, U)^2$.
	For a distribution $\cD_X$, consider a Fourier transform of this distribution given by $\hat{\cD}_X(k) = \E_{j \sim \cD_X} \omega^{jk}$, where $\omega = \exp(-2 \pi i/q)$. As usual, we have $\hat{\cD}_{X+Y}(k) = \hat{\cD}_X(k) \hat{\cD}_Y(k)$.
	
	Moreover, by Parseval's identity we will show that 
	\[d_2(\cD_X, U)^2  = \frac{1}{q}\sum_{k \not= 0} \hat{\cD}_X(k)^2.\] 
	Indeed --- as in the proof of \cref{lem:kl-l2-comparable}, define $\cD_X(i) =:  \frac{1 + \delta_i}{q}$. Then by Parseval's identity we have
	\[\frac{1}{q}\cdot\sum_{k\in\F_q} \hat{\cD}_X(k)^2 = \sum_{i\in\F_q} \frac{(1+\delta_i)^2}{q^2} =\frac 1{q^2}\parens{\sum_{i\in\F_q} 1 +2\sum_{i\in \F_q} \delta_i +\sum_{i\in\F_q} \delta_i^2}\frac{1}{q} +d_2(\cD_X, U)^2,\]
	which implies the claimed bound by noting that $\hat{\cD}_X(0)=1$. (In the above the last equality follows from \cref{eq:sum-delta-i} and \cref{eq:L-2-dist}.
	
	This yields
	\begin{align*}
		d_2(\cD_{X+Y}, U)^2 & = \frac 1q\cdot \sum_{k \not= 0} \hat{\cD}_X(k)^2 \hat{\cD}_Y(k)^2  \\
		& \frac 1q \cdot\leq \left(\sum_{k \not=0} \hat{\cD}_X(k)^2\right)\left(\sum_{k\not=0} \hat{\cD}_Y(k)^2\right) = q d_2(\cD_X, U)^2 d_2(\cD_Y, U)^2. \qedhere
	\end{align*}
\end{proof}

\begin{lemma}
	\label{lem:suc-upr-mult}
	Let $X_1, X_2 \in \F_q$ be a pair of random variables, and let $A_1, A_2$ be pair of discrete random variables, such that $(X_1, A_1)$ and $(X_2, A_2)$ are independent. Then
	\begin{equation*}
		1 - \bH(X_1 + X_2 | A_1, A_2) \leq (1 - \bH(X_1 | A_1))(1 - \bH(X_2 | A_2)) \poly(q).
		\label{}
	\end{equation*}
\end{lemma}
\begin{proof}
	We have
	\begin{eqnarray*}
		\lefteqn{1 - \bH(X_1 + X_2 | A_1, A_2)}\\
		&=&  \sum_{a_1, a_2} \P(A_1 = a_1) \P(A_2 = a_2) (1 - \bH(X_1 + X_2|A_1 = a_1, A_2= a_2))  \\
		&\leq& \poly(q) \sum_{a_1, a_2} \P(A_1 = a_1) \P(A_1 = a_1) (1 - \bH(X_1 | A_1 = a_1, A_2 = a_2))(1 - \bH(X_2 | A_1 = a_1, A_2 = a_2)) \\
		& =& \poly(q) \sum_{a_1, a_2} \P(A_1 = a_1) (1 - \bH(X_1 | A_1 = a_1)) \P(A_2 = a_2) (1 - \bH(X_2 | A_2 = a_2)) \\
		& =& \poly(q) \left(\sum_{a_1} \P(A_1 = a_1) (1 - \bH(X_1 | A_1 = a_1))\right)\left(\sum_{a_2} \P(A_2 = a_2) (1 - \bH(X_2 | A_2 = a_2))\right) \\
		& =& \poly(q) (1 - \bH(X_1~|~A_1)) (1 - \bH(X_2~|~A_2)),
		\label{}
	\end{eqnarray*}
	where the inequality follows from \cref{lem:entropy-unconditional} and the second equality follows from independence of $(X_1,A_1)$ and $(X_2,A_2)$.
\end{proof}

We are now ready to prove \cref{lem:suc-upr-cond}.

\begin{proof}[Proof of \cref{lem:suc-upr-cond}]
	Given $\gamma$, $q$, take $\tau = \gamma/P(q)$ where $P(q)$ is the polynomial appearing in the statement of \cref{lem:suc-upr-mult}. By applying the conclusion of \cref{lem:suc-upr-mult}, we have
	\begin{eqnarray*}
		1 - \bH(X_1 + X_2 | A_1, A_2)
		& \leq & (1 - \bH(X_1 ~|~ A_1))(1 - \bH(X_2 ~|~ A_2) P(q) \\
		& \leq & (1 - \bH(X_1 ~|~ A_1)) \tau P(q) \\
		& = & \gamma (1 - \bH(X_1 ~|~ A_1)). \qedhere
	\end{eqnarray*}
\end{proof}

\subsection{Suction at the lower end}
In this subsection will show \cref{lem:suc-lwr-cond}. To this end, we want to show that for pairs $(X_1, A_1)$ and $(X_2, A_2)$ with low conditional entropy $\bH(X_1 ~|~ A_1) < \tau, \bH(X_2 ~|~ A_2) < \tau$, the entropy of the sum is almost as big as sum of corresponding entropies, i.e. $\bH(X_1 + X_2 ~|~ A_1, A_2) \geq (1 - \gamma)(\bH(X_1 ~|~ A_1) + \bH(X_2 ~|~ A_2))$ --- and the statement of \cref{lem:suc-lwr-cond} will follow (as we show later) by application of chain rule. To this end, we first show the same type of statement for non-conditional entropies, i.e. if $\bH(X_1) < \tau, \bH(X_2) < \tau$, then $\bH(X_1 + X_2) > (1 - \gamma)(\bH(X_1) + \bH(X_2))$ --- this fact can be deduced  by reduction to the analogous fact for binary random variables, where it becomes just a simple computation. Then we proceed by lifting this statement to the corresponding statement about conditional entropies --- this requires somewhat more effort than in \cref{lem:suc-upr-cond}.

\begin{lemma}
	\label{lem:suc-lwr}
	Let $X, Y$ be independent random variables in $\F_q$.
	For any $\gamma < 1$, there exists $\alpha=\alpha(\gamma)$ such that:
	if $\bH(X) \leq \alpha$ and $\bH(Y) \leq \alpha$, then
	$$\bH(X+Y) \geq (1-\gamma)(\bH(X) + \bH(Y)).$$
\end{lemma}

First, we will show some preliminary useful lemmas.
\begin{assumption}
	In the following,
	without loss of generality, let $0$ be the most likely value for both random variables $X, Y$. (This shifting does not affect entropies).
\end{assumption}

\begin{lemma}
	\label{lem:ML}
	Let $X$ be a random variable over $\F_q$, such that $0$
	is the most-likely value of $X$.
	Then for any $q$ and any $\gamma < 1$ there exists
	$\alpha_2(q,\gamma) > 0$ such that
	$$\bH(X) \leq \alpha_2(q,\gamma) \implies \Pr[X \neq 0] \leq \gamma \bH(X).$$
\end{lemma}
\begin{proof}
	Let $\beta := \Pr[X \neq 0]$, and $\alpha := \bH(X)$.
	We have
	$$
	\alpha \log q = H(X) \geq H(\bar\delta(X))
	= H(\beta)
	\geq \beta\log(1/\beta).
	$$
	In the above the inequality follows from the fact that applying a deterministic fuction to a random variable can only decrease its entropy.
	Thus,
	\begin{align*}
		\Pr[X \neq 0] = \beta &\leq \frac{\alpha \log q}{\log(1/\beta)}\\
		&\leq \frac{\alpha \log q}{\log(1/\alpha) - \log \log q}
	\end{align*}
	where we used the fact that $\beta \leq \alpha \log q$
	from \cref{lem:ML-basic}. Hence, as soon as $\log \frac{1}{\alpha} \ge \frac{\log q}{\gamma} + \log \log q$, the statement of the lemma holds.
\end{proof}

\begin{lemma}[Suction-at-lower-end in the Binary Case]
	\label{lem:suc-lwr-binary}
	Let $U, V$ be independent binary random variables.
	There exists a function $\alpha_0(\gamma)$ 
	such that, for any $0 < \gamma < 1$,
	$$H(U), H(V) \leq \alpha_0(\gamma) \implies
	H(U \oplus V) \geq (1-\gamma)(H(U) + H(V)).\footnote{We note that we could have replaced $\oplus$ by just $+$ as those operations are over $\F_2$ but we chose to keep $+$ for addition over reals in this lemma.}$$
\end{lemma}

\begin{proof}
	Let $p_1$ and $p_2$ be the biases of $U, V$ respectively, such that $U \sim \text{Bernoulli}(p_1)$ and $V \sim \text{Bernoulli}(p_2)$.
	Let $p_1 \circ p_2 = p_1(1-p_2) + (1-p_1)p_2$ be the bias of $U \oplus V$,
	that is $U \oplus V \sim \text{Bernoulli}(p_1 \circ p_2)$.
	
	We first describe some useful bounds on $H(p)$. On the one hand we have 
	$H(p) \geq p \log 1/p$. For $p \leq 1/2$ we also have
	\[-(1-p)\log (1-p) \leq (1/\ln 2) (1-p) (p+p^2) \leq (1/\ln 2) p \leq 2p,\]
	where the first inequality follows from the fact that $-\ln(1-x)\le x+x^2$ for $x\le \frac 12$
	And so we have $H(p) \leq p (2 + \log 1/p)$.
	Summarizing, we have 
	\[p\log(1/p) \leq H(p) \leq p\log(1/p) + 2p.\]
	Suppose $H(p_1), H(p_2) \leq \tau$.
	We now consider $H(p_1) + H(p_2) - H(p_1\circ p_2)$. WLOG assume $p_1\le p_2$. Note that this implies
	\[p_1\circ p_2\le p_1+p_2\le 2p_2.\]
	We have
	\begin{align*}
		\lefteqn{H(p_1) + H(p_2) - H(p_1\circ p_2)}\\
		& \leq  p_1(\log(1/p_1) + 2) + p_2(\log(1/p_2)+2)  - (p_1\circ p_2)\log (1/(p_1\circ p_2))\\
		& \leq  p_1(\log(1/p_1) + 2) + p_2(\log(1/p_2)+2)  - (p_1+p_2-2p_1p_2)\log (1/(2p_2))\\
		& =  p_1\log(2p_2/p_1) + p_2 \log(2p_2/p_2)  + 2p_1p_2\log (1/(2p_2))+2(p_1+p_2)\\
		& \leq  p_1\log(p_2/p_1)   + 2p_1p_2\log (1/(p_2))+6p_2\\
		& \leq  2p_1 H(p_2) + 7p_2 \tag{Using $p_1\log (p_2/p_1) \leq p_2$ }\\
		& \leq  2p_1 H(p_2) + 7 H(p_2)/\log (1/p_2)\\
		& \leq  9 H(p_2)/\log(1/\tau).
	\end{align*}
	In the above, the last inequaliy follows from the assumption that $\tau\le 1/8$ (which will be true in our case). Indeed, note that with this assumption $\tau\log(1/\tau)\le 1$ (which along with the fact that $p_1\le \tau$ implies $p_1\le 1/\log(1/\tau)$) and $p_2\le \tau$ (since we have $p_2\log(1/p_2)\le \tau$).
	Thus, we have
	$$H(U), H(V) \leq \tau \implies H(U) + H(V) - H(U \oplus V) \leq 9H(V)/\log(1/\tau)$$
	This implies the desired statement, for $\alpha_0(\gamma) := 2^{-9/\gamma}$.
\end{proof}

Let $\df: \F_q \to \{0, 1\}$ be the complemented Kronecker-delta function, $\df(x) := \1\{x \neq 0\}$. We show that for small enough entropies, the entropy $H(\df(X))$ is comparable to $H(X)$:
\begin{lemma}
	\label{lem:delta}
	There exists a function $\alpha_1(\gamma)$ such that
	for any given $0 < \gamma < 1$, and any arbitrary random variable $X \in \F_q$ such that $0$ is the most likely value of $X$,
	$$\bH(X) \leq \alpha_1(\gamma) \implies \bH(X) \geq \frac{1}{\log q} H(\df(X)) \geq (1-\gamma)\bH(X).$$
\end{lemma}
\begin{proof}
	The first inequality $\bH(X) \log q = H(X) \geq H(\bar\delta(X))$ always holds, by the fact that deterministic postprocessing does not increase entropy.
	Thus, we will now show the second bound: that for small enough entropies, $\frac{1}{\log q} H(\bar\delta(X)) \geq (1-\gamma)\bH(X)$. This is equivalent with showing that $H(\df(X)) \geq (1 - \gamma)H(X)$.
	Given $\gamma$, let $\alpha_1 := \alpha_2(q,\gamma)$ be the entropy guaranteed by \cref{lem:ML}, so that
	if $\bH(X) \leq \alpha_1$ then $\Pr[\bar\delta(X)=1]=\Pr[X \neq 0] \leq \gamma \bH(X)$.
	Now, for $H(X) \leq \alpha_1$, we have 
	\begin{align*}
		H(X) &= H(X, \df(X)) - H(\df(X)|X) \tag{Chain rule}\\
		&= H(X, \df(X)) \tag{as $\df(X)|X$ is deterministic}\\
		&= H(\df(X)) + H(X | \df(X)) \tag{Chain rule}\\
		&= H(\df(X))  + H(X | \df(X) = 1)\Pr[\df(X) = 1] \tag{as $H(X|\df(X)=0) = 0$ since $X|\df(X)=0$ is deterministic}\\
		&\leq H(\df(X))  + \log(q)\Pr[\df(X) = 1] \tag{as $X \in \F_q$, so $H(X | \df(X) = 1)\le H(X) \leq \log(q)$}\\
		&\leq H(\df(X))  + \log(q) \gamma \bH(X) \tag{by \cref{lem:ML}} \\
		&= H(\df(X)) + \gamma H(X).
	\end{align*}
	Thus, if $H(X) \leq \alpha_1$, then
	$(1-\gamma)H(X) \leq H(\bar\delta(X))$ as desired.
\end{proof}

Now, by combining these, we can reduce suction-at-the-lower-end from $\F_q$ to the binary case.
\begin{proof}[Proof of \cref{lem:suc-lwr}]
	Given $\gamma$, we will set $\alpha \leq 1/4$, to be determined later.
	Notice that we have
	\begin{equation}
		\label{eqn:entrop-bar}
		\bH(X+Y) = \frac{1}{\log q} H(X + Y) \geq \frac{1}{\log 	q} H(\df(X+Y)).
	\end{equation}
	We will proceed to show first that
	\begin{equation}
		\label{eqn:delta-entropies} H(\df(X+Y)) \geq H(\df(X) \oplus \df(Y)).
	\end{equation}
	This inequality is justified by comparing the distributions of $\df(X+Y)$ and $\df(X) \oplus \df(Y)$, both binary random variables, and noticing that
	$$\Pr[\df(X+Y) = 0] = \Pr[X + Y = 0] \leq
	\Pr[\{X = 0, Y = 0\} \cup \{X \neq 0, Y \neq 0\}]
	=\Pr[\df(X) \oplus \df(Y) = 0].$$
	Moreover, let us observe that
	$\Pr[\df(X + Y) = 0]  = \Pr[X + Y = 0] \geq 1/2$. Indeed,
	\begin{equation*}
		\Pr[X + Y \neq 0] \leq H(X+Y) \leq H(X, Y) \leq H(X)+H(Y) \le 2\alpha \leq 1/2.
	\end{equation*}
	In the above, the second inequality follows since $X+Y$ is a deterministic function of $X,Y$ and the third inequality follows from the chain rule and the fact that conditioning can only decrease entropy.
	Therefore, by monotonicity of the binary entropy function $H(p)$ for $1/2 \leq p \leq 1$,
	and since $\P[\df(X + Y) = 0] \leq \P[\df(X) \oplus \df(Y) = 0]$
	we have
	\begin{align*}
		H(\df(X+Y) ) &\geq H(\df(X) \oplus \df(Y)).
	\end{align*}
	This justifies \cref{eqn:delta-entropies}.
	
	Now we conclude by using the suction-lemma in the binary case, applied to $\df(X) \oplus \df(Y)$.
	
	Let $\gamma'$ be a small enough constant, such that
	$(1-\gamma')^2 \geq (1-\gamma)$.
	Let $\alpha_0 := \alpha_0(\gamma')$ be the entropy bound provided by \cref{lem:suc-lwr-binary},
	and let $\alpha_1 := \alpha_1(\gamma')$ be the entropy bound provided by \cref{lem:delta}.
	Set $\alpha := \min\{\alpha_0, \alpha_1, 1/4\}$.
	
	Then, for $\bH(X), \bH(Y) \leq \alpha$, we have
	\begin{align*}
		\bH(X+Y) \log q & \geq H(\df(X + Y)) \tag{\cref{eqn:entrop-bar}} \\
		& \geq H(\df(X) \oplus \df(Y)) \tag{\cref{eqn:delta-entropies}} \\
		&\geq (1-\gamma')(H(\df(X)) + H(\df(Y)))
		\tag{\cref{lem:suc-lwr-binary} and $\bH(\df(Z))\le \bH(Z)$ for r.v. $Z$ }\\
		&\geq (1-\gamma')^2(\bH(X) + \bH(Y))\log q.
		\tag{\cref{lem:delta}}
	\end{align*}
	With our setting of $\gamma'$, this concludes the proof.
\end{proof}

We will now see how \cref{lem:suc-lwr} implies its strengthening for conditional entropies.
\begin{lemma}
	\label{lem:suc-lwr-sum}
	Let $(X_1, A_1)$ and $(X_2, A_2)$ be independent random variables with $X_i \in \F_q$, and such that $A_1, A_2$ are identically distributed, and moreover for every $a$ we have $\bH(X_1 | A_1 = a) = \bH(X_2 | A_2 = a)$. Then for every $\gamma > 0$, there exist $\tau$ such that if $\bH(X_1 |A_1) \leq \tau$, then
	\begin{equation}
		\bH(X_1 + X_2 | A_1, A_2) \geq (1 - \gamma)(\bH(X_1|A_1) + \bH(X_2|A_2)).
		\label{eq:suc-lwr-cond}
	\end{equation}
\end{lemma}
\begin{proof}
	Let us take $\alpha := \bH(X_1|A_1) = \bH(X_2|A_2)$. For given $\gamma$ we shall find $\tau$ such that if $\alpha < \tau$ then \cref{eq:suc-lwr-cond} is satisfied.  Let us now consider $G_A := \{ a : \bH(X_1 | A_1 = a) < \alpha_1 \}$, for $\alpha_1 = \frac{\alpha}{\gamma}$.
	(In the remainder of the proof when we want to talk about a random variable from the identical distribution from which $A_1$ and $A_2$ are drawn, we will denote it by $A$.) By Markov inequality
	\begin{equation*}
		\P(A \not\in G_A) \leq \frac{\alpha}{\alpha_1} = \gamma.
	\end{equation*}
	
	Let us fix now $\tau$ which appears in the statement of this lemma to be smaller than $\gamma$
	and moreover small enough so that when $\alpha < \tau$ for every $a_1, a_2 \in G_A$ we can apply \cref{lem:suc-lwr} to distributions $(X_1 | A_1 = a_1)$ and $(X_2 | A_2 = a_2)$ to ensure
	that $H(X_1 + X_2 | A_1 = a_1, A_2 = a_2) \geq (1 - \gamma)(H(X_1 | A_1 = a_1) + H(X_2 | A_2 = a_2))$.
	
	Let us use shorthand $S(a_1, a_2) = \bH(X_1 + X_2 | A_1 = a_1, A_2 = a_2) \P(A_1 = a_1, A_2 = a_2)$. We have
	\begin{align}
		\bH(X_1 + X_2 | A_1, A_2) & = \sum_{a_1, a_2} S(a_1, a_2) \nonumber \\
		& \geq \sum_{\substack{a_1 \in G_A\\ a_2 \in G_A}} S(a_1, a_2) + \sum_{\substack{a_1 \not\in G_A \\ a_2\in G_A}} S(a_1, a_2) + \sum_{\substack{a_1\in G_A\\ a_2 \not\in G_A}} S(a_1, a_2).
		\label{eq:three-terms}
	\end{align}
	
	If both $a_1$ and $a_2$ are in $G_A$, then by \cref{lem:suc-lwr} we have
	\begin{equation*}
		S(a_1, a_2) \geq (1 - \gamma) (\bH(X_1 | A_1 = a_1) + H(X_2 | A_2 = a_2)) \P(A_1 = a_1, A_2 = a_2),
	\end{equation*}
	therefore
	\begin{equation}
		\sum_{a_1 \in G_A, a_2 \in G_A} S(a_1, a_2) \geq 2 (1 - \gamma) \P(A \in G_A) \sum_{a_1 \in G_A} H(X_1 | A_1 = a_1) \P(A_1 = a_1),
		\label{eq:first-term}
	\end{equation}
	where in the above we have used the fact that $A_1$ and $A_2$ are identically distributed.
	
	On the other hand, for $a_1 \not\in G_A, a_2 \in G_A$ let us  bound
	\begin{align*}
		S(a_1, a_2) & = \bH(X_1 + X_2 | A_1 = a_1, A_2 = a_2) \P(A_1 = a_1, A_2 = a_2) \\
		& \geq \bH(X_1 + X_2 | A_1 = a_1, A_2 = a_2, X_2) \P(A_1 = a_1, A_2 = a_2) \\
		& = \bH(X_1 | A_1 = a_1) \P(A_1 = a_1, A_2 = a_2)
	\end{align*}
	where the inequality follows from the fact that additional conditioning decreases entropy and for the second equality we  used the fact that since $X_1$ and $X_2$ are independent, $\bH(X_1 + X_2 | A_1 = a_1, A_2 = a_2, X_2)= \bH(X_1 | A_1 = a_1, A_2 = a_2, X_2)= \bH(X_1 | A_1 = a_1, A_2 = a_2)=\bH(X_1 | A_1 = a_1)$.
	Summing this bound over all such pairs yields
	\begin{equation}
		\sum_{a_1 \not\in G_A, a_2 \in G_A} S(a_1, a_2) \geq \P(A \in G_A) \sum_{a_1\not\in G_A} \bH(X_1 | A_1 = a_1)\P(A_1 = a_1)
		\label{eq:second-term}
	\end{equation}
	and symmetrically for the third summand, we get
	\begin{equation}
		\sum_{a_1 \in G_A, a_2 \not\in G_A} S(a_1, a_2) \geq \P(A \in G_A) \sum_{a_2\not\in G_A} \bH(X_2 | A_2 = a_2)\P(A_2 = a_2).
		\label{eq:third-term}
	\end{equation}
	
	Plugging in~\cref{eq:first-term,eq:second-term,eq:third-term}  into \cref{eq:three-terms} (and using the fact that $A_1$ and $A_2$ are identically distributed) we find
	\begin{align*}
		\bH(X_1 + X_2 | A_1, A_2) & \geq 2(1 - \gamma) \P(A_1 \in G_A)\sum_{a_1} \bH(X_1 | A_1 = a_1)\P(A_1 = a_1) \\
		& = 2 (1 - \gamma)\P(A \in G_A) \bH(X_1 | A_1).
	\end{align*}
	We have $\P(A \in G_A) \geq (1 - \gamma)$, which yields
	\begin{equation*}
		\bH(X_1 + X_2 | A_1, A_2) \geq 2(1 - \gamma)^2 \alpha  \geq 2(1 - 2\gamma) \alpha
	\end{equation*}
	and the statement of the lemma follows, after rescaling $\gamma$ by half.
\end{proof}

Finally, we are ready to prove \cref{lem:suc-lwr-cond}:
\begin{proof}[Proof of \cref{lem:suc-lwr-cond}]
	By chain rule we have
	\begin{align*}
		\bH(X_1 ~|~ X_1 + X_2, A_1, A_2) & = \bH(X_1, X_1 + X_2 ~|~ A_1, A_2) - \bH(X_1 + X_2 ~|~ A_1, A_2) \\
		& = \bH(X_1, X_2 ~|~ A_1, A_2) - \bH(X_1 + X_2 ~|~ A_1, A_2) \\
		& = 2 \bH(X_1 ~|~ A_1) - \bH(X_1 + X_2 ~|~ A_1, A_2),
	\end{align*}
	where the last equality follows from the independence of $(X_1,A_1)$ and $(X_2,A_2)$.
	Now we can apply \cref{lem:suc-lwr-sum} to get
	\begin{align*}
		\bH(X_1 ~|~ X_1 + X_2, A_1, A_2) & \leq 2 \bH(X_1 ~|~ A_1) - (1 - \gamma)(2 \bH(X_1 ~|~ A_1) = 2 \gamma \bH(X_1 ~|~ A_1)
	\end{align*}
	and the statement follows directly from \cref{lem:suc-lwr-sum} and rescaling $\gamma$ by half.
\end{proof}

\section{Exponential matrix  polarization}\label{sec:exp-local-arikan}

The main result of this section shows the exponential matrix polarization of $M^{\otimes 2}$ for every mixing matrix.

\begin{lemma}
	\label{lem:every-matrix-works}
	For every prime $p$, every mixing matrix $M \in \F_q^{k \times k}$ and every $\varepsilon > 0$, the matrix $M^{\otimes 2}$ satisfies $(\frac{1}{k^2}, 2-\varepsilon)$-exponential matrix polarization.
\end{lemma}

Before turning to the proof we first note that this immediately yields \cref{thm:triangle-local-exp}. 

\begin{proof}[Proof of \cref{thm:triangle-local-exp}]
	By \cref{lem:every-matrix-works} we have that for every prime $q$ and mixing matrix $M\in\F_q^{k \times k}$, the matrix $M^{\otimes 2}$ satisfies $(\frac{1}{k^2}, 2-\varepsilon)$-exponential matrix polarization. By \cref{lem:polarizing-matrix-implies-martingale} we then have that for every symmetric memoryless channel $\C_{Y|Z}$, the Arikan martingale associated with $M^{\otimes 2}$ and $\C_{Y|Z}$ is $(\frac{1}{k^2}, 2-\varepsilon)$-exponentially locally polarizing.
\end{proof}

The rest of the section is devoted to the proof of \cref{lem:every-matrix-works}. 
We start with a simple proposition.

\begin{proposition}\label{prop:m2mixes}
	For every field $\F_q$ and every matrix $M \in \F_q^{k\times k}$, its tensor $M^{\otimes 2}$ is mixing if $M$ is mixing.
\end{proposition}

\begin{proof}
	Let $S_j = \{i \in [k] | \exists j' \in [j] \mbox{ s.t. } M_{i,j'}\ne 0\}$ then $\exists j$ s.t. $|S_j| > j$.
	By \cref{lem:mixing-equiv}, there exists a $j$ such that $|S_j|>j$.
	With this observation, the proposition follows easily. Given mixing $M$, let $j$ be the index such that $|S_j| > j$. Recall that $M^{\otimes 2}$ is composed of $k^2$ submatrices of dimensions $k\times k$ each, with the $i,j$th submatrix being $M_{ij}\cdot M$. Let $i$ be an index such that $M_{i1}\ne 0$. (Such an index must exist or else we have an all zero column which contradicts the invertibility of $M$.) Then the first $k$ columns of $M^{\otimes 2}$ contain the $k\times k$ submatrix $M_{i1}\cdot M$ and in this submatrix itself we have that the support of the first $j$ columns has size larger than $j$. We conclude the first $j$ columns of $M^{\otimes 2}$ have support size larger than $j$ and so by \cref{lem:mixing-equiv}, $M^{\otimes 2}$ is mixing.
\end{proof}

\subsection{Exponential polarization of a $2 \times 2$ matrix}
We will first prove that a single specific matrix, namely $\HadMat$, after taking second Kronecker power satisfies exponential polarization. Recall that in \cref{sec:local-k-by-k} the local polarization of a mixing matrix was shown essentially by reducing to this case. We will follow a similar plan in this section.

\begin{lemma}
	Let $q$ be a prime and let $H = \HadMat$ for $\alpha \in \F_q^*$. Then, for every $\varepsilon > 0$, the matrix $H^{\otimes 2}$ satisfies $(\frac{1}{4}, 2-\varepsilon)$ exponential matrix polarization.
	\label{lem:4-by-4}
\end{lemma}

\begin{proof}
	Note that since $H$ is mixing, by \cref{prop:m2mixes}, we have that $H^{\otimes 2}$ is also mixing. And so by \cref{lem:k-by-k-polarize}, we have that $H^{\otimes 2}$ satisfies the conditions of matrix polarization (specifically, variance in the middle and suction at the upper and lower ends from \cref{def:matrix-polarization}). It remains only to argue exponential matrix polarization, i.e., strong suction at the ends. 
	
	Given $\varepsilon > 0$ let $\tau>0$ be such that for every $\delta < \tau$ we have $6 (\log \frac{1}{3 \delta^{2}} + \log q) \leq \delta^{-\varepsilon}$.
	Note that the identity is satisfied for small enough $\delta$ since the LHS is $O\parens{\log\parens{\frac 1\delta}}$ while the RHS is $\Omega\parens{\parens{\frac 1\delta}^\eps}$.
	Now let $\delta < \tau$ and now consider arbitrary sequence of i.i.d. random variables $(\bvec{U}_1, W_1), \ldots (\bvec{U}_4, W_4)$ with $H(\bvec{U}_i | W_i) = \delta$. We can explicitly write down matrix $H^{\otimes 2}$ as
	\begin{equation*}
		H^{\otimes 2} = \left[\begin{array}{cccc} 1 & 0 & 0 & 0 \\ \alpha & 1 & 0 & 0 \\ \alpha & 0 & 1 & 0 \\ \alpha^2 & \alpha & \alpha & 1\end{array} \right].
	\end{equation*}
	Matrix $H^{\otimes 2}$ has four rows. So, to achieve exponential polarization with $\eta = \frac{1}{4}$, we need to show that there is at least one index $i$ satisfying the strong suction inequality (with parameter $b = 2 - \varepsilon$). We do so for $i = 4$.
	Let us consider vector $\bvec{U} = (\bvec{U}_1, \ldots, \bvec{U}_4)$ and similarly $W = (W_1, \ldots, W_4)$, and let $\bvec{V} = (\bvec{V}_1, \ldots, \bvec{V}_4) = \bvec{U}\cdot H^{\otimes 2}$. We want to bound
	\begin{align*}
		\H( \bvec{V}_4 | \bvec{V}_{<4}, W) & = \H(\bvec{U}_4 | \bvec{U}_1 + \alpha \bvec{U}_2 + \alpha \bvec{U}_3 + \alpha^2 \bvec{U}_4, \bvec{U}_2 + \alpha \bvec{U}_4, \bvec{U}_3 + \alpha \bvec{U}_4, W) \\
		& = \H(\bvec{U}_4 | \bvec{U}_1 - \alpha^2 \bvec{U}_4, \bvec{U}_2 + \alpha \bvec{U}_4, \bvec{U}_3 + \alpha \bvec{U}_4, W), 
	\end{align*}
	where the equality follows since $\bvec{U}_1 - \alpha^2 \bvec{U}_4=\bvec{U}_1 + \alpha \bvec{U}_2 + \alpha \bvec{U}_3 + \alpha^2 \bvec{U}_4 - \alpha\parens{\bvec{U}_2 + \alpha \bvec{U}_4} -\alpha\parens{\bvec{U}_3 + \alpha \bvec{U}_4}$ and hence the map 
	\[\parens{\bvec{U}_1 + \alpha \bvec{U}_2 + \alpha \bvec{U}_3 + \alpha^2 \bvec{U}_4, \bvec{U}_2 + \alpha \bvec{U}_4, \bvec{U}_3 + \alpha \bvec{U}_4}\mapsto \parens{\bvec{U}_1 - \alpha^2 \bvec{U}_4, \bvec{U}_2 + \alpha \bvec{U}_4, \bvec{U}_3 + \alpha \bvec{U}_4} \] is a bijection.
	
	The main idea to bound the conditional entropy of $\bvec{U}_4$ above is that if any of $\bvec{U}_i$ is `known' for $i \in \{1,2,3\}$, then given the variables being conditioned on, $\bvec{U}_4$ is also `known'. Of course, none of the $\bvec{U}_i$'s are known, but each is predictable given $W_i$ and we use this predictability to bound the conditional entropy. Details follow. 
	
	Let $\Sigma$ denote the domain of $W_i$'s. Using $\H(\bvec{U}_i|W_i) = \delta$, by \cref{lem:ML-conditional}, we have that there exists some function $f: \Sigma \to \F_q$, such that $\P(f(W_i) \not= \bvec{U}_i) \leq \delta$. Let $\bvec{V}'_1 := -\alpha^2 \bvec{U}_4 + \bvec{U}_1$. We now give a predictor
	$g(\bvec{V}'_1,\bvec{V}_2,\bvec{V_3},W)$ for $\bvec{U}_1$. Let 
	$X_1 = -\alpha^{-2} (\bvec{V}'_1 - f(W_1))$, $X_2 = \alpha^{-1} (\bvec{V}_2 - f(W_2))$, and $X_3 = \alpha^{-1} (\bvec{V}_3 - f(W_3))$. Note that if for some $i$ we have $f(W_i) = \bvec{U}_i$ then we have $X_i = \bvec{U}_4$. Using this we set $g$ as follows: If  two of $X_1, X_2, X_3$ have the same value, we define $g(\bvec{V}'_1,\bvec{V}_2,\bvec{V_3},W)$  to be this value, otherwise we set $g(\bvec{V}'_1,\bvec{V}_2,\bvec{V_3},W)$ arbitrarily.
	
	By construction of $g$ we have that  if there exist two choices of $i \in \{1,2,3\}$ satisfying $f(W_i) = \bvec{U}_i$, then $g(\bvec{V}'_1,\bvec{V}_2,\bvec{V_3},W)=\bvec{U}_4$.
	In turn this implies $\P(g(\bvec{V}'_1,\bvec{V}_2,\bvec{V_3},W) \not= \bvec{U}_4) \leq 3 \delta^2$ since  by symmetry, we have \[\P(g(\bvec{V}'_1,\bvec{V}_2,\bvec{V_3},W) \not= \bvec{U}_4) \leq 3 \P(f(W_1) \not= \bvec{U}_1 \land f(W_2) \not= \bvec{U}_2) = 3 \P(f(W_1) \not= \bvec{U}_1))^2 \leq 3 \delta^2,\]
	where the equality follows since $(\bvec{U}_i,\bvec{W_i})$ are independent.
	
	Converting the predictability of $\bvec{U}_1$ by $g(\cdots)$ into an entropy bound by Fano's inequality \cref{lem:prediction-gives-entropy}, we have $\H(\bvec{U}_4 | \bvec{U}_1 - \alpha^2 \bvec{U}_4, \bvec{U}_2 + \alpha \bvec{U}_4, \bvec{U}_3 + \alpha \bvec{U}_4, W) \leq 6 \delta^2 (\log \frac{1}{3 \delta^{2}} + \log q)$. By the choice of $\tau$ and $\delta < \tau$ we have $6 (\log \frac{1}{3\delta^{2}} + \log q) \leq \delta^{-\varepsilon}$ 
	and so  
	\begin{equation}
		\label{eq:lem-7.3-proof-bound}
		\H( \bvec{V}_4 | \bvec{V}_{<4}, W) \leq \delta^{2 - \varepsilon} = \parens{\H(\bvec{U}_1|W_1)}^{2-\eps},
	\end{equation}
	as desired.
\end{proof}

\subsection{Exponential polarization of any mixing matrix via useful containment}

We will now proceed to show that exponential polarization of $M^{\otimes 2}$ for any mixing matrix $M$ can be reduced to the lemma above. 
We first provide  an intuitive explanation of the reasoning below.

In order to show that a matrix $M'$ satisfies an exponential polarization (or just \textit{suction at the lower end} condition of local polarization), one needs to show that for any i.i.d. variables $U_i$ with entropy $H(U_i) = \delta$ and some index $j$, we can upper bound $\H( (\bvec{U}M')_j | (\bvec{U}M')_{<j})$ (for the sake of the clarity of this exposition, we skip conditioning on $W_i$). If we write $\bvec{V}_i = (\bvec{U}M)_i$, we wish to upper bound $\H(\bvec{V}_j | \bvec{V}_{1}, \ldots \bvec{V}_{j-1})$, (where all $\bvec{V}_i$ are  linear forms in $\{\bvec{U}_i\}_{i \in [k]}$). Now, for any $\bvec{B}_1, \ldots,\bvec{B}_m$
that all can be expressed as linear combinations of $\bvec{V}_1, \ldots \bvec{V}_{j-1}$, we have
\[ \H(\bvec{V}_j | \bvec{V}_{1} \ldots \bvec{V}_{j-1}) = \H(\bvec{V}_j + \bvec{B}_m | \bvec{V}_1, \ldots \bvec{V}_{j-1}, \bvec{B}_1, \ldots \bvec{B}_{m-1}) \leq \H(\bvec{V}_j + \bvec{B}_m | \bvec{B}_1, \ldots, \bvec{B}_{m-1}) \ . \]

In \cref{sec:reduction} we showed using Gaussian elimination that for any mixing matrix $M$, one can find $j,\ell,s$, and linear forms $W_1, W_2$ s.t. $V_j + W_2 = \bvec{U}_{\ell}$ and $W_1 = \alpha \bvec{U}_{\ell} + \bvec{U}_{s}$, which implied $\H(\bvec{V}_j + W_2 | W_1) = \H(\bvec{U}_{\ell} | \alpha \bvec{U}_{\ell} + \bvec{U}_{s})$. This can be thought of as showing that in some sense any mixing matrix $M$ \textit{contains} a matrix $H = \HadMat$, and reduces the problem of showing the local polarization of the former, to understanding local polarization of the latter. 

Here we introduce a technical notion of \textit{useful containment} that is tailored to extend this reasoning in a way that has a convenient property expressed by \cref{lem:useful-square} --- i.e. since matrix $M$ contains $H$ in this specific sense, the matrix $M^{\otimes 2}$ contains $H^{\otimes 2}$ and by the reasoning outlined in the previous paragraph, we can deduce exponential local polarization of $M^{\otimes 2}$ from this containment and the entropy upper bound proved in \cref{lem:4-by-4}.

We wish to note here that the subsequent definition and lemmas are tailored to the specific statement we are proving. In particular \textit{useful containment} is not a transitive relation. More importantly, and unfortunately, it is not true that for any exponentially polarizing matrix $R$, if $R$ is usefully contained in $M$, than $M$ is exponentially polarizing. \cref{lem:useful-usefulness} asserts this property only for $R = H^{\otimes 2}$.

The following definition of containment relation for matrices will be used to implement the ideas outlined above. 

\begin{definition}[Matrix (useful) containment]
	For any finite field $\F_q$ and integers $k\ge m\ge 1$,
	we say that a matrix $M \in \F_q^{k\times k}$ contains a matrix $R \in \F_q^{m\times m}$, if there exist some $T \in \F_q^{k\times m}$ and a permutation matrix $P \in \F_q^{k \times k}$, such that $PMT = \left[\begin{array}{c} R \\ 0 \end{array}\right]$. We say that $P$ and $T$ {\em witness} the containment of $R$ in $M$. If moreover the last non-zero row of $T$ is scaling of the standard basis vector, i.e. $T_j = \alpha e_m$ for some $\alpha\in\F_q^*$, we say that containment is $R$ in $M$ is useful and we denote it by $R \usef M$. 
	
\end{definition}
We emphasize that useful containment is \emph{not} a partial order.

Comparing this definition to the exposition above, the permutation $P$ is used to express the fact that we can freely permute labels of variables $\bvec{U}_1, \ldots, \bvec{U}_k$, whereas the matrix $T$ encodes coefficients for linear forms $\bvec{B}_1, \ldots \bvec{B}_{m-1}, \bvec{B}_m + \alpha \bvec{V}_j$. Finally, the condition on the last non-zero row of $T$ being of form $\alpha e_m$ is here to express the idea that $\bvec{V}_j$ is not allowed to appear in any of the forms $\bvec{B}_1, \ldots \bvec{B}_m$.

The following fact about useful containment will be helpful.
\begin{proposition}
	If $R \usef M$, then for every upper triangular matrix $U$ with non-zero diagonal elements $U_{i,i}$, we also have $R \usef MU^{-1}$. 
	\label{claim:useful-upper}
\end{proposition}
\begin{proof}
	Consider matrix $T$ and permutation $P$ as in the definition of useful containment for $R \usef M$. We can pick the very same permutation $P$ and matrix $T' = UT$ to witness $R \usef MU^{-1}$. All we have to show is that last non-zero row of $T'$ is the (scaled) standard basis vector $\alpha\bvec{e}_m$. Indeed, if $j_0$ is the last non-zero row of $T$, and $j > j_0$, rows $(U)_{j}$ are supported exclusively on elements with indices larger than $j_0$, hence $(UT)_{j} = (U)_j T = 0$. On the other hand $(U T)_{j_0} = \sum_{i} U_{j_0, i} T_i$. Since for $i < j_0$ the entry $U_{j_0, i} = 0$, and for $i > j_0$ we have $T_i=0$, this implies $(U T)_{j_0} = U_{j_0, j_0} T_{j_0} = U_{j_0, j_0} \alpha \bvec{e}_m$, where the last equality follows from the fact that $T$ was useful --- that is $T_{j_0} = \alpha \bvec{e}_m$ and $T_i = 0$ for $i > j_0$. Since both $U_{j_0, j_0}\ne 0$ and $\alpha\ne 0$, we have $U_{j_0, j_0} \alpha\neq 0$, as desired.
\end{proof}

\cref{lem:red-lower} can now be reinterpreted as the following lemma. We give a full new proof here, as we describe it now in the language of useful containment. 

\begin{lemma}\label{lem:everyMcontainsH}
	Every mixing matrix $M \in \F_q^{k\times k}$ usefully contains matrix $H = \HadMat$, for some $\alpha\in\F_q^*$.
\end{lemma}

\begin{proof}
	For every matrix $M$, there is some permutation matrix $P'$ and pair $L, U$, such that $P'M = L U$ where $L$ is lower triangular (such that its diagonal is all $1$s), and $U$ is upper triangular.\footnote{This e.g. follows from Gaussian Elimination and the corresponding ``$LU$ decomposition'' of any matrix. Also note that the the assumption on the diagonal elements of $L$ holds without loss of generality.} 
	Matrix $M$ being mixing is equivalent to the statement that $L$ and $U$ are invertible, and moreover $L$ is not diagonal. (In particular $M$ is invertible if and only if $L$ and $U$ are and $M = (P')^{-1}LU$ is the permutation of an upper triangular matrix if and only if $L$ is diagonal.) 
	Thus by \cref{claim:useful-upper} it suffices to show that every lower-triangular $L$, which is not diagonal, contains $H
	\usef L$. 
	Indeed, let $s$ be the last column of $L$ that contains more than a single non-zero entry, and let $r$ to be the last row of non-zero entry in column $L_{\cdot, s}$. Note that column $L_{\cdot, r}$ has single non-zero entry $L_{r,r} = 1$. We will show a matrix $T \in \F_q^{k \times 2}$ as in the definition of useful containment. Let us specify a second column of $T_{\cdot, 2} := \bvec{e}_r$-- note that in this case $(LT)_{\cdot,2}=\bvec{e}_r$. To specify the first column of $T$ we wish to find a linear combination of columns of $L_{1, \cdot}, \ldots, L_{r-1, \cdot}$ such that $\sum_{i\leq r-1} t_{i} L_{i, \cdot} = \bvec{e}_s + \alpha \bvec{e}_r$, where $\alpha=L_{r,s}\ne 0$. Then coefficients $t_i$ can be used as the first column of matrix $T$, which would imply that $(LT)_{\cdot,1}=\bvec{e}_s + \alpha \bvec{e}_r$. We can set those coefficients to $t_{i} = - L_{s, i}$ for $i \in [s+1, r-1]$, $t_s = 1$ and $t_i=0$ for $i<s$ --- this setting is correct, because columns $L_{i,\cdot}$ for $i \in [s+1, r-1]$ has only one non-zero entry $L_{i,i}$. As already observed the first column of $LT$ is $\bvec{e}_s + \alpha \bvec{e}_r$ while the second column is $\bvec{e}_r$. Thus, if $P$ is any matrix corresponding to a permutation which maps $s \mapsto 1$ and $r \mapsto 2$, the containment $H \usef L$ is witnessed by pair $P$ and $T$, as desired.
\end{proof}

\begin{lemma}
	\label{lem:useful-square}
	If matrix $R \usef M$ where $R\in \F_q^{s\times s}$ and $M \in \F_q^{k \times k}$, then $R^{\otimes 2} \usef M^{\otimes 2}$.
\end{lemma}
\begin{proof}
	Consider matrix $T$ and permutation $P$ that witness the useful containment for $R\usef M$. Note that by the mixed product property of tensors, $P^{\otimes 2}M^{\otimes 2} T^{\otimes 2} = (PMT)^{\otimes 2}$. 
	As such, restriction of a matrix $P^{\otimes 2} M^{\otimes 2} T^{\otimes 2}$ to rows corresponding to $[k]\times [k]$ is exactly $R^{\otimes 2}$, and all remaining rows are zero. We can apply additional permutation matrix $\tilde{P}$ so that those are exactly first $k^2$ rows of the matrix $\tilde{P} P^{\otimes 2} M^{\otimes 2} T^{\otimes 2}$ give matrix $R^{\otimes 2}$, and the remaining rows are zero.  Finally, since the last non-zero row of $T$ was a scaling of the standard basis vector, the same is true for $T^{\otimes 2}$.
\end{proof}

\begin{lemma}
	\label{lem:useful-usefulness}
	If matrix $M\in\F_q^{k \times k}$ usefully contains matrix $R = \HadMat^{\otimes 2}$, then matrix $M$ satisfies the strong suction condition of $(\frac{1}{k}, 2 - \varepsilon)$ exponential polarization.
\end{lemma}
\begin{proof}
	By the definition of exponential matrix polarization it suffices to prove that there exists an index $j \in [k]$ such that $\H( (\bvec{U}M)_{j} | (\bvec{U}M)_{<j}, W) \leq \H( (\bvec{U}R)_4 | (\bvec{U}R)_{< 4}, W)$. Once we have this, the proof of \cref{lem:4-by-4} (specifically \cref{eq:lem-7.3-proof-bound}) asserts that the conditional entropy is bounded as required.
	So we turn to proving this.
	
	Take $P \in \F_q^{k\times k}$ and  $T \in \F_q^{k\times 4}$ witness the containment $R \usef M$. Let moreover $j$ be the last non-zero row of $T$. We have 
	\begin{align*}
		\H( (\bvec{U}M)_{j} | (\bvec{U}M)_{<j}, W) & = \H( (\bvec{U}M)_{j} T_{j, 4} + (\bvec{U}M)_{<j} T_{<j, 4} | (\bvec{U}M)_{<j}, W) \\
		& = \H( (\bvec{U}MT)_4 | (\bvec{U}M)_{<j}, W) \\
		& = \H( (\bvec{U}MT)_4 | (\bvec{U}M)_{<j}, (\bvec{U}M)_{<j} T_{<j, <4}, W)\\
		& \leq \H( (\bvec{U}MT)_4 | (\bvec{U}M)_{<j} T_{<j, <4}, W).
	\end{align*}
	In the above the first equality follows since $T_{j, 4}\ne 0$ (and hence the map $(\bvec{U}M)_{j} \mapsto (\bvec{U}M)_{j} T_{j, 4}$ is a bijection) and the fact that $(\bvec{U}M)_{<j} T_{<j, 4}$ is deterministic function of $(\bvec{U}M)_{<j}$. The second equality follows since $M_{>j,\cdot}=\bvec{0}$, the third one introduces conditioning on $(\bvec{U}M)_{<j} T_{<j, <4}$ which is deterministic given $(\bvec{U}M)_{<j}$, and the inequality follows because entropy is decreasing under additional conditioning. 
	Observe now that $(\bvec{U}M)_{<j} T_{<j, <4} = (\bvec{U}MT)_{<4}$. Indeed --- according to the definition of useful containment and because $j$ is last non-zero row of $T$, we have $T_{j, <4} = 0$ ($j$-th row has only one non-zero entry $T_{j, 4}$), as well as $T_{>j, <4} = 0$. Therefore
	\begin{align*}
		\H( (\bvec{U}M)_{j} | (\bvec{U}M)_{<j}, W) & \leq \H( (\bvec{U}MT)_4 | (\bvec{U}MT)_{<4}, W) \\
		& = \H( (\bvec{U} P^{-1} R)_4 | (\bvec{U}P^{-1} R)_{<4}, W) \\
		& = \H( (\bvec{U}R)_4 | (\bvec{U}R)_{< 4}, W),
	\end{align*}
	where the last inequality follows from the fact that variables $\bvec{U}_i$ are i.i.d. hence for the permutation matrix $P$, $\bvec{U} P^{-1}$ has the same distribution as $\bvec{U}$. 
\end{proof}

With the above ingredients in place we are ready to prove \cref{lem:every-matrix-works}.

\begin{proof}[Proof of \cref{lem:every-matrix-works}]
	Since $M$ is mixing we have that $M^{\otimes 2}$ is also mixing (\cref{prop:m2mixes}) and so by \cref{lem:k-by-k-polarize} we have that $M^{\otimes 2}$ satisfies the conditions of matrix polarization. So it suffices to prove $M^{\otimes 2}$ satisfies the conditions of $(\frac{1}{k^2},2-\epsilon)$ exponential matrix polarization. 
	
	By \cref{lem:everyMcontainsH} we have that $M$ usefully contains $H=\HadMat$. Then, by \cref{lem:useful-square} we have that $M^{\otimes 2}$ usefully contains $H^{\otimes 2}$. Finally by \cref{lem:useful-usefulness} applied to $M^{\otimes 2}$ (which is a $k^2 \times k^2$ matrix) we have that $M^{\otimes 2}$ satisfies $(1/k^2,2-\epsilon)$ exponential matrix polarization. 
\end{proof}

\section{Nearly optimal decoding error probabilities \label{sec:lift}}

Finally we turn to the proofs of \cref{thm:thm2,thm:thm3}. Recall that the former yields codes achieving decoding error probability $\exp(-N^{\beta})$ for any $\beta<1$ while doing so at block lengths polynomial in the gap to capacity. The latter result shows that the techniques in this paper are essentially optimal (for a broad class of channels) by showing that any analysis that bounds the decoding error probability can be used as a black box to achieve a similar decoding error probability in our analysis framework while additionally guaranteeing convergence at polynomial lengths in the gap to capacity. 
We first present the former, though before doing so, we make a small digression to recollect some known definitions of linear codes that we will use in this section (for more details see e.g.~\cite[Chap. 2]{essential-coding}).

\subsection{Basics of linear error-correcting codes}

A linear $q$-ary error correcting code $C$ of block length $n_0$ and dimension $k_0$ is a linear subspace of $\F_q^{n_0}$ of dimension $k_0$. Equivalently, there exists a full rank $G\in\F_q^{k_0\times n_0}$ such that $C=\set{\bvec{v}\cdot G|\bvec{v}\in\F_q^{k_0}}$-- $G$ is called the generator matrix of $C$. The kernel/null-space/dual of $C$, denoted by $C^{\perp}$ or $\ker G$, is given by $\set{\bvec{w}| \inprod{\bvec{w},\bvec{c}}=0\text{ for all }\bvec{c}\in C}$. A generator matrix of $C^{\perp}$ is called a parity-check matrix of $C$. The distance of a code $C$ is the minimum number of positions any two codewords in $C$ differ in. For linear code $C$, its distance is exactly $\min_{\bvec{c}\in C\setminus \set{\bvec{0}}} \wt(\bvec{c})$, where $\wt(\bvec{x})$ is the number of non-zero elements in $\bvec{x}$.

\subsection{Polar codes with decoding failure probability approaching $2^{-N^{1-o(1)}}$}
\label{sec:strong-exp-failure}

\cref{thm:thm2} is proved by giving a sufficient structural condition on matrices for very strong exponential polarization. The following lemma states this condition.

\begin{lemma}
	\label{lem:code-exp-suction} Let $q$ be prime.
	If a mixing matrix $M\in \F_q^{k\times k}$ is decomposed as $M = \left[ M_0 | M_1 \right]$, where $M_0 \in \F_q^{k \times (1 - \eta) k}$ is such that $\ker M_0^T$ is a linear code of distance larger than $2 b$, then matrix $M$ satisfies $(\eta, b - \varepsilon)$-exponential matrix polarization for every $\varepsilon > 0$.
\end{lemma}

\begin{proof}
	By \cref{lem:k-by-k-polarize}, we have that $M$ satisfies the conditions of matrix polarization (specifically, variance in the middle and suction at the upper and lower ends from \cref{def:matrix-polarization}). It remains only to argue exponential matrix polarization, i.e., strong suction at the lower end. 
	
	Let us again consider a sequence of i.i.d. pairs $(\bvec{U}_i, W_i)$ for $i \in [k]$, such that $H(\bvec{U}_i | W_i) = \delta$. By \cref{lem:ML-conditional}, there is some $f : \Sigma \to \F_q$ such that $\P(f(W_i) \not= \bvec{U}_i) \leq \delta$ (for every $i\in [k]$). Let us define $\tilde{\bvec{U}}_i := \bvec{U}_i - f(W_i)$.
	
	We will bound $\H( (\bvec{U}M)_j | (\bvec{U}M)_{<j}, W)$, for all $j > (1 - \eta) k$. We have 
	\begin{equation*}
		\H( (\bvec{U}M)_j | (\bvec{U}M)_{<j}, W) \leq \H(\bvec{U} |(\bvec{U}M)_{<j}, W) \leq \H(\bvec{U} | \bvec{U}M_0, W) = H(\tilde{\bvec{U}} | \tilde{\bvec{U}}M_0, W) \leq \H(\tilde{\bvec{U}} | \tilde{\bvec{U}} M_0),
	\end{equation*}
	where the first two inequalities follow from the fact that for random variables $(X, Y, S, T)$ it is always the case that $\H( X | S, T) \leq \H(X, Y | S, T) \leq \H(X, Y | S)$ (the second inequality also uses the fact that $(\bvec{U}M)_{<j}$ is a sub-matrix of $\bvec{U}M_0$). The equality follows from the definition of $\tilde{\bvec{U}}_i$ and the fact that $f(\cdot)$ is  deterministic function. The final inequality follows from the fact that conditioning can only decrease the entropy.
	
	Given $\tilde{\bvec{U}} M_0$ we can produce estimate $\hat{\bvec{U}} := \argmin_{\bvec{V}} \{ \wt(\bvec{V}) : \bvec{V} M_0 = \tilde{\bvec{U}} M_0\}$, where $\wt(\bvec{V}) = |\{j : \bvec{V}_j \not=0 \}|$.
	
	We note that if $\wt(\tilde{\bvec{U}}) \leq b$ then $\hat{\bvec{U}} = \tilde{\bvec{U}}$. Indeed, we have $\wt(\hat{\bvec{U}}) \leq \wt(\tilde{\bvec{U}})$, therefore $\wt(\hat{\bvec{U}} - \tilde{\bvec{U}}) \leq 2 \wt(\tilde{\bvec{U}}) \leq 2 b$, but on the other hand $(\hat{\bvec{U}} - \tilde{\bvec{U}}) M_0 = 0$, and by the assumption on distance of $\ker M_0^T$ we deduce that $\hat{\bvec{U}} - \tilde{\bvec{U}} = 0$. Therefore $\P(\tilde{\bvec{U}} \not= \hat{\bvec{U}}) \leq \P(\wt(\tilde{\bvec{U}}) > b)$. All coordinates of $\tilde{\bvec{U}}$ are independent, and each $\tilde{\bvec{U}}_i$ is nonzero with probability at most $\delta$, therefore
	\begin{equation*}
		\P(\wt(\tilde{\bvec{U}}) > \beta_1) \leq \binom{k}{b} \delta^{b}.
	\end{equation*}
	Further, by Fano inequality (\cref{lem:prediction-gives-entropy}), we have
	\begin{equation*}
		H(\tilde{\bvec{U}} | \tilde{\bvec{U}} M_0) \leq 2 C \delta^{b} (b \log \delta^{-1} +  \log C + \log q )
	\end{equation*}
	where $C = \binom{k}{b}$. Again, for any $\varepsilon$, and small enough $\delta$ (with respect to $\varepsilon, b, k, q$), we have $H(\tilde{\bvec{U}} | \tilde{\bvec{U}} M_0) \leq \delta^{b - \varepsilon}$.
	
	This shows that for any $j > (1 - \eta) k$ (note that there are at least $\eta k$ such values of $j$) and small enough $\delta$ we have
	\begin{equation*}
		\H( (\bvec{U}M)_j | (\bvec{U}M)_{<j}, W) \leq \delta^{b - \varepsilon},
	\end{equation*}
	which completes the proof of a exponential matrix polarization for $M$.
\end{proof}

We are now almost ready to prove \cref{thm:thm2}. We start with a corollary which uses standard results on existence of codes with good distance.

\begin{corollary}
	\label{cor:beta-close-to-one}
	For every $\nu>0$ and every prime field $\F_q$, there exist $k$, and matrix $M \in \F_q^{k\times k}$, such that matrix $M$ satisfies $(1 - \nu, k^{1-\nu})$ exponential matrix polarization.
\end{corollary}
\begin{proof}
	Consider a parity check matrix $M_0$ of a BCH code with distance $2 k^{1 - \nu}$. We can achieve this with a matrix $M_0 \in \F_q^{k \times k_0}$, where $k_0 = \Oh(k^{1 - \nu} \log k)$ (see e.g.~\cite[Exercise 5.10]{essential-coding}).  
	Hence, as soon as $k > \Omega (2^{\nu^{-1} \log \nu^{-1})})$, we have $k_0 < \nu k$. Note that if $k_0=\nu_0 k$, then by \cref{lem:code-exp-suction} we can hope for $(1-\nu_0,k^{1-\nu_0}-\eps)$ exponential matrix polarization. We can now complete $M_0$ to a mixing matrix 
	to get overall $(1 - \nu, k^{1-\nu})$ exponential matrix polarization (since $\nu_0<\nu$). In order to complete matrix $M_0$ to a mixinig matrix, by \cref{lem:mixing-equiv} it is enough to complete it in arbitrary way to an invertible matrix, since already the first column of $M_0$ has support larger than $1$.
\end{proof}

\begin{remark}[Exponential polarization of random kernels]
	It is worth noting, that by the same argument and standard results on the distance of random linear codes, a random matrix $M \in \F_q^{k\times k}$ with high probability satisfies a $(1-\nu, k^{1-\nu})$ local polarization, with $\nu \to 0$ as $k \to \infty$.  Thus polar codes arising from a large random matrix will usually have this property.
\end{remark}

We now complete the proof of \cref{thm:thm2}.

\begin{proof}[Proof of \cref{thm:thm2}]
	Given $\beta < 1$ and $q$, let $\nu = (1-\beta)/3$. Now let $k$ and $M$ be as given by \cref{cor:beta-close-to-one}. 
	By \cref{lem:polarizing-matrix-implies-martingale}  we have that for every channel $\C_{Y|Z}$, $M$ satisfies $(1-\nu, k^{1-\nu})$-exponential local polarization. By \cref{thm:local-to-global-exp} we have that the same martingale satisfies $\Lambda$-exponentially strong polarization for $\Lambda = (1-\nu)^2 \log_2 k \geq (1-2\nu)\log_2 k$. By \cref{thm:main-quant} (in particular, \cref{rem:using-channel-coding-result}) we then get that the resulting codes have failure probability $O\parens{N \cdot \log q \cdot \exp\parens{-N^{1-2\nu}}} \le \exp\parens{-N^{1-3\nu}} =  \exp\parens{-N^{\beta}}$, where the first inequality holds for sufficiently large $N$ (as a function of $\nu$). 
\end{proof}

\subsection{Universality of Local Polarization}
\label{sec:univ-local-polarization}

Suppose we know that polar codes associated with a matrix $M \in \F_q^{k \times k}$ achieve capacity with error probability $\exp(-n^{\beta})$ in the limit of block lengths $n \to \infty$ (which may happen at lengths growing super polynomially in $\epsilon$ the gap to capacity). In this section, we prove a general result (previously stated as \cref{thm:thm3}) that `lifts' (in a black box manner) such a statement to the claim that, for every $\beta' < \beta$, polar codes associated with $M$ achieve polynomially fast convergence to capacity (i.e., the block length $n$ can be as small as $\mathrm{poly}(1/\epsilon)$ for rates within $\epsilon$ of capacity), and $\exp(-n^{\beta'})$ decoding error probability \emph{simultaneously}. Thus convergence to capacity at finite block length comes with almost no price in the (exponent of) decoding failure probability. 

Put differently, the result states that one can get polynomial convergence to capacity for free once one has a proof of convergence to capacity in the limit of $n \to \infty$ with root-exponential decoding error probability. Such proofs of convergence to capacity has been shown in \cite{KSU10} for the binary alphabet and \cite{mori-tanaka} for general alphabets. 
Yet another way of viewing the results of this section are that every proof of convergence to capacity has a proof of local polarization embedded in it. 

We get our result by proving a structural result that is roughly a converse to \cref{lem:code-exp-suction}.
Specifically in \cref{lem:compression-implies-distance} we show that if a matrix $M$ leads to a polar code with exponentially small failure probability then some high (but constant sized) tensor power $M^{\otimes t}$ of $M$ contains the parity check matrix of a high distance code. In fact more generally if a matrix in $\F_q^{k \times s}$ is the parity check matrix of a code which has a decoding algorithm that corrects errors from a $q$-symmetric channel with failure probability $\exp(-k^\beta)$ then this code has high distance.

Combining \cref{lem:compression-implies-distance} with \cref{lem:code-exp-suction} we get that every matrix that leads to a polar code with low error probability has a constant sized tensor that is a exponentially polarizing matrix. This immediately leads to a proof of \cref{thm:thm3}.

To derive our results we focus on a simple $q$-ary symmetric channel defined next.

\begin{definition}
	\label{defn:q-ary-bernoulli}
	For any finite field $\F_q$ and $\gamma \in [0,1]$, we will denote by $B_{q}(\gamma)$ 
	the distribution on $\F_q$ such that for $Z \sim B_{q}(\gamma)$ we have $\P(Z = 0) = 1 - \gamma$, and $\P(Z = k) = \frac{\gamma}{q-1}$ for any $k\not=0$.
\end{definition}

\begin{lemma}
	\label{lem:compression-implies-distance}
	Consider a matrix $H \in \F_q^{k \times s}$ and arbitrary decoding algorithm $\Dec : \F_q^s \to \F_q^k$, such that for independent random variables $\bvec{U}_1, \ldots \bvec{U}_i \sim B_{q}(\gamma)$ with $\gamma < \frac{1}{2}$, we have $\P(\Dec(\bvec{U}H) \not= \bvec{U}) < \exp(-k^{\beta})$. Then $\ker H$ is a code of distance at least $\frac{k^{\beta}} {\ln^{-1}(q/\gamma)}$.
\end{lemma}

\begin{proof}
	Consider maximum likelihood decoder $\Dec'(y) := \argmax_{x \in \F_q^k} \P(\bvec{U} = x | \bvec{U}H = y)$. By definition, we have $\P(\Dec'(\bvec{U}H) \not= \bvec{U}) < \P(\Dec(\bvec{U}H) \not= \bvec{U}) < \exp(-k^{\beta})$.
	
	Note that for $\bvec{U}$ distributed according to $B_{q}(\gamma)$, we have $\Dec'(\bvec{y}) = \argmin_{\bvec{x} : \bvec{x}H = \bvec{y}} \wt(\bvec{x})$, where $\wt(\bvec{x})$ is number of non-zero elements of $\bvec{x}$.
	
	Consider set $E = \{\bvec{x} \in \F_q^k \mid \text{ there exists } \bvec{h}\in \ker M, \wt(\bvec{x}+\bvec{h}) < \wt(\bvec{x})\}$, and observe that $\P(\Dec'(\bvec{U}H) \not = \bvec{U}) \geq \P(\bvec{U} \in E)$. We say that vector $\bvec{u} \in \F_q^k$ is \emph{dominated} by $\bvec{v} \in \F_q^k$ (denoted by $\bvec{u} \dom \bvec{v}$) if and only if $\forall i \in \supp(\bvec{u}),\, \bvec{u}_i = \bvec{v}_i$. We will argue that for any $\bvec{w_1} \in E$ and any $\bvec{w_2} \succeq \bvec{w_1}$, we have $\bvec{w_2} \in E$. Indeed, if $\bvec{w_1} \in E$, then there is some $\bvec{h} \in \ker H$ such that $\wt(\bvec{w_1} + \bvec{h}) < \wt(\bvec{w_1})$. We will show that $\wt(\bvec{w_2} + \bvec{h}) < \wt(\bvec{w_2})$, which implies that $\bvec{w_2} \in E$. Given that $\bvec{w_1} \preceq \bvec{w_2}$, we can equivalently say that there is a vector $\bvec{d}$ with $\bvec{w_1} + \bvec{d} = \bvec{w_2}$ and $\wt(\bvec{w_2}) = \wt(\bvec{w_1}) + \wt(\bvec{d})$. Hence
	\begin{equation*}
		\wt(\bvec{w_2} + \bvec{h}) = \wt(\bvec{w_1} + \bvec{d} + \bvec{h}) \leq \wt(\bvec{w_1} + \bvec{h}) + \wt(\bvec{d}) < \wt(\bvec{w_1}) + \wt(\bvec{d}) = \wt(\bvec{w_2})
	\end{equation*}
	
	Consider now $\bvec{w_0} \in \ker H$ to be minimum weight non-zero vector, and let us denote $A = \wt(\bvec{w_0})$. We wish to show a lower bound for $A$. By definition of the set $E$ we have $\bvec{w_0} \in E$, and by upward closure of $E$ with respect to domination we have 
	\[\P(\bvec{U} \in E) \geq \P(\bvec{w_0} \preceq \bvec{U}) = \parens{\frac{\gamma}{q - 1}}^A\ge \parens{\frac \gamma{q}}^A.\]
	
	On the other hand we have 
	\[\P(\bvec{U} \in E) \leq \P(\Dec'(\bvec{U}H) \not= \bvec{U}) \leq \P(\Dec(\bvec{U}H) \not= \bvec{U}) \leq \exp(-k^{\beta}).\]
	By comparing these two inequalities we get
	\begin{equation*}
		A \geq \frac{k^\beta}{\ln (q/\gamma)} \ . \qedhere
	\end{equation*}
\end{proof}

\begin{proof}[Proof of \cref{thm:thm3}]
	Consider the channel that outputs $\bvec{X} + \bvec{Z}$ on input $\bvec{X}$, where $\bvec{Z} \sim B_q(\gamma)$ for some $\gamma > 0$ (depending on $\beta, \beta'$). 
	The hypothesis on $M$ implies that for sufficiently large $n$ the polar code of block length $n$ corresponding to $M$ will have failure probability at most $\exp(-n^{\beta})$ on this channel.  Using the well-known equivalence between correcting errors for this additive channel and linear compression schemes (see e.g.~\cite[Prop. 11.2.1]{essential-coding}), we obtain that for all large enough $t$ there is some subset $S$ of $(h_q(\gamma)+\epsilon) k^t$ columns of $M^{\otimes t}$ that defines a linear compression scheme (for $k^t$ i.i.d copies of $B_q(\gamma)$), along with an accompanying decompression scheme with error probability (over the randomness of the source) at most $\exp(-k^{\beta t})$. 
	
	We now claim that  for all $\beta' < \beta$, there exists $t_0 = t_0(\beta',\beta)$ such that the Arikan martingale associated with some column permuted version of $M^{\otimes t_0}$, is $\beta' t_0 \log_2 k$-exponentially strongly polarizing.
	
	The proof of this claim  is in fact immediate, given the ingredients developed in previous sections. Apply the hypothesis about $M$ in the theorem with the choice $\epsilon  = (\beta-\beta')/4$ and $\gamma$ chosen small enough as a function $\beta,\beta'$ so that $h_q(\gamma)  \le (\beta-\beta')/4$ and let $t_0$ be a larger than promised value of $t$ in the statement, and large enough so that $3 \ln(q/\gamma) < m^{(\beta - \beta')/2}$ with $m := k^{t_0}$. Take moreover $\ell =  (h_q(\gamma)+\epsilon) m$ and $L = M^{\otimes t_0}$. 
	Using \cref{lem:compression-implies-distance} and the equivalence between linear coding for source and channel coding (mentioned above), we know there is submatrix $L' \in \F_q^{m \times \ell}$ of $L$ such that $\mathrm{ker}((L')^T)$ defines a code of distance $\Delta \ge m^{\beta}/\ln(q/\gamma)$. Define $M_0 = [ L' \mid \cdot] \in \F_q^{m \times m}$ to be any matrix obtained by permuting the columns of $L$ such that the columns in $L'$ occur first. By \cref{lem:code-exp-suction},  the matrix $M_0$ is $(1 -\ell/m,\Lambda)$-exponential matrix polarizing with $\Lambda = \Delta/2 - o(1) > \Delta/3$.

	For our choice of $\gamma,\epsilon$, we have $\ell/m \le \frac{\beta -\beta'}{2}$ and for our choice of $t_0$ (and therefore $m$) we have  $\Lambda \ge m^{(\beta+\beta')/2}$. Using \cref{lem:polarizing-matrix-implies-martingale} and \cref{thm:local-global}, it follows that the Arikan martingale associated with $M_0$ exhibits 
	$(\beta+\beta')/2 \times \left(1 - \frac{\beta-\beta'}{2} \right) \log_2 m$-exponentially strong polarization. 
	Since 
	\[\beta^"\stackrel{\text{def}}{=} \frac{\beta+\beta'}2 \cdot \left(1 - \frac{\beta-\beta'}{2} \right)= \beta'+ \frac{\beta-\beta'}2\cdot \parens{1- \frac{\beta+\beta'}2}> \beta',\] 
	the claim follows (in the above we used the fact that $0<\beta'<\beta<1$).

	Applying \cref{thm:main-quant} (and \cref{rem:using-channel-coding-result}) to the matrix $M_0 = M^{\otimes t_0}$ we conclude that there is a polynomial $p$ such that  given the gap to capacity $\varepsilon > 0$, and for every $s$ satisfying $N = k^{t_0 s} \geq p(\frac{1}{\varepsilon})$ there is an affine code generated by a subset of rows of $(M_0^{-1})^{\otimes s}$ which achieves $\varepsilon$-gap to capacity and has failure probability $\exp\parens{-N^{\beta"}}\cdot N\cdot \log{q}<\exp\parens{-N^{\beta'}}$ for large enough $N$. But this resulting code is simply an affine code generated by a subset of the rows of $(M^{-1})^{\otimes t}$, for $t = s t_0$, which concludes the proof. 
\end{proof}

\bibliographystyle{plain}
\bibliography{polar-refs}
\newpage
\appendix

\renewcommand{\t}{\widetilde}
\newcommand{\USC}{\bvec{\hat U}^{\textsc{SC}}}
\newcommand{\UF}{\bvec{\hat U}^{\textsc{F}}}
\newcommand{\UpF}{\bvec{\hat U}^{\prime\textsc{F}}}
\newcommand{\Zpj}{\bvec{Z}^{\prime (j)}}

\section{Codes from Polarization}
\label{sec:codes-from-polarization}

In this section, we describe the construction of polar codes,
and analyze the failure probability of decoders by corresponding
them to the \Arikan\ martingale.
This proves \cref{thm:main-quant,thm:poly-code}.

Specifically, 
we first describe the polar encoder
along with a fast $\mathcal{O}(n \log n)$-time implementation,
where $n$ is the blocklength.
Then, in \cref{sec:sc-decoder}
we define the (inefficient) successive-cancellation decoder,
and analyze its failure probability assuming a
correspondence between polar coding and the \Arikan\ martingale.
In \cref{sec:fast-decoder}, we describe a fast $\mathcal{O}(n \log n)$-time decoder
that is functionally equivalent to the successive-cancellation decoder.
Finally, in \cref{sec:correspondence},
we prove the required correspondence between polar coding and the \Arikan\
martingale.

Throughout this section, fix parameters
$k \in \N$ as the dimension of the mixing matrix $M \in \F_q^{k \x k}$, $\F_q$ as a finite field, and
$n = k^t$ as the codeword length.

\subsection{Polar Encoder}
\label{sec:polar-encoder}
Given a
set $S \subseteq [n]$ and a fixing
$\bvec{\alpha} \in \F_q^{|S^c|}$,\footnote{We use the notation $S^c=[n]\setminus S$.}
we define the polar code of dimension $|S|$ by giving the encoder mapping
$\F_q^S \to \F_q^n$
as follows:

\begin{algorithm}[H]
	\caption{Polar Encoder}
	\label{algo:slow-encoder}
	\begin{algorithmic}[1]
		\Const{$M \in \F_q^{k \x k}, S \subseteq [n], \bvec{\alpha} \in \F_q^{S^c}$}
		\Require{$\bvec{U} \in \F_q^S$}
		\Ensure{$\bvec{Z} \in \F_q^n$}
		\Procedure{Polar-Encoder}{$\bvec{U}; \bvec{\alpha}$}
		\State{Extend $\bvec{U}$ to $\bar{\bvec{U}} \in \F_q^n$
			by letting $(\bar{\bvec{U}}_i)_{i \not \in S} = \bvec{\alpha}$
			for coordinates not in $S$}
		\State \Return{$\bvec{Z} = \bar{\bvec{U}}\cdot (M^{-1})^{\otimes t}$}
		\EndProcedure
	\end{algorithmic}
\end{algorithm}

The above gives a polynomial time algorithm for encoding. An $\Oh_q(n\log n)$ algorithm can also be obtained by using the recursive structure imposed by the tensor powers.

Below, we switch to considering vectors in $\F_q^{k^t}$ as tensors in $(\F_q^k)^{\otimes t}$, indexed by multiindices $\bvec{i} \in [k]^t$.
The following encoder takes as input the `extended' message $\bvec{\bar U}$, as
defined above.

\begin{algorithm}[H]
	\caption{Fast Polar Encoder}
	\label{algo:fast-encoder}
	\begin{algorithmic}[1]
		\Const{$M \in \F_q^{k \x k}$}
		\Require{$\bvec{\bar U} \in (\F_q^k)^{\otimes t}$}
		\Ensure{$\bvec{Z}  = \bvec{\bar U} \cdot (M^{-1})^{\otimes t}$}
		\Procedure{Fast-Polar-Encoder$_t$}{$\bvec{\bar U}$}
		\If{$t=0$} 
		\State \Return $\bvec{\bar U}$ 
		\EndIf
		\ForAll{$j \in [k]$}
		\State{$\bvec{Z}^{(j)} \gets$ \Call{Fast-Polar-Encoder$_{t-1}$}{$\bar{\bvec{U}}_{[\cdot, j]}$}}
		\EndFor
		\ForAll{$\bvec{i} \in [k]^{t-1}$}
		\State
		$\bvec{Z}_{[\bvec{i}, \cdot]} \gets  (
		\bvec{Z}^{(1)}_{\bvec{i}},
		\bvec{Z}^{(2)}_{\bvec{i}},
		\dots,
		\bvec{Z}^{(k)}_{\bvec{i}}
		)\cdot M^{-1}$
		\EndFor
		\State \Return $\bvec Z$
		\EndProcedure
	\end{algorithmic}
\end{algorithm}

It is not too hard to verify that \cref{algo:fast-encoder} runs in $\mathcal{O}_{k,q}(n\log n)$ time. Indeed if $T(n)$ is the runtime of the algorithm on inputs of size $n=k^t$, then each call results in $k$ recursive calls to inputs of size $\frac nk$. Further, each recursive call solve $\frac nk$ systems of linear equations (each of which can be solved in $\Oh_q(k^3)$ time). Thus we get the recurrence (using the fact that $k$ is a constant) of $T(n)=k\cdot T(n/k) + \Oh_{k,q}(n)$, which results in the desired $\Oh_{k,q}(n\log{n})$ runtime. 

\subsection{The Successive-Cancellation Decoder}
\label{sec:sc-decoder}
Here we describe a successive-cancellation decoder.
Note that this decoder is not efficient,
but the fast decoder described later will nearly have the same error probability
as this decoder.

For given channel outputs $\bvec{Y}$,
let $\bvec{Z}$ be the posterior distribution
on channel inputs given outputs $\bvec{Y}$.
Each $\bvec{Z}_{\bvec{i}} \in \Delta(\F_q)$ is the conditional distribution
$\bvec{Z}_\bvec{i} | \bvec{Y}_\bvec{i}$ defined by the channel
$\mathcal{C}_{Y|Z}$ and the received output $\bvec{Y}_\bvec{i}$.

Now we define the decoder on the distribution vector $\bvec{Z}$ and the fixing 
$\bvec{\alpha} \in (\F_q \cup \{\bot\})^n$ as follows. We implicitly represent the subset $S^c$ of fixed positions by denoting $\bvec{\alpha}_i=\bot$ for those indices.

\begin{algorithm}[H]
	\caption{Successive-Cancellation Decoder}
	\label{algo:sc-decoder}
	\begin{algorithmic}[1]
		\Const{$M \in \F_q^{k \x k}, n = k^s,$}  
		\Require{$\bvec{Z} \in \Delta(\F_q)^n, \bvec{\alpha} \in (\F_q\cup \{\bot\})^{n}$ 
		}
		\Ensure{$\bvec{\hat U} \in \F_q^n$, {\color{brown} $\bvec{P} \in (\Delta(\F_q) \cup \bot)^n$}}
		\Procedure{SC-Decoder}{$\bvec{Z}; \bvec{\alpha}$}
		\State Compute the 
		distribution $\bvec{U} \in \Delta(\F_q^n)$ defined by $ \bvec{U} \leftarrow \bvec{Z} M^{\otimes s}$
		\ForAll{$i \in [n]$} \label{line:for-loop-scd}
		\If{$\bvec{\alpha}_i = \bot$}
		\State For $x \in \F_q$, $\hat{\bvec{U}}_i \gets \argmax_{x \in \F_q} \{\P_{\bvec{U}}(\bvec{U}_i = x)\}$ ;~~{\color{brown} $\bvec{P}_i(x) \gets \P_{\bvec{U}}(\bvec{U}_i = x)$} \label{line:scd-x}
		\Else
		\State $\hat{\bvec{U}}_i \gets \bvec{\alpha}_i$; {\color{brown} $\bvec{P}_i \gets \bot$}
		\EndIf
		\State Update distribution $\bvec{U} \leftarrow (\bvec{U} | \bvec{U}_i = \hat{\bvec{U}}_i)$ \label{line:update-U-scd}
		\EndFor
		\State \Return $\bvec{\hat U}$, {\color{brown} $\bvec{P}$}
		\EndProcedure
	\end{algorithmic}
\end{algorithm}

\begin{remark}
\label{rem:brown-part-sc-decoder}
    We note that parts in {\color{brown} brown} are not needed for the algorithm itself and only used in the analysis. Further, unless explicitly stated otherwise, we will use \textsc{SC-Decoder} to just denote the $\bvec{\hat U}$ part of the output (i.e. we will ignore {\color{brown} $\bvec{P}$} by default).
\end{remark}

Note that several of the above steps, including computing the joint distribution
of $\bvec{U}$ and marginal distributions of $\bvec{U}_i$,
are not computationally efficient
though we will get efficient algorithms effectively approximating these distributions later. Even then, we will only get an algorithm that gets an estimate of the probabilities $\P_{\bvec{U}}(\bvec{U}_i = x)$ to within an additive error of $1/4$ for every $x \in \F_q$. In what follows we will use the following definition:
	\begin{definition}
	\label{def:approx-sc-decoder}
	We will term an algorithm that runs an \textsc{SC-Decoder} where the algorithm gets an estimate of the probabilities $\P_{\bvec{U}}(\bvec{U}_i = x)$ to within an additive error of $1/4$ for every $x \in \F_q$  an {\em Approximate-Successive-Cancellation Decoder}.
	\end{definition}

\subsubsection{Decoding Analysis}

For this section, it will be useful to keep \cref{rem:brown-part-sc-decoder} in mind.

We will first reason about the `genie-aided' case, when the fixing $\bvec{\alpha} \in (\F_q\cup\{\bot\})^n$ of non-message bits is chosen uniformly at random, and revealed to both the encoder and decoder.
Then, we will argue that it is sufficient to use a deterministic fixing $\bvec{\alpha} = \bvec{\alpha}_0$.

We now argue that
over a uniform choice of message $\bvec{U}_S$,
and a uniform fixing $\bvec{\alpha}$ of non-message bits,
the probability of decoding failure is bounded as follows.

\begin{claim}
	\label{clm:small-failure-probability}
	For $S \subseteq [n]$ let $\bvec{V} \in (\F_q \cup\{\bot\})^n$ be given by $\bvec{V}_i \sim \F_q$ if $i \in S$ and $\bot$ otherwise.
	Let $\bvec{\alpha} \in (\F_q \cup\{\bot\})^n$ be given by $\bvec{\alpha}_i \sim \F_q$ if $i \not\in S$ and $\bot$ otherwise.
	Let $\bvec{Z} :=
	\textsc{Polar-Encoder}(\bvec{V};\bvec{\alpha})$ and $\bvec{Y}$ sampled according to the channel $\bvec{Y} := \cC_{Y|Z}(\bvec{Z})$. 
	Let $\bvec{U} \in \F_q^n$ be given by $\bvec{U}_i = \bvec{V}_i$ if $i \in S$ and $\bvec{\alpha_i}$ if $i \not\in S$. 
	With this notation, we have
	\begin{equation*}
		\Pr[
		\text{\sc SC-Decoder}(\bvec{Y}; \bvec{\alpha}) \neq \bvec{U} ]
		\leq  \sum_{i \in S} H( \bvec{U}_{i} ~|~ \bvec{U}_{< i}, \bvec{Y}).
	\end{equation*}
		Furthermore for every approximate-successive-cancellation decoder $D$ we have
		\begin{equation*}
			\Pr[
			D(\bvec{Y}; \bvec{\alpha}) \neq \bvec{U} ]
			\leq  3 \sum_{i \in S} H( \bvec{U}_{i} ~|~ \bvec{U}_{< i}, \bvec{Y}).
		\end{equation*}
\end{claim}

\begin{proof}
	Note that $\bvec{U}$ is uniform over $\F_q^n$.
	Now, we have:
	\begin{align*}
		\P\parens{\textsc{SC-Decoder}(\bvec{Y}; \bvec{\alpha}) \not= \bvec{U}} & =
		\P\parens{\exists i~ \hat{\bvec{U}}_i \not= \bvec{U}_i} \\
		& = \sum_{i \leq n} \P\parens{\hat{\bvec{U}}_i \not= \bvec{U}_i \mbox{ and } \hat{\bvec{U}}_{<i} = \bvec{U}_{<i}} \\
		& \leq \sum_{i \leq n} \P\parens{\hat{\bvec{U}}_i \not= \bvec{U}_i ~|~ \hat{\bvec{U}}_{<i} = \bvec{U}_{<i}}.
	\end{align*}
	Clearly for $i \not\in S$ we have $\P[\hat{\bvec{U}}_i \not= \bvec{U}_i] = 0$, since both are defined to be equal to $\bvec{\alpha}_i$ on those coordinates. It is enough to show that for $i\in S$ we have
	\begin{equation*}
		\P(\hat{\bvec{U}}_i \not= \bvec{U}_i ~|~ \bvec{U}_{<i} = \hat{\bvec{U}}_{<i}) \leq H(\bvec{U}_i ~|~ \bvec{U}_{<i}, \bvec{Y}).
	\end{equation*}
	This follows directly from \cref{lem:ML-conditional}, as $\hat{\bvec{U}}_i$
	is defined exactly as a maximum likelihood estimator of $\bvec{U}_i$ given channel
	outputs $\bvec{Y}$ and conditioning on $\bvec{U}_{<i}$ (note that the conditioning is happening in Line~\ref{line:update-U-scd}).
	
	The furthermore part of the claim follows from using the furthermore part of \cref{lem:ML-conditional} in the final step above.
\end{proof}
\begin{claim}
	\label{clm:exist-S}
	Let $n=k^t$, $\bvec{U} \sim \F_q^n, \bvec{Z} := \bvec{U} (M^{-1})^{\otimes t}, \bvec{Y} := \cC_{Y|Z}(\bvec{Z})$, where $\cC_{Y|Z}$ is a symmetric channel.
	If \Arikan~Martingale associated with $(M, \cC)$  satisfies $(\tau_\ell, \tau_h, \varepsilon)$-polarization, then there
	exists a subset $S \subset [n]$ of size $(\mathrm{Capacity}(\cC_{Y|Z}) - \varepsilon
	- \two\tau_h) n$,
	such that
	\begin{equation*}
		\sum_{i \in S} H( \bvec{U}_{i} ~|~ \bvec{U}_{< i}, \bvec{Y}) \leq \tau_\ell n \log q.
	\end{equation*}
\end{claim}
\begin{proof}

	Applying \cref{lem:mart-code}, we can deduce that for uniformly random index $i \in [n]$, normalized entropies $\bH(\bvec{U}_i | \bvec{U}_{< i}, \bvec{Y})$ are distributed identically as $X_t$ in the \Arikan~Martingale.
	
	Now, for symmetric channels, the uniform distribution achieves
	capacity (see e.g., \cite[Theorem 7.2.1]{CoverThomas}). And since matrix $(M^{(-1)})^{\otimes t}$ is invertible, vector $\bvec{Z}$ also has a uniform distribution.
	Thus, for uniform channel
	input $\bvec{Z}$,
	\begin{equation}
	\label{eq:cap-in-terms-of-Z}
	n\cdot \text{Capacity}(\cC_{Y|Z})
	= \bH(\bvec Z) - \bH(\bvec Z | \bvec Y)
	= n-\bH(\bvec Z |  \bvec Y).	    
	\end{equation}

	Let $S$ be the set of all indices $i$ such that $\bH(\bvec{U}_i ~|~ \bvec{U}_{<i}, \bvec{Y}) < \tau_\ell$.
	By definition, we have
	$$
	\sum_{i \in S}
	\bar H( \bvec{U}_{i} ~|~ \bvec{U}_{< i}, \bvec{Y}) \leq \tau_\ell n
	$$
	as desired.
	
	Now observe that polarization of martingale $X_t$ and \cref{lem:mart-code} directly implies that we have at most $\varepsilon n$ indicies $i$ satisfying $\bH(\bvec{U}_i ~|~ \bvec{U}_{<i}) \in (\tau_\ell, 1-\tau_h)$ (recall that in \cref{lem:mart-code} we pick one such index uniformly at random). Let $S'$ be a set of indices for which $\bH(\bvec{U}_i ~|~ \bvec{U}_{<i}, \bvec{Y}) > 1 - \tau_h$. We have
	\begin{align*}
		n (1 - \text{Capacity}(\cC_{Y|Z})) & = \bH(\bvec{U} (M^{-1})^{\otimes t} ~|~ \bvec{Y}) \tag{\cref{eq:cap-in-terms-of-Z}}\\
		& = \bH(\bvec{U}_{1}, \ldots, \bvec{U}_{n} | \bvec{Y}) \tag{Since
			$(M^{-1})^{\otimes t}$ is full rank}\\
		& = \sum_{i \in [n]} \bH(\bvec{U}_{i} | \bvec{U}_{<i}, \bvec{Y}) \tag{Chain
			rule}\\
		& \geq \sum_{i \in S'} \bH(\bvec{U}_{i} | \bvec{U}_{<i}, \bvec{Y}) \\
		& \geq (1 - \tau_h) |S'| \ge |S'| - \tau_h n , 
	\end{align*}
	which implies that
	\begin{equation*}
		|S'| \leq n(1 - \text{Capacity}(\cC_{Y|Z}) + \two \tau_h),
	\end{equation*}
	and finally 
	\begin{equation*}
		|S| \geq n - |S'| - \varepsilon n \geq
		n (\text{Capacity}(\cC_{Y|Z}) - \varepsilon - \two \tau_h) \ . \qedhere
	\end{equation*}
\end{proof}

We can now combine the above to prove
a version of
\cref{thm:poly-code} for the (inefficient) successive-cancellation decoder:

\begin{theorem}
	\label{thm:sc-decoder-quant}
	Let $\C$ be a $q$-ary symmetric memoryless channel and let $M \in \F_q^{k \times k}$ be an
	invertible matrix.
	If the \Arikan\ martingale associated with $(M,\C)$ satisfies $(\tau_\ell, \tau_h,
	\varepsilon)$-polarization, then for every $t$, there is an affine code $C$,
	that is generated by the rows of $(M^{-1})^{\otimes t}$ and an affine shift,
	such that the rate of $C$ is at least $\mathrm{Capacity}(\C) - \varepsilon(t) -
	\two \tau_h(t)$, and $C$ can be encoded in time $\Oh(n \log n)$ where $n = k^t$.
	Furthermore, the successive-cancellation decoder
	succeeds with probability at least $1-n \log(q) \tau_\ell$
	and every approximate-successive-cancellation decoder
		succeeds with probability at least $1-3n \log(q) \tau_\ell$.
\end{theorem}

\begin{proof}
	Let $\bvec{\bar U} \sim \F_q^n$, $\bvec{\bar Z} := \bvec{U} (M^{-1})^{\otimes t}$,
	and $\bvec{\bar Y} := \cC_{Y|Z}(\bvec{\bar Z})$.
	
	By \cref{clm:exist-S},
	there exist a set $S \subset [n]$ of size $(\text{Capacity}(\cC_{Y|Z}) -
	\varepsilon - \two\tau_h) n$, such that
	\begin{align*}
		\sum_{i \in S} H( \bvec{\bar U}_{i} ~|~ \bvec{\bar U}_{< i}, \bvec{\bar Y}) \leq \tau_\ell n \log q
	\end{align*}

	On the other hand, by \cref{clm:small-failure-probability}, the failure probability of
	the successive-cancellation decoder is bounded by
	\begin{equation}\label{eq:scd-failure}
		\Pr_{U, \alpha, Y}[
		\text{\sc SC-Decoder}(\bvec{Y}; \bvec{\alpha})_S \neq U ]
		\leq  \sum_{i \in S} H( \bvec{U}_{i} ~|~ \bvec{U}_{< i}, \bvec{Y}),
	\end{equation}
	where random variables $\bvec{U}, \bvec{Y}, \bvec{\alpha}$
	are defined as in \cref{clm:small-failure-probability}. Note that, in fact the joint distribution of $(\bvec{U}, \bvec{Y}, \bvec{Z})$ and $(\bvec{\bar{U}}, \bvec{\bar{Y}}, \bvec{\bar{Z}})$, are the same, despite superficially more complicated way in which sampling from distribution $(\bvec{U}, \bvec{Y}, \bvec{Z})$ was defined. Therefore
	\begin{align*}
		\sum_{i \in S} H( \bvec{U}_{i} ~|~ \bvec{U}_{< i}, \bvec{Y}),
		&=  \sum_{i \in S} H( \bvec{\bar U}_{i} ~|~ \bvec{\bar U}_{< i}, \bvec{\bar Y})\\
		&\leq \tau_\ell n \log q.
	\end{align*}
	Note that this failure probability is an average over random choice
	of fixing $\bvec{\alpha}$,
	but this implies there is some deterministic fixing $\bvec{\alpha} = \bvec{\alpha}_0$ with
	failure probability at least as good. Further, by linearity of the encoding
	(\cref{algo:fast-encoder})
	such a deterministic fixing yields an affine code.
	The rate of this code is $|S|/n \geq (\text{Capacity}(\cC_{Y|Z}) - \varepsilon -
	\two\tau_h)$ as desired.
	
	If we replace the successive-cancellation decoder by an approximate successive cancellation decoder, then the theorem follows by using the furthermore part of \cref{clm:small-failure-probability} in \cref{eq:scd-failure} above.
\end{proof}

\subsubsection{Fast Decoder}
\label{sec:fast-decoder}
In this section we will define the recursive \textsc{Fast-Decoder} algorithm.
The observation that polar codes admit a recursive fast-decoder was made in
the original work of \Arikan~\cite{arikan-polar}. Our presentation is somewhat different in that it decodes general product distributions (and does not require the marginals to be identical). 

\textsc{Fast-Decoder}  will take on input descriptions of the posterior
distributions on channel inputs $\{\bvec{Z}_\bvec{i}\}_{\bvec{i} \in [k]^s}$ for some $s$, where each individual $\bvec{Z}_{\bvec{i}} \in \Delta(\F_q)$ is a distribution over $\F_q$, as well as $\bvec{\alpha} \in (\F_q \cup \{\bot\})^{[k]^s}$ where $\bvec{\alpha}_{\bvec{i}} \in \F_q$ 
are the fixed values corresponding to non-message positions. The output of \textsc{Fast-Decoder} is a vector $\bvec{\hat{Z}} \in (\F_q^k)^{\otimes s}$---the guess for the actual channel inputs. To recover the message, it is enough to apply $\bvec{\hat{U}} := \bvec{\hat{Z}} M^{\otimes s}$, and restrict it to the positions where $\bvec{\alpha}_i  =\bot$.

In \cref{algo:fast-decoder}, for $\bvec{W}_{\bvec{i}} \in \Delta(\F_q^k)$---a description of joint probability distribution over $\F_q^k$, we will write $\pi_j(\bvec{W}_{\bvec{i}}) \in \Delta(\F_q)$ as a $j$-th marginal of $\bvec{W}_{\bvec{i}}$ for $j\in [k]$, i.e. projection on the $j$-th coordinate. In addition, we will use $\pi_{\le j}(\bvec{W}_{\bvec{i}}) \in \Delta(\F_q)^j$ to denote the projection of $\bvec{W}$ to the first $j$ marginal coordinates.

\begin{algorithm}[H]
	\caption{Fast Decoder}
	\label{algo:fast-decoder}
	\begin{algorithmic}[1]
		\Const{$M \in \F_q^{k \x k}$}
		\Require{$\bvec Z = \{\bvec{Z}_{\bvec{i}} \in \Delta(\F_q)\}_{\bvec i \in [k]^s}, ~\bvec{\alpha} \in (\F_q \cup \{\perp\})^{[k]^s}$}
		\Ensure{$\bvec{\hat Z} \in (\F_q^k)^{\otimes s}$, {\color{brown} $\bvec{Q} \in(\Delta(\F_q^k) \cup \{\bot\})^{\otimes s}$, $\UF \in(\F_q^k)^{\otimes s}$}}
		\Procedure{Fast-Decoder$_s$}{$\bvec{Z}$; $\bvec{\alpha}$}
		\If{$s = 0$}
		\If{$\bvec{\alpha} = \bot$}
		\State \Return $\hat Z = \argmax_{x \in \F_q} \P\parens{\bvec{Z} = x}$, {\color{brown} $\bvec{Q} = \bvec{Z}$, $\UF = \hat{Z}$}\label{line:base-case}
		\Else
		\State \Return $\hat Z = \bvec{\alpha}$, {\color{brown} $\bvec{Q} = \bot$, $\UF = \bvec{\alpha}$}
		\EndIf
		\Else
		\ForAll{$\bvec{i} \in [k]^{s-1}$}
		\State Compute joint distribution $\bvec{W}_{\bvec{i}} \in \Delta(\F_q^k)$,
		given by $\bvec{W}_{\bvec{i}} \leftarrow \bvec{Z}_{[\cdot, \bvec{i}]} M$ \label{line:update-W_i}
		\EndFor
		
		\ForAll{$j \in [k]$}
		\State $\Zpj \leftarrow \{\pi_j(\bvec{W}_\bvec{i})\}_{\bvec{i} \in [k]^{s-1}}$ \label{line:z}
		\State $\hat{\bvec{V}}_{[j,\cdot]}$, {\color{brown} $\bvec{Q}_{[j,\cdot]}$, $\UF_{[j,\cdot]}$} $\leftarrow \textsc{Fast-Decoder}(\Zpj; \bvec{\alpha}_{[j, \cdot]}, s-1)$ \label{line:v}
		
		\ForAll{$\bvec{i} \in [k]^{s-1}$}
		\State Update distribution
		$\bvec{W}_{\bvec{i}} \leftarrow (\bvec{W}_{\bvec{i}} | \pi_{\leq j}(\bvec{W}_{\bvec{i}})= \bvec{\hat{V}}_{[\leq j,\bvec{i}]})$ 
		\label{line:updateW}
		\EndFor
		
		\EndFor
		
		\ForAll{$\bvec{i} \in [k]^{s-1}$}
		\State $\hat{\bvec{Z}}_{[\cdot, \bvec{i}]} \leftarrow \bvec{V}_{[\cdot,\bvec{i}]}\cdot M^{-1}$ \label{line:setZ}
		\EndFor
		
		\State \Return $\bvec{\hat Z}$, {\color{brown} $\bvec{Q}$, $\UF$}\label{line:ret-fast-dec}
		
		\EndIf
		\EndProcedure
	\end{algorithmic}
\end{algorithm}

We make an remark analogous to \cref{rem:brown-part-sc-decoder} for \textsc{Fast-Decoder}:
\begin{remark}
\label{rem:brown-part-fast-decoder}
    In the code above, the parts in {\color{brown} brown} are not needed for the running of the algorithm but included since they help with the analysis. Further, unless explicitly stated otherwise, we will use \textsc{SC-Decoder} to just denote the $\bvec{\hat Z}$ part of the output (i.e. we will ignore {\color{brown} $\bvec{Q}$, $\UF$} by default).
\end{remark}

Analogous to \cref{def:approx-sc-decoder}, we define a similar approximate version of \textsc{Fast-Decoder}:
	\begin{definition}
	\label{def:approx-fast-decoder}
	We will term an algorithm that runs an \textsc{Fast-Decoder} where the algorithm gets an estimate of the probabilities $\P\parens{\bvec{Z} = x}$ to within an additive error of $1/4$ for every $x \in \F_q$  a {\em precision-bounded \textsc{Fast-Decoder}}.
	\end{definition}

The \textsc{Fast-Decoder} as described above runs in time $\Oh(n \log n)$, where $n = k^s$ is block length if one assumes infinite precision arithmetic. Furthermore even a bounded-precision model only requires $\Oh(n \log n)$ operations in the ``floating point RAM'' model --- the model where a non-negative real number $r \in [0,1]$ is represented with two $\ell = \Oh(\log n)$ bit integers $a, b$ as $a\cdot 2^b$ and two such numbers can be added, multiplied or divided in a single step. 

In bit more detail, the above representation is also known as the {\em Floating point number system}~\cite[Chapter 2]{higham-book}. Before we go into the details of the runtime analysis of \textsc{Fast-Decoder}, we quickly summarize the relevant properties of the floating point number system.

\paragraph{Floating point number system and floating point RAM model.}

We recall the definition of the floating point number system:
\begin{definition}[\cite{higham-book}, Section 2.1]
    \label{def:fl}
    A floating point number system $F\subset \R$ is a subset of real numbers whose elements have the form
    \[ y = \pm a\cdot \beta^{e-\Delta},\]
    where
    \begin{itemize}
        \item The integer $\beta\ge 2$ is the {\em base} or {\em radix}
        \item The natural number $\Delta$ is the {\em precision}
        \item The integer $e$ is the {\em exponent} and has the range $e_{\min}\le e\le e_{\max}$ for integers $e_{\min}\le e_{\max}$
        \item The natural number $a$ is the {\em significand} and it is assumed that
        \[\beta^{\Delta-1}\le a \le \beta^{\Delta}-1.\]
    \end{itemize}
    
    The {\em representation range} of $F$ is given by $\brackets{\beta^{e_{\min}-1},\beta^{e_{\max}}\parens{1-\beta^{-\Delta}} }$.
\end{definition}

Before we proceed, we note the simplications to the above definition that we use in our model:
\begin{definition}
    \label{def:fl-simple}
    We use the floating point number system from \cref{def:fl} with the following simplifications/modifications:
    \begin{itemize}
        \item Set $\beta=2$.
        \item $\Delta=\ell$.\footnote{Since we are using $\ell$ {\em bits} to represent $a$.}
        \item $e_{\min}=-2^{\ell}$ and $e_{\max}=2^\ell$.
    \end{itemize}
\end{definition}

For the rest of this discussion we will assume the parameters that we have set in \cref{def:fl-simple}. Next we recall some properties of the floating point number system that we will use as given in our runtime analysis of \textsc{Fast-Decoder}.

Before we present the results, we fix some more notation. For $x\in \R$ falling within the representation range of the floating point system, we will use $\rnd{x}$ to denote the closest approximation of $x$ in the floating point system. For any vector $\bvec{y}\in\R^k$, we will overload notation and use $\rnd{\bvec{y}}$ to denote the vector obtained by applying $\rnd{\cdot}$ to each component of $\bvec{y}$. This leads to the following definition, which defines a crucial quantity that will turn up in our approximation bounds.

\begin{definition}
    \label{def:roundoff}
    The {\em unit roundoff} is defined as
    \[u = 2^{-\Delta}.\]
\end{definition}

We first recall a bound on the approximation error that the rounding entails:
\begin{lemma}[\cite{higham-book}, Theorem 2.2]
    \label{lem:rounding-error}
    Let $x\in\mathbb{R}$ be in the representation range of the floating point system. Then
    \[\rnd{x}=(1+\delta)\cdot x\text{    where    } \abs{\delta}<u.\]
\end{lemma}

We will also use $\rnd{\cdot}$ applied to a formula, to denote a result of a floating-point evaluation of this formula. We will use the so called {\em standard model}~\cite[Section 2.2]{higham-book}:
\begin{definition}[Standard Model]
\label{def:standard-model}
The standard model  assumes the following precision bounds on binary operations.
Given $x,y\in F$ and $\mathrm{op}\in\set{+,-,\times, \div}$, we 
have

\[\rnd{x~\mathrm{op}~y} = (x~\mathrm{op}~y)\cdot (1+\delta) \text{    where    } \abs{\delta}\le u,\]
as long as $x~\mathrm{op}~y$ is in the representation range.
\end{definition}
In particular, even if $x~\mathrm{op}~y$ happens to have the exact representation in the floating point number system $F$, we do not require the result of this floating point operation to be exact.

For the rest of the section, we will assume the standard model in our floating point RAM model.

Next, we present a technical lemma that will be useful for us:
\begin{lemma}[Simple generalization of Lemma 3.1 in~\cite{higham-book}]
\label{lem:higham-bound}
    Let $\delta_1,\dots,\delta_n$ be such that $\sum_{i=1}^n  \abs{\delta_i}<1$ and let $\rho_i\in \set{-1,1}$ for all $i\in [n]$. Then we have
    \[\prod_{i=1}^n \parens{1+\delta_i}^{\rho_i} = 1 +\theta,\]
    where
    \[\abs{\theta} \le \frac{\sum_{i=1}^n  \abs{\delta_i}}{1-\sum_{i=1}^n  \abs{\delta_i}}.\]
\end{lemma}

Finally, we present approximation error bounds for computing a bounded-degree rational function, which will be crucial in our runtime analysis of \textsc{Fast-Decoder}:
\begin{lemma}
    \label{lem:rational-fn-error}
Let $f(X_1,\dots,X_N)$ and $g(X_1,\dots,X_N)$ be {\em multi-linear} polynomials\footnote{The result can be proven for general polynomials as well. However, since we only need the result for multilinear polynomials and the notation for multi-linear polynomials is slightly cleaner, we stick with the multi-linear case.} such that both satisfy the following properties:
\begin{itemize}
    \item the degree is at most $d$
    \item there are at most $m$ monomials
    \item all the coefficients are non-negative and have exact representation in the floating point number system.
\end{itemize}
Further, let $\bvec{x},\tvec{x}\in\R_{\ge 0}^N$, be such that there exists an $\eps>0$ such that for every $i\in [N]$, we have
\[ \abs{\bvec{x}_i-\tvec{x}_i}\le \eps \bvec{x}_i,\]
and moreover let $e_0$ be such that all $\tvec{x}_i$  and all coefficients of $f, g$ lie in $\brackets{2^{-e_0}, 2^{e_0}}$.
Then, assuming 
\begin{eqnarray}
    \label{eq:accuracy-pre-bound}
    4\parens{d\cdot \eps + (d+\log{m})\cdot u}+1& \le&  \frac 12, \\
\label{eq:range-pre-bound}
    e_1 := 2 (d+1)(e_0 + 1) + 4 \log m + 1  & \leq & 2^\ell,
\end{eqnarray}
we have that
\begin{equation}
    \label{eq:final-rational-fn-bound}
    \abs{\frac{f(\bvec{x})}{g(\bvec{x})} - \rnd{\frac{f(\tvec{x})}{g(\tvec{x})} } } 
\le 8 \cdot \parens{d\cdot \eps + (d+\log{m}+1)\cdot u} \cdot \frac{f(\bvec{x})}{g(\bvec{x})},
\end{equation}
 and moreover
\begin{equation}
\label{eq:range-post-bound}
    \left|\log \rnd{\frac{f(\tvec{x})}{g(\tvec{x})}}\right| \leq e_1,
\end{equation}

where the $\rnd{\frac{f(\tvec{x})}{g(\tvec{x})} }$ is computed by using pair-wise operations (and paying for approximation error for each such operation as in the standard model).

\end{lemma}
\begin{proof}
    We will compute $\rnd{\frac{f(\tvec{x})}{g(\tvec{x})} }$ by first computing each monomial in $f(\tvec{x})$ and $g(\tvec{x})$ and then summing up the at most $m$ values in a depth $\log{m}$ tree fashion. Finally we divide $f(\tvec{x})$ by $g(\tvec{x})$ to obtain our answer.
    
    For notational convenience for each $i\in [N]$, define $\eps_i$ such that $\tvec{x}_i = (1+\eps_i)\cdot \bvec{x}_i$. Note that we have $\abs{\eps_i}\le \eps$.
    
    To see the error bound consider an arbitrary monomial, which we assume WLOG to be $\prod_{i=1}^d X_i$. We compute $\prod_{i=1}^d \tvec{x}_i$ in the obvious way. It is easy to check that $\rnd{\prod_{i=1}^d \tvec{x}_i}=\parens{\prod_{i=1}^d \tvec{x}_i}\cdot \prod_{i=1}^{d-1}(1+\delta_i)$, where $\abs{\delta_i}\le u$. Further, by definition of $\eps_i$, we have
    \[\rnd{\prod_{i=1}^d \tvec{x}_i}=\parens{\prod_{i=1}^d \bvec{x}_i}\cdot \prod_{i=1}^{d}(1+\delta_i)(1+\eps_i),\]
    where for notational simplicity define $\delta_d=1$.
    
    To apply the error bounds for the floating point operations we need to argue that all the results of the multiplications in the computation above are within the representations range. Indeed, since $\tvec{x}_i \geq 2^{-e_0}$, all the intermediate results in the multiplication above are at least $\left(\frac{2^{-e_0}}{1 + u}\right)^{d+1} \geq 2^{-(e_0 + 1)(d+1)}$. Similarly, for the upper bound: since $\tvec{x_i} \leq 2^{e_0}$, all the intermediate results are at most $\left((1+u)2^{e_0}\right)^{(d+1)} \leq 2^{(e_0 + 1)(d+1)}$, which is assumed to be within the representation range \eqref{eq:range-pre-bound}. 
    
    Now let us consider the computation of $\rnd{f(\tvec{x})}$. Let $\mathcal M$ be the collection of all subset of size at most $d$ that correspond to the monomials in $f(X_1,\dots,X_N)$. Then when computing $\rnd{f(\tvec{x})}$, for each $S\in \mathcal{M}$, we first compute $\rnd{\prod_{i\in S} \tvec{x}_i}$, which satisfies by the above discussion,
        \[\hat{m}_S\stackrel{\text{def}}{=}\rnd{\prod_{i\in S} \tvec{x}_i}=\parens{\prod_{i\in S} \bvec{x}_i}\cdot \prod_{i\in S}(1+\delta_i)(1+\eps_i).\]
    Now recall, we need to compute $\sum_{S\in \mathcal{M}} \hat{m}_S$. This in turn adds more error. In particular, if we use the algorithm that computes the sum in a recursive-pairwise manner, we get that
    \[\rnd{f(\tvec{x})} = \sum_{S\in \mathcal{M}} \widetilde{m}_S,\]
    where
    \[\widetilde{m}_S = \hat{m}_S\prod_{j=1}^{\log{\abs{\mathcal{M}}}} \parens{1+\delta^{(S)}_j},\]
    where each $\abs{\delta^{(S)}_j}\le u$. In other words, we have
    \[\widetilde{m}_S = \parens{\prod_{i\in S} \bvec{x}_i}\cdot \parens{\prod_{i\in S}(1+\delta_i)(1+\eps_i)}\cdot \parens{\prod_{j=1}^{\log{\abs{\mathcal{M}}}} \parens{1+\delta^{(S)}_j}}.\]
    The above along with \cref{lem:higham-bound}, shows that for every $S\in\mathcal{M}$,
    \begin{align*}
    \abs{\widetilde{m}_S - \parens{\prod_{i\in S} \bvec{x}_i}} & \le \frac{\abs{S}\parens{\eps+u}+\log{m}\cdot u}{1- \abs{S}\parens{\eps+u}+\log{m}\cdot u} \cdot \parens{\prod_{i\in S} \bvec{x}_i}\\
    & \le 2\cdot\parens{d(\eps+u)+\log{m}\cdot u}\cdot \parens{\prod_{i\in S} \bvec{x}_i},        
    \end{align*}
    here the first inequality follows from the facts that $\abs{\mathcal{M}}\le m$, $\abs{\delta_i}\le \eps$, $\abs{\delta^{(S)}_j}\le u$ and $\abs{\eps_i}\le \eps$ and the second inequality follows from the fact that $\abs{S}\le d$ and $d(\eps+u)+\log{m}\cdot u\le \frac{1}{2}$ (which in turn follows from \cref{eq:accuracy-pre-bound}).
    
    Now, using the fact that all cofficients in $f(x)$ are non-negative and have an exact representation in the floating point number system, the above then implies that
    \[\abs{\rnd{f(\tvec{x})} - f(\bvec{x})} \le 2\cdot\parens{d(\eps+u)+\log{m}\cdot u}\cdot f(\bvec{x}).\]
    By a similar argument we get
        \[\abs{\rnd{g(\tvec{x})} - g(\bvec{x})} \le 2\cdot\parens{d(\eps+u)+\log{m}\cdot u}\cdot g(\bvec{x}).\]
     As earlier, to apply the error bounds on the result of each floating point addition in the calculation, we need to ensure that all results of all the intermediate computations are within the representation range. Since we are adding exactly represented non-negative values, the lower bound of the representation range is trivially smaller  than any of those intermediate values. The largest intermediate value can appear at the end of the calculation, and is upper bounded by $(1+u)^{\log m} m ((1+u) 2^{e_0})^{d+1}) \leq 2^{(d+1)(e_0 + 1) + 2 \log m}$, which is assumed to be in the representation range \eqref{eq:range-pre-bound}.  
    
    Then noting that to compute the final answer we divide $\rnd{f(\tvec{x})}$ by $\rnd{g(\tvec{x})}$, which along with \cref{def:standard-model}, \cref{lem:higham-bound} and \cref{eq:accuracy-pre-bound}, proves the claimed bound in \cref{eq:final-rational-fn-bound}, as desired.
    
    Moreover, since $|\log \rnd{g(\tvec{x})}| \leq (d+1)(e_0 + 1) + 2 \log m$, and similarly for $|\log \rnd{f(\tvec{x})}|$, the quotient satisfy $|\log \rnd{ \frac{f(\tvec{x})}{g(\tvec{x})}}| \leq 2 (d+1)(e_0 + 1) + 4 \log m + 1 = e_1$, proving \cref{eq:range-post-bound}.
\end{proof}

Finally, we state the definition of a floating point RAM:
\begin{definition}[Floating point RAM model]
\label{def:floating-RAM}
A floating point RAM works with numbers in the floating point system as in \cref{def:fl-simple} with $\ell=O(\log{n})$ for inputs of size $n$. Each arithmetic operation in the floating point number  system is assumed to take unit time.
\end{definition}
We note that in the above, each floating point number can be represented with constant many registers of  $O(\log{n})$ bits and that each of the basic floating operations translates to constant many operations over constant many registers of $O(\log{n})$ bits. In other words, each such floating point operation can be done in $O(1)$ time in the standard RAM model and this justifies the assumption on floating point operations taking unit time in the above definition.

\paragraph{Runtime analysis of \textsc{Fast-Decoder}.}

We are now ready to do a runtime analysis of \textsc{Fast-Decoder}:

	\begin{lemma}\label{lem:nlogn}
		For $n = k^s$, \textsc{Fast-Decoder} runs in $\Oh_{q,k}(n\log n)$ time assuming unit cost infinite precision arithmetic. Furthermore it can be implemented in a bounded-precision floating point RAM
		model (of \cref{def:floating-RAM}) to compute every intermediate real number to within an additive error of $1/4$ in $\Oh_{q,k}(n\log n)$ time, as long as the description of the channel $\mathcal{C}_{Y|Z}$ is given in a floating point number system using $\Oh(\log n)$ bits per conditional probability. In other words, bounded-precision \textsc{Fast-Decoder} can also be implemented in $\Oh_{q,k}(n\log n)$ time in the floating point RAM.
	\end{lemma}
	
	\begin{proof}
	    We first remark that we use a ``truth-table'' representation for each probability distribution, i.e. we store tables with $q$ and $q^k$ floating point numbers respectively to represent a distribution in $\Delta(\F_q)$ and $\Delta(\F_q^k)$, respectively. In other words, each $\bvec{Z_j}$ for each $\bvec{j}\in [k]^s$ is a vector length $q$ and $\bvec{W_i}$  for each  $\bvec{i}\in [k]^{s-1}$ is a vector of length $q^k$.
	    
	    Let us separate out the computing on real numbers and the rest. It is easy to see that for a recursive call with $n=k^s$, all the operations that do not involve floating point operations can be done in $\Oh_{q,k}(n)$ time. We also note that Lines~\ref{line:update-W_i}~and~\ref{line:updateW} are the only places where we have to perform floating point operations. Further, it can be checked that there are $\Oh_{q,k}(n)$ such operation.
		Thus, the running time (in both infinite precision setting and floating point RAM model), $T(n)$ of \textsc{Fast-Decoder} satisfies the recurrence $T(n) \leq kT(n/k) + \Oh_{q,k}(n)$ which yields $T(n) = \Oh_{q,k}(n \log n)$.
		
		Finally, we prove the claim on the claimed precision in the floating point RAM model.
		We  note that while our final desired precision is only an additive $1/4$, intermediate precision needs to be high since the precision goes down at each recursive call. More precisely, our goal is to use \cref{lem:rational-fn-error} to bound this error. Before we can apply \cref{lem:rational-fn-error}, we verify that the pre-conditions of the lemma holds. 
		
		As mentioned earlier, Lines~\ref{line:update-W_i}~and~\ref{line:updateW} are the only places where we have to perform floating point operations are the only places to perform floating point operations. In particular, the input are the $N=q\cdot k^s$ probability values in $\bvec{Z}$ (denote these $N$ probability values by $\bvec{p}=\parens{p_1,\dots,p_N}$). Line~\ref{line:update-W_i} computes for each of the $q^k$ values in $\bvec{W_i}$  a degree $k$ multi-linear polynomial in $k$ out of the $N$ variables (in fact this polynomial is actually a monomial). Line~\ref{line:updateW} is where we update the $q^k$ values of $\bvec{W_i}$. In particular, each computed value is a rational function $\frac{f(\bvec{p})}{g(\bvec{p})}$, where $f(X_1,\dots,X_N)$ is still a monomial in $k$ variables and $g(X_1,\dots,X_N)$ is a multilinear polynomial of degree $k$ with at most $q^k$ monomials each with a coefficient of $1$. Note that $f$ and $g$ satisfy the pre-conditions of \cref{lem:rational-fn-error}.
		
		Now, consider a recursive call to \textsc{Fast-Decoder} with $s\gets s-i$. We first note that we do not have access to $\bvec{p}$ but rather an approximation $\tvec{p}$ where each entry has an error bounded by $1\pm \eps_i$, where we define $\eps_i$ soon. Moreover we will maintain the bound $e_i$ on the magnitude of the exponents of the approximations at the $i$-th level of the recursion, namely we shall ensure that on the $i$-th level of recursion for each $j$ we have $|\log \tvec{p}_j| \leq e_i$ --- the $e_i$ will be defined soon as well.
		
		First we note that by \cref{lem:rounding-error}, we have that $\eps_0 \le u$, and $e_0 \leq 2^{\Oh(\log n)}$ since we assumed that the description of the channel is specified using $\Oh(\log n)$ bits.
		Now applying \cref{lem:rational-fn-error} with $d\gets k, m\gets q^k, \eps\gets \eps_i, \bvec{x}\gets \bvec{p}$ and $\tvec{x}\gets \tvec{p}$ from \cref{eq:final-rational-fn-bound} (it can be verified that \cref{eq:accuracy-pre-bound} will be satisfied with our parameter choice), we get
		\[\eps_{i+1}\le 8\parens{k\cdot \eps_i + (k+k\log{q}+1)\cdot u}\le 32\cdot k\log{q} \cdot \eps_i,\]
		where the inequality uses $k\ge 1$ and  the fact that $\eps_i$ is increasing in $i$ and hence $u\le \eps_0\le \eps_i$.
		Thus, we have that
		\[\eps_s \le \parens{32\cdot k\log{q}}^s\cdot u.\]
		
		Similarly, from \cref{eq:range-post-bound}, we get
		\[
		e_{i+1} \leq 2(k+1)(e_i + 1) + 4 k \log q + 1 \leq (13 k \log q) \cdot e_i,
		\]
		and therefore $e_s \leq (13 k \log q)^s \cdot e_0$. Since $e_i \leq e_s$ for each $i \leq s$, to ensure condition \cref{eq:range-pre-bound} in all applications of \cref{lem:rational-fn-error}, it is enough to pick $\ell$ such that $(13 k \log q)^s e_0 \leq 2^{\ell}$, that is 
		\[\ell \geq s \cdot (\log 13 + \log k + \log \log q) + \log e_0.\]
		
		On the other hand, note that at any stage the additive error for any probability value calculated by \textsc{Fast-Decoder} is upper bounded by $\eps_s$. Thus, if we pick
		\[\ell \geq s\cdot \parens{\log{k}+\log\log{q}+5}+2,\]
		then we have $\eps_s\le \frac 14$ (since $u=2^{-\ell}$). The proof is complete by noting that if we chose $\ell$ to be maximum of those two necessary lower bounds bounds, we have $\ell=\Oh_{k,q}(\log{n})$ and hence we indeed are working with a floating point RAM model.
	\end{proof}

\paragraph{Correctness of \textsc{Fast-Decoder}.} With the runtime analysis of \textsc{Fast-Decoder} out of the way, in the next lemma we show that \textsc{Fast-Decoder} is equivalent to the
\textsc{SC-Decoder} on the same input. For this lemma we assume that $[n]$ is equated with $[k]^s$ and elements of $[k]^s$ are enumerated in lex order by \textsc{SC-Decoder}. Also it would be useful to keep \cref{rem:brown-part-fast-decoder} and \cref{rem:brown-part-sc-decoder} in mind.

\begin{lemma}
	\label{lem:fast-equivalence}
	Let $\bvec{Z}$ be a product distribution
	(where each $\bvec{Z}_{\bvec{i}} \in \Delta(\F_q)$ is a distribution over
	$\F_q$) and let 
	$\bvec{\alpha} \in (\F_q \cup \{\bot\})^{[k]^s}$.  
For $\bvec{i} \in [k]^s$, let $\bvec{P}_{\bvec{i}}$ be the quantity defined on Line~\ref{line:scd-x} of \textsc{SC-Decoder} for input $(\bvec{Z}; \bvec{\alpha})$, and let $\bvec{Q}_{\bvec{i}}$ be from the output of 
	$\textsc{Fast-Decoder}(\bvec{Z}; \bvec{\alpha},s)$. Then we have for every $\bvec{i} \in [k]^s$, $\bvec{P}_{\bvec{i}} = \bvec{Q}_{\bvec{i}}$ and 
	$$
	\textsc{Fast-Decoder}(\bvec{Z}; \bvec{\alpha}) \cdot M^{\otimes s}
	=
	\textsc{SC-Decoder}(\bvec{Z}; \bvec{\alpha}).
	$$
	Furthermore, the output of the precision-bounded \textsc{Fast-Decoder} equals the output of an approximate-successive-cancellation decoder on $(\bvec{Z}; \bvec{\alpha})$.
	
\end{lemma}

\begin{proof}
	We prove the lemma by induction on $s$. For $s=0$ the lemma is immediate (from line~\ref{line:scd-x} in \textsc{SC-Decoder} and line~\ref{line:base-case} in \textsc{Fast-Decoder}), so assume the lemma holds for $s'<s$.
	
	Our proof will compare two sets of variables, $\UF$ from the definition of \textsc{Fast-Decoder} and $\USC$ which we define next. 
	Given $\bvec Z, \bvec{\alpha}$ as in the statement of the lemma, 
	let $\bvec U$ be the joint distribution defined by 
	\[\bvec U := \bvec Z M^{\otimes s}.\]
	Now define $\USC$ such that for all $\bvec i \in [k]^s$:
	\begin{equation}
		\label{eqn:u-est}
		\USC_\bvec{i} =
		\begin{cases}
			\argmax_{x \in \F_q}\Pr\parens{\bvec{U}_\bvec{i} = x |
				\bvec{U}_{\prec \bvec{i}}
				= \USC_{\prec \bvec{i}}}  = \argmax_{x \in \F_q} \bvec X_{\bvec{i}}(x) & \text{if $\bvec{\alpha}_{\bvec{i}} =\bot$}\\
			\bvec{\alpha}_\bvec{i} & \text{if $\bvec{\alpha}_{\bvec{i}}\in \F_q$}.
		\end{cases}
	\end{equation}
	We start by noting that $\USC =\textsc{SC-Decoder}(\bvec{Z}; \bvec{\alpha})$ (this can be argued e.g. by induction on $\bvec{i}$). If $\bvec{\alpha}_i \in \F_q$, then it is easy to check that $\UF_{\bvec{i}}=\USC_\bvec{i}$, so for the rest of the proof we will assume this as given and the focus will be on indices $\bvec{i}$ such that $\bvec{\alpha}_{\bvec{i}} =\bot$.
	Next we note that the outputs $\bvec{\hat{Z}}$ and $\UF$ of \textsc{Fast-Decoder} are related by the condition $\bvec{\hat{Z}} = \textsc{Fast-Polar-Encoder}(\UF)$. 
	(In particular Lines~\ref{line:base-case},~\ref{line:v}~and~\ref{line:setZ} correspond exactly to the code of \textsc{Fast-Polar-Encoder}.) 
	Restated this implies 
	\begin{equation}
		\label{eqn:vhat-prelim}
		\bvec{\hat Z} \cdot M^{\otimes s}
		= \UF.
	\end{equation}

	Thus to prove the lemma, it suffices to prove that $\UF = \USC$. To do so we use the recursive structure of \textsc{Fast-Decoder} and prove that for every 
	$j \in [k]$, $\UF_{[j,\cdot]} = \USC_{[j,\cdot]}$. We do so by induction on $j$. 
	
	First recall that $\UF_{[j,\cdot]} = \textsc{Fast-Decoder$_{s-1}$}(\Zpj; \bvec{\alpha}_{[j, \cdot]})$
	with 
	\[\Zpj = \{ (\bvec{Z}\cdot M)_{[j,\cdot]} | (\bvec{Z}\cdot M)_{[<j,\cdot]} = \bvec{\hat V}_{[<j,\cdot]}\},\] 
	where the equality follows from Lines~\ref{line:update-W_i} and~\ref{line:updateW}.
	To compare with
	$\USC_{[j,\cdot]}$ we need a inductive structure on $\USC$ and we use a simple property that we describe informally first and then
	describe in formal notation. Informally, if the input stream to the successive cancellation decoder is split into three parts, the prefix
	$A$, the central part $B$ and the suffix $C$, then the decoding on the central part is independent of the suffix. Furthermore the decoding of the central
	part is the output of the successive cancellation decoder on a modified input which incorporates the conditioning induced by the decoding of the prefix.
	Formally the above can be expressed as the following: 
	Let $A \in (\Delta(\F_q))^a$, $B \in \Delta(\F_q)^b$ and $C \in \Delta(\F_q)^c$, and $\alpha \in (\F_q \cup \{\bot\})^a$,
	$\beta \in (\F_q \cup \{\bot\})^b$ and $\gamma \in (\F_q \cup \{\bot\})^c$. Then if $\hat{A} = \textsc{SC-Decoder}(A,\alpha)$ we have
	$\textsc{SC-Decoder}(A\circ B \circ C, \alpha \circ \beta \circ \gamma)_{[a+1,a+b]} = \textsc{SC-Decoder}(\tilde{B}, \beta)$ where $\tilde{B_i} = \{B_{a+i} | A = \hat{A} \}$ for  $i \in [b]$. 
	Applied in our context with $A = \bvec{U}_{[<j,\cdot]}$ and $B = \bvec{U}_{[j,\cdot]}$ we get 
	$\USC_{[j,\cdot]} = \textsc{SC-Decoder}(\bvec{\tilde{U}}^{(j)},\bvec{\alpha}_{[j,\cdot]})$ where $\bvec{\tilde{U}}^{(j)} = \{\bvec{U}_{[j,\cdot]} | \bvec{U}_{[<j,\cdot]} = \USC_{[<j,\cdot]}\}$ plays the role of $\tilde{B}$. We now use induction to show that the resulting sequences $\UF_{[j,\cdot]}$ and $\USC_{[j,\cdot]}$  are the same.
	
	By the (outer) inductive hypothesis (on $s$), it suffices to show that $\Zpj\cdot M^{\otimes s-1}$ is distributed identically\footnote{Technically, we want $\Zpj$ to be identically distributed to $\bvec{\tilde{U}}^{(j)}$ but this condition is equivalent since $M^{\otimes s-1}$ has full rank.} to $\bvec{\tilde{U}}^{(j)}$. We now simplify the former. We have 
	$$\Zpj\cdot M^{\otimes s-1} = \{ (\bvec{Z}\cdot M^{\otimes s})_{[j,\cdot]} | (\bvec{Z}\cdot M)_{[<j,\cdot]} = \bvec{\hat V}_{[<j,\cdot]}\} = \{ \bvec{U}_{[j,\cdot]} | (\bvec{Z}\cdot M)_{[<j,\cdot]} = \bvec{\hat V}_{[<j,\cdot]}\},$$
	where the first equality uses the fact that $(\bvec{Z}\cdot M)_{[j,\cdot]}\cdot M^{\otimes s-1}=(\bvec{Z}\cdot M^{\otimes s})_{[j,\cdot]}$.

	Comparing with the definition of $\bvec{\tilde{U}}^{(j)} = \{\bvec{U}_{[j,\cdot]} | \bvec{U}_{[<j,\cdot]} = \USC_{[<j,\cdot]}\}$, it thus suffices to show that the conditioning events 
	$(\bvec{Z}\cdot M)_{[<j,\cdot]} = \bvec{\hat V}_{[<j,\cdot]}$
	and $\bvec{U}_{[<j,\cdot]} = \USC_{[<j,\cdot]}$ are identical. 
	For every $\ell < j$, we have, by applying \cref{eqn:vhat-prelim} to the outputs of $\textsc{Fast-Decoder}(\Zpj,\bvec{\alpha}_{[\ell,\cdot]},s-1)$ in Line~\ref{line:v}, we have $\bvec{\hat{V}}_{[\ell,\cdot]}\cdot M^{\otimes s-1} = \UF_{[\ell,\cdot]}$. Now using (inner) inductive hypothesis on $\ell < j$ we have $\bvec{\hat{V}}_{[\ell,\cdot]}\cdot M^{\otimes s-1} = \USC_{[\ell,\cdot]}$. We use this and the invertibility of $M^{\otimes s-1}$ to rephrase the event $(\bvec{Z}\cdot M)_{[<j,\cdot]} = \bvec{\hat V}_{[<j,\cdot]}$
	as 
	$(\bvec{Z}\cdot M)_{[<j,\cdot]}\cdot M^{\otimes s-1} =  \bvec{\hat V}_{[<j,\cdot]} \cdot M^{\otimes s-1} = \USC_{[<j,\cdot]}$. 
	Simplifying the left hand side we get $(\bvec{Z}\cdot M)_{[<j,\cdot]}\cdot M^{\otimes s-1} = (\bvec{Z}\cdot M^{\otimes s})_{[<j,\cdot]} = \bvec{U}_{[<j,\cdot]}$.
	Thus we get that the two events are indeed identical, and thus yield $\UF_{[j,\cdot]} = \USC_{[j,\cdot]}$.
	
	The proof that $\bvec{P}_{[j,\cdot]} = \bvec{Q}_{[j,\cdot]}$ for every $j \in [k]$ is completely similar and we omit the details. For the furthermore part, note that an equivalent view of \textsc{Fast-Decoder} is that it is an efficient algorithm to compute the $\bvec{Q}_{\bvec{i}}$'s which it then uses to run \textsc{SC-Decoder}. Thus if a bounded-precision \textsc{Fast-Decoder} computes every entry of $\bvec{P}_{\bvec{i}}$ to within an additive error of $1/4$ then the bounded precision \textsc{Fast-Decoder} implements an approximate-successive-cancellation decoder.
\end{proof}

\paragraph{Proofs of \cref{thm:main-quant} and \cref{thm:poly-code}.}
Now we can prove \cref{thm:main-quant} (modulo \cref{clm:exist-S}, which we prove in the next sub-section).
\begin{proof}[Proof of \cref{thm:main-quant}]
	In the model of infinite precision arithmetic, \cref{thm:main-quant} follows from 
	\cref{thm:sc-decoder-quant} and the equivalence of \textsc{SC-Decoder} and \textsc{Fast-Decoder} from \cref{lem:fast-equivalence}
	with the running time bound following from \cref{lem:nlogn}. 
	
	In the bounded precision case, by Lemma~\ref{lem:fast-equivalence} we have that the bounded-precision \textsc{Fast-Decoder} implements an approximate-successive-cancellation decoder. Applying \cref{thm:sc-decoder-quant} again in this setting we have that the decoding error probability still remains $\Oh(n \tau \log q)$, and the running time of $\Oh(n \log n)$ from \cref{lem:nlogn} is now in the standard floating point RAM model.
\end{proof}

Finally, \cref{thm:poly-code} is essentially a corollary of \cref{thm:main-quant} and the definition of (exponential) strong polarization.
\begin{proof}[Proof of \cref{thm:poly-code}]
	Fix some constant $c$, and take $\gamma < k^{-c - 1} \log^{-1} q$, with $n=k^t$. Note that this implies that
	\begin{equation}
	\label{eq:gamma-t-bound}
	    \gamma^t = \frac{1}{\parens{k^t}^{c+1}\cdot \log^t{q}}= \frac{1}{\parens{k^t}^{c+1}\cdot \log^t{q}}.
	\end{equation}

	By the definition of strong polarization property, we know that for some constants $\beta, \eta$, martingale $X_t$ is $(\gamma^{t}, \gamma^t, \beta \cdot \eta^t)$-polarizing. Hence, by \cref{thm:main-quant} corresponding polar code has rate at least
	\begin{equation*}
		\text{Capacity}(\cC) - \beta \eta^t - \two\gamma^t 
	\end{equation*}
	for $t = \Theta_{\eta, \beta}(\log(1/\varepsilon))$, we have $\beta \eta^t + \gamma^t \leq \varepsilon$, where the inequality follow from \cref{eq:gamma-t-bound} and our choice of $t$.
	
	The probability of decoding failure is at most
	\begin{equation*}
		n \gamma^t \log q \leq n (n)^{-c-1} \log^{-t + 1}(q) \leq n^{-c},
	\end{equation*}
	where the first inequality follows from \cref{eq:gamma-t-bound}.
	
		By the definition of strong polarization property, we know that for some constants $\beta, \eta,\Lambda$, martingale $X_t$ is $(2^{-2^{\Lambda t}}, \gamma^t, \beta \cdot \eta^t)$-polarizing. We use the same choice of $t$ as in the strong polarizing case and using the same argument as in that case we get that the polar code has the claimed rate. The probability of decoding error is at most
		\[n\log{q}\cdot 2^{-2^{\Lambda t}} = n\log{q}\cdot 2^{-2^{\Lambda \frac{\log{n}}{\log{k}}}} = n\log{q}\cdot 2^{- n^\frac{\Lambda}{\log{k}}} \le 2^{-n^{\beta'}},\]
		for some $\beta'=\Omega_{\Lambda,k,q}\parens{1}$, as desired.
	
\end{proof}

\subsubsection{\Arikan~Martingale and Polar Coding}
\label{sec:correspondence}

Here we build a correspondence between the definition of the \Arikan~Martingale and the process of polar coding, which was used in the proof of \cref{clm:exist-S}.

\begin{lemma}
	\label{lem:mart-code}
	For a matrix $M \in \F_q^{k \x k}$ and symmetric channel $\mathcal{C}_{Y | Z}$,
	let $\{X_t\}$ be the associated \Arikan~Martingale.
	For a given $t$,  let $L = M^{\otimes t}$ be the polarization transform,
	and let $n = k^t$ be the block length.
	Let the channel inputs $\bvec{Z}_i$ be i.i.d. uniform in $\F_q$,
	and channel outputs $\bvec{Y}_i \sim \mathcal{C}_{Y|Z}(\bvec{Z}_i)$.
	
	Then, for a uniformly random index $i \in [n]$,
	the normalized entropy $\bH((\bvec{Z} L)_i ~|~   (\bvec{Z} L)_{< i}, \bvec{Y})$
	is distributed identically as $X_t$.
\end{lemma}

\begin{proof}
		Throughout this proof, we will switch to
	considering vectors in $\F_q^{k^t}$ as tensors in
	$(\F_q^k)^{\otimes t}$, for convenience---this correspondence is induced by lexicographic ordering $\prec$ on tuples $[k]^t$.
	Also, we will write $H(\bvec Z)$ to mean
	the operator $H$ acting on $\bvec{Z}$. More specifically for a linear map defined by matrix $H$ we use $H(\bvec{Z})=\bvec{Z}H$.
	In this notation, we wish to show that the distribution of $X_t$ is identical to $\bH((M^{\otimes t}(\bvec{Z}))_{\bvec i} ~|~ \bvec{Y}, (M^{\otimes t}(\bvec{Z}))_{\prec \bvec{i}})$
	for a uniformly random multiindex $\bvec{i} \in [k]^t$.
	
	We will show by induction that for all $t$, there is some permutation
	of coordinates\footnote{This is in fact just a reversal of the co-ordinates, i.e. $\sigma'((i_1, i_2,
		\dots i_t)) = (i_t, \dots, i_2, i_1)$.} 
	$\sigma': [k]^t \to [k]^t$
	such that the joint distributions
	\begin{equation}
	\label{eq:arikan-is-polar}
	    \{(\bvec{A'}, \bvec{B'})\}_{(\bvec{A'}, \bvec{B'}) \sim D_t}
	\equiv
	\{(M^{\otimes t}(\bvec{Z}), \sigma'(\bvec{\mathcal{C}(Z)}))\}_{\bvec Z \sim (\F_q^k)^{\otimes t}},
	\end{equation}
	where $(\bvec{A'}, \bvec{B'}) \sim D_t$ are the distributions defined in the $t$-th step of the \Arikan~martingale, and
	$\bvec{Z} \sim (\F_q^k)^{\otimes t}$ is sampled with i.i.d. uniform coordinates.
	This is sufficient, because a permutation of the channel outputs does not affect the relevant entropies. That is,
	$$
	\bH(\bvec{A'}_{\bvec i} ~|~ 
	\bvec{A'}_{\prec\bvec{i}}, \bvec{B'})
	=
	\bH(\bvec{A'}_{\bvec i} ~|~ 
	\bvec{A'}_{\prec\bvec{i}}, \sigma'(\bvec{B'})).
	$$
	
	First, the base case $t=0$ follows by definition of the distribution $D_0$ in the \Arikan~martingale (and the fact that $M(\bvec{Z}_1)\sim\F_q$).
	
	For the inductive step, assume the claim holds for $t-1$.
	Let $\sigma$ be the permutation guaranteed for $t-1$.
	For each $j \in [k]$, sample an independent uniform
	$\bvec{Z}^{(j)} \sim (\F_q^k)^{\otimes t-1}$
	and define
	
	\begin{equation}
		\label{eqn:def-sampling}
		(\bvec{A}^{(j)}, \bvec{B}^{(j)})
		:=
		(  M^{\otimes t-1}(\bvec{Z}^{(j)}) ~,~ \sigma(\mathcal{C}(\bvec{Z}^{(j)}))).
	\end{equation}
	
	By the inductive hypothesis,
	$(\bvec{A}^{(j)}, \bvec{B}^{(j)}) \sim D_{t-1}$, for each $j \in [k]$.
	
	As in the \Arikan~martingale, define
	$(\bvec{A}', \bvec{B'})$
	deriving from $\{(\bvec{A}^{(j)}, \bvec{B}^{(j)})\}_{j \in [k]}$ as
	\begin{equation}
		\label{eqn:tensor-step}
		\bvec{A}'_{[\bvec{i}, \cdot]}
		:=
		M((
		\bvec{A}^{(1)}_{\bvec{i} }
		~
		,\dots,
		~
		\bvec{A}^{(k)}_{\bvec{i}}
		))
		\quad
		\text{and}
		\quad
		\bvec{B}'_{[j, \cdot]} := \bvec{B}^{(j)}.
	\end{equation}
	Note that $\bvec{B'}$ can equivalently be written (unwrapped) as
	$$\bvec{B}' := (\bvec{B}^{(1)}, \bvec{B}^{(2)}, \ldots, \bvec{B}^{(k)})$$
	By definition of the \Arikan~martingale, we have
	$(\bvec{A}', \bvec{B'}) \sim D_t$.
	
	Finally, define $\bvec{Z} \in (\F_q^k)^{\otimes t}$ by
	\begin{equation}
	\label{eq:lsb-first}
\bvec{Z}_{[\cdot, j]} := \bvec{Z}^{(j)}.
	\end{equation}
	To finish the proof, we will show that
	$(\bvec{A}', \bvec{B'}) =
	(M^{\otimes t}(\bvec{Z}) , \sigma'(\bvec{\mathcal{C}(Z)}))$
	for some permutation $\sigma'$.
	
	The main claim is the following. 
	\begin{claim}
		\label{claim-one}
		For every instantiation of the underlying randomness in
		$\bvec{Z}$, we have
		$$\bvec{A'} =  M^{\otimes t}(\bvec{Z}).$$
	\end{claim}
	\begin{proof}[Proof of Claim~\ref{claim-one}]
		
		Expanding the recursive definition of the tensor product, Equation~\eqref{eqn:tensor1}, we have:
		\begin{align*}
			[M^{\otimes t}(\bvec{Z})]_{[\bvec{i}, \cdot]}
			&=  M((
			\bvec{W}_{\bvec i}^{(1)},
			\bvec{W}_{\bvec i}^{(2)},
			\dots
			\bvec{W}_{\bvec i}^{(k)}
			))
		\end{align*}
		where
		$$
		\bvec{W}^{(j)}
		:= M^{\otimes t-1}(\bvec{Z}_{[\cdot, j]})
		= M^{\otimes t-1}(\bvec{Z}^{(j)})
		= \bvec{A}^{(j)}.
		$$
		Here the last equality 
		is by the inductive assumption.
		Thus,
		\begin{align*}
			[M^{\otimes t}(\bvec{Z})]_{[\bvec{i}, \cdot]}
			&=
			M((
			\bvec{A}^{(1)}_{\bvec{i} }
			~
			,\dots,
			~
			\bvec{A}^{(k)}_{\bvec{i}}
			))\\
			&= \bvec{A'}_{[\bvec{i}, \cdot]}. \tag{By definition, given in \cref{eqn:tensor-step}}
		\end{align*}
		And so
		$M^{\otimes t}(\bvec{Z})=
		\bvec{A'}$ as desired.
	\end{proof}
	
	Continuing the proof of \cref{lem:mart-code},
	we now have
	\begin{align*}
		(\bvec{A'} ~,~ \bvec{B'})
		&= (\bvec{A'} ~,~
		(
		\bvec{B}^{(1)},
		\bvec{B}^{(2)},
		\dots,
		\bvec{B}^{(k)})) \\ 
		(
		\sigma(\mathcal{C}(\bvec{Z}^{(1)})),
		\sigma(\mathcal{C}(\bvec{Z}^{(2)})),
		\dots,
		\sigma(\mathcal{C}(\bvec{Z}^{(k)}))
		)
		\tag{Definition of sampling, \cref{eqn:def-sampling}}\\
		&= (\bvec{A'} ~,~ \sigma'(\mathcal{C}(\bvec{Z}))) \tag{$\star$}\\
		&= (M^{\otimes t}(\bvec{Z}) ~,~ \sigma'(\mathcal{C}(\bvec{Z}))) \tag{\cref{claim-one}}.
	\end{align*}
	In the above, the equality in line $(\star)$ follows by taking $\sigma'$ to be the permutation that sorts $[k]^t$ in the order of \emph{least significant symbol} first (based on our definition in \eqref{eq:lsb-first}), and then sorts each group (thought of as $[k]^{t-1}$ in the natural way) recursively according to $\sigma$. Unwinding this recursion, one can see that $\sigma'$ is in fact the \emph{symbol-reversal} permutation on $[k]^t$. 
	
	This establishes the equivalence of the distributions claimed in \eqref{eq:arikan-is-polar} and completes the proof.
\end{proof}

\end{document}